\documentclass[a4paper]{article}

\usepackage{amsmath}
\usepackage{amssymb}
\usepackage{amsthm}
\usepackage{mathrsfs}
\usepackage{wasysym}
\usepackage{bbm}		%
\usepackage[inline]{enumitem}
\usepackage{nccmath}	%
\usepackage{mathtools}
\usepackage{braket}		%
\usepackage[normalem]{ulem}		%
\usepackage{nicefrac}
\usepackage{cancel}
\usepackage{xcolor}
\usepackage[a4paper, left=2.7cm, right=2.7cm, top=3.5cm, bottom=2cm]{geometry}
\usepackage[pdftex=true,hyperindex=true,pdfborder={0 0 0}]{hyperref}
\usepackage[nottoc]{tocbibind}	%
\usepackage{doi}
\usepackage[numbers]{natbib}
\usepackage{cleveref}
\usepackage{tikz}
\usetikzlibrary{shapes,arrows,automata,matrix,fit,patterns}
\usepackage{soul} %

\usepackage{autonum}		
\theoremstyle{plain}
\newtheorem{theorem}{Theorem}[section]
\newtheorem{lemma}[theorem]{Lemma}
\newtheorem{corollary}[theorem]{Corollary}
\newtheorem{proposition}[theorem]{Proposition}

\theoremstyle{definition}
\newtheorem{definition}[theorem]{Definition}
\newtheorem{remark}[theorem]{Remark}
\newtheorem{example}[theorem]{Example}
\newtheorem{conditions}[theorem]{Conditions}
\newtheorem{claim}{Claim}

\newcommand{\exampleqed}{\ensuremath{\ocircle}\par}
\newcommand{\remarkqed}{\ensuremath{\Diamond}\par}
\newcommand{\ZZ}{\mathbb{Z}}			%
\newcommand{\NN}{\mathbb{N}}			%
\newcommand{\RR}{\mathbb{R}}			%
\newcommand{\Sites}{\mathcal{S}}			%
\newcommand{\Edges}{\mathcal{E}}			%
\newcommand{\Lang}{\mathcal{L}}			%
\newcommand{\symb}[1]{\mathtt{#1}}		%
\newcommand{\isdef}{\triangleq}			%
\DeclarePairedDelimiter\abs{\lvert}{\rvert}		%
\DeclarePairedDelimiter\norm{\lVert}{\rVert}	%
\makeatletter
\newcommand{\vast}{\bBigg@{4}}
\newcommand{\Vast}{\bBigg@{5}}

\makeatother
\newcommand{\Var}{%
	\operatorname{\mathrm{Var}}%
}	
\newcommand{\oo}{\circ}					%
\renewcommand{\complement}{%
	\mathsf{c}%
}
\newcommand{\supp}{%
	\operatorname{\mathrm{supp}}%
}

\newcommand{\dd}{\mathrm{d}}			%
\newcommand{\ee}{\mathrm{e}}			%
\newcommand{\smallo}{o}					%
\newcommand{\xPr}{\operatorname{\mathbb{P}}}				%
\newcommand{\xExp}{\operatorname{\mathbb{E}}}				%
\newcommand{\rv}[1]{\mathbf{#1}}		%
\newcommand{\indicator}[1]{\mathbbm{1}_{#1}}					%

\newcommand{\measures}[1]{\mathscr{#1}}		%
\newcommand{\banach}[1]{\mathscr{#1}}		%
\newcommand{\field}[1]{\mathscr{#1}}		%
\newcommand{\relation}[1]{\mathcal{#1}}		%
\newcommand{\sgn}{\operatorname{sign}}

\newcommand{\blank}{\symb{\diamond}}		%
\newcommand{\simI}{%
	\overset{\raisebox{-2pt}{$\scriptscriptstyle 1$}}{\sim}%
}
\newcommand{\simII}{%
	\overset{\raisebox{-2pt}{$\scriptscriptstyle 2$}}{\sim}%
}

\newcommand{\VS}{\mathsf{VS}}				%
\newcommand{\NS}{\mathsf{NS}}				%
\newcommand{\Sull}{\mathsf{Sull}}			%

\newcommand{\normVS}[2][]{\norm[#1]{#2}_{\VS}}

\newcommand{\normsull}[2][]{\norm[#1]{#2}_{\Sull}}

\newcommand{\normNS}[2][]{\norm[#1]{#2}_{\NS}}
\newcommand{\normNSD}[2][]{\norm[#1]{#2}^*_{\NS}}

\newcommand{\Ham}{\operatorname{Ham}}

\newcommand{\hardcore}{\mathsf{hc}}		%
\newcommand{\coloring}{\mathsf{col}}	%
\newcommand{\finiterange}{\mathsf{f}}	%

\newcommand{\ball}[2]{ B_{#1}(#2)}

\newcommand{\xConfig}[1]{%
	\begin{tikzpicture}[
		baseline=-\the\dimexpr\fontdimen22\textfont2\relax,ampersand replacement=\&]
		\matrix[
			matrix of math nodes,
			nodes={
				minimum size=1.4ex,text width=1.4ex,
				text height=1.4ex,inner sep=3pt,draw={gray!20},anchor=center
			}, row sep=1pt,column sep=1pt
		] (config) {#1};
		\node[draw,rectangle,help lines,gray!50, dashed,fit=(config),inner sep=-1pt] {};
	\end{tikzpicture}
}

\title{%
	Gibbsian representations of continuous specifications:\\
	the theorems of Kozlov and Sullivan revisited
}

\author{
	Sebasti\'an Barbieri\footnote{%
		Departamento de Matem\'atica y Ciencia de la Computaci\'on, Universidad de Santiago de Chile.	
	}
	\and
	Ricardo G\'omez\footnote{%
		Institute of Mathematics, National Autonomous University of Mexico, Mexico
	}
	\and
	Brian Marcus\footnote{
		Department of Mathematics, University of British Columbia, Canada	
	}
	\and
	Tom Meyerovitch\footnote{%
		Department of Mathematics, Ben-Gurion University of the Negev, Israel
	}
	\and
	Siamak Taati\footnote{%
		Bernoulli Institute, University of Groningen, The Netherlands
	}
}

\begin{document}

\maketitle

\begin{abstract}
	The theorems of Kozlov and Sullivan characterize Gibbs measures as measures
	with positive continuous specifications.
	More precisely, Kozlov showed that every positive continuous specification
	on symbolic configurations of the lattice is generated by a norm-summable interaction.
	Sullivan showed that every shift-invariant positive continuous specification
	is generated by a shift-invariant interaction satisfying
	the weaker condition of variation-summability.
	These results were proven in the 1970s.  An open question since that
	time is whether Kozlov's theorem holds in the shift-invariant setting,
	equivalently whether Sullivan's conclusion can be improved from
	variation-summability to norm-summability.  We show that the answer is no:
	there exist shift-invariant positive continuous specifications that
	are not generated by any shift-invariant norm-summable interaction.
	On the other hand, we give a complete proof of an extension, suggested by Kozlov, of Kozlov's theorem to a characterization
	of positive continuous specifications on configuration spaces with arbitrary hard constraints. 	
	We also present an extended version of Sullivan's theorem.
	Aside from simplifying some of the arguments in the original proof,
	our new version of Sullivan's theorem applies in various
	settings not covered by the original proof.
	In particular, it applies when the support of the specification is
	the hard-core shift or the two-dimensional $q$-coloring shift for~$q\geq 6$.

	\medskip


	\smallskip
	
	
	\renewcommand{\contentsname}{\vspace{-1.5em}}
	{\footnotesize\tableofcontents}
\end{abstract}

\section{Introduction}
\label{sec:intro}

A \emph{(nearest-neighbour) Markov random field} on the $d$-dimensional integer lattice $\ZZ^d$ is a probability measure  $\mu$ on a space
$\Omega\subseteq\Sigma^{\ZZ^d}$ of configurations of symbols from a finite alphabet $\Sigma$ on $\ZZ^d$ that
satisfies the following Markovian condition:  for every finite subset $\Lambda$ of sites,
the probability distribution on events of $\Lambda$, conditioned on the complement of $\Lambda$ is a function
of only the restriction of the configuration on the external boundary of $\Lambda$.
In the case where all configurations are allowed, i.e., $\Omega = \Sigma^{\ZZ^d}$ and $\mu$ is fully supported, it follows from
the Hammersley--Clifford Theorem~\cite{HamCli68, Spi71, Ave72} that the random field can be expressed
in a Gibbsian form by a nearest-neighbour interaction.

An \emph{almost-Markovian random field} on $\ZZ^d$ is defined by the following weaker condition:
for every finite subset $\Lambda$ of sites, the probability distribution on events of $\Lambda$,
conditioned on the complement of $\Lambda$ is merely continuous as a function of the restriction of the configuration to the complement of $\Lambda$.
Assuming again that $\Omega = \Sigma^{\ZZ^d}$ and $\mu$ is fully supported, can the random field be given
in a Gibbsian form by an interaction, and in particular by a highly regular interaction?
More than forty years ago, first Sullivan~\cite[Theorem 1]{Sul73}
and then Kozlov~\cite[Theorem~1]{Koz74} answered this question positively,
with Kozlov obtaining a stronger result: namely, the random field can be expressed in a Gibbsian form by
a so-called norm-summable (also called absolutely-summable) interaction.  This is a remarkable result:
mere continuity is sufficient to represent the random field in a Gibbsian form with an interaction of
very high regularity.

These results have more to do with the so-called \emph{specification} of a
random field rather than the joint distributions of the random field itself:  the specification is
the collection of all conditional probabilities $K(x|y)$ on events on a finite set of lattice sites
conditioned on the  complement.
This viewpoint  goes back to Dobrushin~\cite{Dob68c}, who referred to specifications as
``consistent systems of conditional probabilities.''
It is the specification that is represented in Gibbsian form
\begin{align}
\label{Gibbsian}
 	K(x|y) &= \frac{1}{Z}\exp\big(-E_\Phi(x \lor y)\big)
\end{align}
where $E_\Phi$ is the energy function corresponding to an interaction $\Phi$, $x$ is a configuration on
a finite set of sites $A$, $y$ is a configuration on $A^c$, and $Z$ is a normalizing factor.  Our main interest is in finding conditions on the conditional probabilities that guarantee a representation (\ref{Gibbsian}) with an interaction
 $\Phi$ of a high degree of regularity.

Now, consider a \emph{shift-invariant} almost-Markovian random field $\mu$, again assuming $\Omega \isdef \Sigma^{\ZZ^d}$  and $\mu$ is fully supported.
It is natural to ask if, in this case, the norm-summable interaction
given by Kozlov's theorem can always be chosen to be shift-invariant.
Both Kozlov and Sullivan addressed this problem in the 1970's,  but obtained weaker results.
Kozlov~\cite[Theorem~3]{Koz74} showed that if one imposes positivity and a stronger form of continuity,
in terms of decay of modulus of continuity, on a shift-invariant specification,
then in fact one can obtain a Gibbsian representation with a shift-invariant norm-summable interaction.
On the other hand, Sullivan~\cite[Corollary to Theorem 2]{Sul73} showed that mere positivity and continuity of
a shift-invariant specification is sufficient to obtain a Gibbsian representation with a
shift-invariant interaction
that satisfies a weaker form of regularity, which he called absolute convergence.
In fact, Sullivan's interaction satisfies what we call \emph{variation-summability},
which is somewhere in between absolute convergence and norm-summability.

The question of whether every positive shift-invariant almost-Markovian specification can be represented
by a shift-invariant norm-summable interaction has been considered by several authors,
including Gross~\cite{Gro82} (see the comments on page~195), van Enter, Fern\'andez, and Sokal~\cite{EntFerSok93}
(see the remark at the end of Section~2.4.9), and Fern\'andez~\cite{Fer06} (Section~4.3.3),
but was left open up until now.
The main result of our paper, Theorem~\ref{thm:Kozlov_is_annoying_INTRO}, %
gives a negative answer to this question. The proof is somewhat nonconstructive, works already in dimension $d = 1$ and uses elementary ideas
from Banach space theory, combinatorics and probabilistic inequalities.

\begin{theorem}
	\label{thm:Kozlov_is_annoying_INTRO}
	There exists a shift-invariant, positive, almost-Markovian specification on $\Omega \isdef \{\symb{0},\symb{1}\}^{\ZZ}$ which does not admit a Gibbsian representation in terms of a shift-invariant norm-summable interaction.
\end{theorem}
In other words, there exists a shift-invariant, full-support, almost-Markovian measure
on $\Omega \isdef \{\symb{0},\symb{1}\}^{\ZZ}$
which is not a Gibbs measure for any shift-invariant norm-summable interaction.

In our proof, we use a sequence of continuous cocycles that we obtain through a probabilistic argument. The first usage of cocycles (also called relative Hamiltonians in this context) to parameterize specifications is attributed by Gross in his lecture notes~\cite{Gro82} to Pirogov and Sinai~\cite{PirSin75}. This idea has also been used by Petersen and Schmidt~\cite{PetSch97} to study certain classes of Gibbs measures in subshifts of finite type. Recent examples of applications of this formalism can be found in~\cite{ChaMey16} or~\cite{DacNah18}. The idea to use the probabilistic method to obtain objects with interesting properties in the context of thermodynamic formalism also has precedents in the literature. See for instance Israel~\cite{Isr79} (pages~29--30) for a proof in the same spirit in a similar context.

One might ask why norm-summability is the most desirable level of regularity, other than the fact that it
has become the standard form of regularity for interactions in Gibbs theory.
One answer is that it is needed for the DLR theory
(equivalence of shift-invariant Gibbs states and equilibrium states, see~\cite{LanRue69,Dob68b}) to work.
However, variation-summability seems to be sufficient
for the ``LR'' part of the theory (equilibrium implies Gibbs) but we do not
know if it is sufficient for the ``D'' part of the theory (Gibbs implies equilibrium).

Section~\ref{sec:prelim} of this paper contains preliminary definitions and results.
In Section~\ref{sec:prelim:specification}, the notion of a specification is defined in its own right,
without mention of an underlying random field, by a simple set of axioms.
Each positive specification is associated with an abstract notion of ``energy'' expressed
by a \emph{cocycle} on the asymptotic relation.  The cocycle assigns to each pair of asymptotic configurations
a real number interpreted as the energy difference between the two configurations.
Cocycles are equivalent representations of positive specifications.
The notion of cocycles is introduced in Section~\ref{sec:prelim:cocycles}, and their correspondence
with positive specifications is discussed in Section~\ref{sec:prelim:specification-cocycle:equivalence}.
For the remainder of the paper, we work in the framework of cocyles rather than specifications.
The classes of interactions that we use in this paper are defined in
Sections~\ref{sec:prelim:interactions:VS} and~\ref{sec:prelim:interactions:NS}.
The Banach space theory that we need is based on the open mapping theorem and is described
in Sections~\ref{sec:prelim:Banach_cocycles} and~\ref{sec:prelim:surjective_Banach}.

The proof of Theorem~\ref{thm:Kozlov_is_annoying_INTRO} is given in Section~\ref{sec:main:non_surjectivity}. It uses a criterion for a linear operator to be non-surjective, Proposition~\ref{prop:Banach_equivalences}. In Sections~\ref{sec:main:3_colored_chessboard} and~\ref{sec:main:square_islands},
we also give two explicit examples of shift-invariant specifications which fail to be represented by a shift-invariant norm-summable interaction,
but in the case where $\Omega$ is a proper subset of  $\Sigma^{\ZZ^d}$, i.e.,
certain configurations may be forbidden. These are two-dimensional examples,
taken from Chandgotia and Meyerovitch~\cite{ChaMey16},
that violate a linear growth condition~\eqref{eq:linear_growth_condition}
necessary for the existence of a shift-invariant norm-summable interaction.
However, in Proposition~\ref{prop:full-shift:linear-growth-is-automatic} we show that this approach
cannot possibly work in the case~$\Omega = \Sigma^{\ZZ^d}$.

In Section~\ref{sec:main:Kozlov}, we give a complete  proof of an extension of Kozlov's theorem to the case
where some configurations may be forbidden, in particular, our result holds for any compact subset $\Omega$ of
$\Sigma^{\ZZ^d}$. This extension was already suggested by Kozlov (see the paragraphs preceding Theorem~3
in~\cite{Koz77}), who did not give a proof.
In our proof, we  first show that Markovian specifications can be represented by finite range interactions and then we proceed by approximating almost-Markovian specifications by Markovian specifications. Our proof does not rely on any aspect of $\ZZ^d$ besides the fact that it is a countable set, as again suggested by Kozlov~\cite{Koz77}.
\begin{theorem}\label{thm:Kozlov2_INTRO}
	Let $\Omega$ be a symbolic configuration space over a countable set of sites $\Sites$. Every positive almost-Markovian specification on $\Omega$ admits a Gibbsian representation in terms of a norm-summable interaction.
\end{theorem}

In Section~\ref{sec:sullivan}, we generalize Sullivan's theorem to configuration spaces $\Omega$
that satisfy a certain kind of constraint on allowed configurations.
Moreover, our proof employs a simplification in the construction of Sullivan's interaction.
In particular, this allows our result to be extended to other lattices and the class of symbolic actions
of countable amenable groups.

\begin{theorem}\label{thm:sull2_INTRO}
	Let $\Omega \subseteq \Sigma^{\ZZ^d}$ be a shift of finite type which is single-site fillable and has the pivot property.
	Every shift-invariant, positive almost-Markovian specification on $\Omega$ admits a Gibbsian representation in terms of a shift-invariant variation-summable interaction.
\end{theorem}

Theorems~\ref{thm:Kozlov_is_annoying_INTRO},
\ref{thm:Kozlov2_INTRO}, and~\ref{thm:sull2_INTRO} are stated in the equivalent language of continuous cocycles in Theorems~\ref{thm:Kozlov_is_annoying},~\ref{thm:gibbs-cocycle:interaction:non-invariant}, and~\ref{thm:sull2} respectively.  See Section~\ref{sec:prelim:specification-cocycle:equivalence} for the correspondence between specifications and cocycles.

Sections~\ref{sec:main}, \ref{sec:main:Kozlov}, and~\ref{sec:sullivan} are independent of one another and so can be read separately once the reader has read the preliminary Section~\ref{sec:prelim}.

\paragraph{Acknowledgements.}

  The authors thank Nishant Chandgotia for helpful discussions.  
  Sebasti\'an Barbieri and Tom Meyerovitch thank the Pacific Institute for Mathematical Sciences (PIMS) and the mathematics department of the University of British Columbia (UBC) where most of this work was done while hosted as a postdoctoral fellow and PIMS distinguished visitor, respectively. Brian Marcus thanks Aernout van Enter for introducing him to Kozlov's theorem.  
  
  Sebasti\'an Barbieri was partially supported by the ANR project CoCoGro (ANR-16-CE40-0005) and the ANR project CODYS (ANR-18-CE40-0007). Ricardo G\'omez was partially supported by DGAPA-PAPIIT grant IN107718. Brian Marcus was partially supported by NSERC grant RGPIN-2017-04550. Tom Meyerovitch was partially supported by the Israeli Science Foundation (ISF grant 1052/18). Siamak Taati was partially supported by NWO grant~612.001.409.

\section{Preliminaries}
\label{sec:prelim}

\subsection{Spaces of configurations}
\label{sec:prelim:spaces_config}

The focus of this paper is on models in which the state of a physical system
is represented by a configuration of symbols.
A \emph{(symbolic) configuration} on a countable set $\Sites$ is an array
$x\isdef(x_k)_{k\in\Sites}$ of symbols from a finite alphabet $\Sigma$
indexed by the elements of $\Sites$.
We refer to the elements of $\Sites$ as \emph{sites}.
A \emph{pattern} (or a \emph{partial configuration}) is an array $w\in\Sigma^A$ where $A\subseteq\Sites$.
We call $A$ the \emph{shape} of pattern $w\in\Sigma^A$.
A \emph{finite} pattern is a pattern whose shape is finite.
The restriction of a pattern $p\in\Sigma^A$ to a shape $B\subseteq A$
is denoted by $p_B$.  Consistently with this notation, we sometimes
denote a pattern $p\in\Sigma^A$ by $p_A$ to emphasize its shape.
Given two patterns $u\in\Sigma^A$ and $v\in\Sigma^B$ satisfying $u_{A\cap B}=v_{A\cap B}$,
we write $u\lor v$ for the pattern with shape $A\cup B$,
for which $(u\lor v)_A=u$ and $(u\lor v)_B=v$.
We will write $A\Subset B$ to indicate that $A$ is a finite subset of $B$.

The set $\Sigma^\Sites$ of all symbolic configurations on $\Sites$
with symbols from~$\Sigma$ is endowed with the product topology,
which is compact and metrizable.
A \emph{cylinder set} in $\Sigma^\Sites$ is a set of the form
$[w_A]\isdef\{x\in\Sigma^\Sites: x_A=w_A\}$ where $w_A\in\Sigma^A$ is a finite
pattern with shape $A\Subset\Sites$.
The set $A$ is called the \emph{base} of the cylinder set $[w_A]$.
The cylinder sets are both open and closed and form a basis for the product topology
on $\Sigma^\Sites$.

By a \emph{(symbolic) configuration space} we shall mean a non-empty compact set
$\Omega\subseteq\Sigma^\Sites$ for some countably infinite set of sites $\Sites$ and finite~$\Sigma$.
The configuration space $\Sigma^\Sites$ is said to be \emph{full}.
We refer to the elements of $\Omega$ as the configurations of $\Sigma^\Sites$ that are
\emph{admissible} for $\Omega$.

For most of the current paper, $\Sites$ will be the $d$-dimensional square \emph{lattice} $\ZZ^d$ for some $d\in\ZZ^+$,
and we will be interested in configuration spaces on $\ZZ^d$ which respect the translation symmetry.
We denote by $\sigma$ the action of $\ZZ^d$ on the full configuration space
$\Sigma^{\ZZ^d}$ by shifts.  More specifically, $\sigma^k x$ denotes the \emph{translation} (or \emph{shift})
of a configuration $x$ by $k\in\ZZ^d$, that is, $(\sigma^k x)_i\isdef x_{i+k}$ for $i\in\ZZ^d$.
A \emph{shift space} (or \emph{subshift}) is a configuration space $\Omega\subseteq\Sigma^{\ZZ^d}$
that is shift-invariant, meaning that $\sigma^k x\in\Omega$ for each $x\in\Omega$ and $k\in\ZZ^d$.

We say that a pattern $w\in\Sigma^A$ is \emph{admissible} for a configuration space $\Omega\subseteq\Sigma^\Sites$
if $w=x_A$ for some $x\in\Omega$.
We denote the set of admissible patterns of $\Omega$ with shape $A$ by $\Lang_A(\Omega)$,
and the set of all finite admissible patterns by $\Lang(\Omega) \isdef \bigcup_{A\Subset\Sites}\Lang_A(\Omega)$.
Given a pattern $q$ with shape $B\subseteq\Sites$ and another shape $A\subseteq\Sites$,
we denote by $\Lang_{A|q}(\Omega)$ the set of all patterns $p$ with shape $A$
such that $p\lor q$ is well defined and is admissible for $\Omega$.

The shift acts on patterns as well as entire configurations: for $w\in\Sigma^A$
and $k \in \ZZ^d$, $\sigma^k w\in \Sigma^{A - k}$ and for $i \in A - k$, $(\sigma^k w)_i = w_{i+k}$.

\begin{definition}[Topological Markov property; TMP]
\label{def:tmp}
	Let $\Omega\subseteq\Sigma^\Sites$ be a configuration space
	and $A\subseteq B\subseteq\Sites$ two subsets of the sites.
	We say that $B$ is a \emph{memory set} for $A$ in $\Omega$
	if whenever $x$ and $y$ are two configurations admissible for $\Omega$ and
	satisfying $x_{B\setminus A}=y_{B\setminus A}$,
	the configuration $x_B\lor y_{\Sites\setminus A}$ is also admissible for $\Omega$.
	A configuration space $\Omega\subseteq\Sigma^\Sites$ is said to have
	the \emph{topological Markov property} (\emph{TMP} for short) if
	every finite set $A\Subset\Sites$ has a finite memory set $B\Subset\Sites$ in~$\Omega$.
\end{definition}

\begin{example}[Markov property w.r.t.\ a graph]
\label{ex:TMP_on_a_graph}
	A special case of the TMP is the Markov property inherited from
	a locally-finite graph.
	Let $G = (\Sites,\Edges)$ be a locally-finite graph with vertex set $\Sites$
	and edge set $\Edges\subseteq\binom{\Sites}{2}$.
	Let $N(s)\isdef\{s'\in\Sites:\{s,s'\}\in\Edges\}$ denote the set of neighbours of site $s\in\Sites$,
	and for $A\subseteq\Sites$, let $N(A)\isdef\bigcup_{s\in A}N(s)\setminus A$.
	Being locally-finite means that $N(s)$ is finite for every $s\in\Sites$.
	In~\cite{ChaHanMarMeyPav14, ChaMey16}, the notion of a topological Markov field was introduced.
	In our terminology, a topological Markov field on the graph $G = (\Sites,\Edges)$ is a configuration space
	$\Omega\subseteq\Sigma^\Sites$ in which for each $A\Subset\Sites$, $A \cup N(A)$ is a memory set for~$A$.
	
	A shift space which has the TMP but not with respect
	to any locally-finite graph is provided in Example~2.4 of~\cite{BarGomMarTaa18}.
	\hfill\exampleqed
\end{example}

\begin{example}[Sunny side up shift]
\label{exp:sunny-side-up}
	Perhaps the simplest example of a shift space which does not have the TMP is the so-called
	\emph{sunny side up} shift $\Omega_{\leq 1}$, which is defined as the collection of
	all configurations $x\colon\ZZ^d\to\{\symb{0},\symb{1}\}$ with at most one occurrence
	of~$\symb{1}$. Indeed, suppose there is a finite memory set $B \Subset \ZZ^d$ for $A=\{0\}$, then for any $k \in \ZZ^d \setminus B$, we have that the indicator functions $\indicator{\{k\}}$ of $\{k\}$ and $\indicator{\{0\}}$ of $\{0\}$ are in $\Omega_{\leq 1}$ and coincide in $B\setminus \{0\}$. This would imply that $\indicator{\{0,k\}}$ belongs to $\Omega_{\leq 1}$.
	\hfill\exampleqed
\end{example}

A shift space $\Omega\subseteq\Sigma^{\ZZ^d}$ is a \emph{shift of finite type} (\emph{SFT} for short)
if there exists a finite set of finite patterns $\mathcal{F}$
defining $\Omega$ in the sense that $x\in\Omega$
if and only if $(\sigma^k x)_A\notin\mathcal{F}$ for all $k\in\ZZ^d$ and $A\Subset\ZZ^d$.
The set $\mathcal{F}$ in this case is called a set of \emph{forbidden patterns}
defining $\Omega$.
A \emph{nearest-neighbour} SFT is one which has a defining
set of forbidden patterns whose shapes are nearest-neighbour pairs in $\ZZ^d$,
i.e., pairs $\{i,j\}\subseteq\ZZ^d$ with $\norm{i-j}_1=1$.
A pattern $q\in\Sigma^B$ is said to be \emph{locally admissible} with respect to
a defining set $\mathcal{F}$ of forbidden patterns for $\Omega$
if $(\sigma^k q)_A\notin\mathcal{F}$ for all $k\in\ZZ^d$ and $A\Subset\ZZ^d$
such that $A\subseteq B-k$.
Observe that every SFT has the TMP.
Namely, if $F$ denotes the union of the shapes of a finite set of forbidden patterns
defining $\Omega$, then each finite set $A\Subset\ZZ^d$ has $A+F-F$ as a memory set in~$\Omega$.
However, the class of shift spaces with the TMP is
much larger than those which are of finite type. The reader can find many examples in~\cite{ChaHanMarMeyPav14,BarGomMarTaa18}.

A symbol $\blank\in\Sigma$ is said to be a \emph{safe symbol} for
a configuration space $\Omega\subseteq\Sigma^\Sites$ if
for every $x\in\Omega$ and $k\in\Sites$, the configuration obtained
from $x$ by replacing the symbol at site $k$ with $\blank$ is admissible for~$\Omega$.
Clearly, every symbol of a full configuration space is safe.
Observe that $\symb{0}$ is a safe symbol for the sunny side up shift of Example~\ref{exp:sunny-side-up}.

\begin{example}[Hard-core shift]\label{ex:hard-core1}
	A non-trivial example of a shift space with a safe symbol is the \emph{hard-core} shift
	\begin{align}
		\Omega_\hardcore &\isdef
			\big\{
				x\in\{\symb{0},\symb{1}\}^{\ZZ^d} :
				\text{$x_i=x_j=\symb{1}$ implies $\norm{i-j}_1\neq 1$}
			\big\}
	\end{align}
	for which $\symb{0}$ is safe.
	\hfill\exampleqed
\end{example}

Although the class of shift spaces with the TMP is much larger than the class of SFTs,
in presence of a safe symbol, the TMP reduces to the finite type property.
Namely, every shift space with the TMP that has a safe symbol is of finite type
(see Proposition~\ref{prop:tmp+safe->sft}).

\begin{definition}[Asymptotic relation]
	Two configurations $x,y\in\Sigma^\Sites$ are said to be \emph{asymptotic}
	(or \emph{homoclinic}) if they disagree in no more than finitely many sites. Namely, there is $A \Subset \Sites$ so that $x_{\Sites \setminus A} = y_{\Sites \setminus A}$.
	The set of all asymptotic pairs of configurations from a configuration space $\Omega$
	is an equivalence relation which we call the \emph{asymptotic}
	(or \emph{homoclinic}) \emph{relation}
	and denote by $\relation{T}(\Omega)$.
\end{definition}

The equivalence classes of $\relation{T}(\Omega)$ will be referred to as
\emph{asymptotic classes} of $\Omega$.
Given a finite set $A\Subset\Sites$, we write $\relation{T}_A(\Omega)$
for the subset of $\relation{T}(\Omega)$
consisting of all pairs of configurations from $\Omega$ that agree outside~$A$.
Observe that $\relation{T}_A(\Omega)$ is itself an equivalence relation
and is topologically closed in $\Omega\times\Omega$.
Furthermore, $\relation{T}(\Omega)=\bigcup_{A\Subset\Sites}\relation{T}_A(\Omega)$.

\begin{remark}[A topology on the asymptotic relation]
\label{rmk:etale-topology}
	Throughout the text we will implicitly use the following topology on the asymptotic relation $\relation{T}(\Omega)$
	of a configuration space $\Omega$.
	In this topology,
	a sequence of pairs $(x_n,y_n) \in \relation{T}(\Omega)$ converges to a pair
	$(x,y) \in \relation{T}(\Omega)$ if and only if there exists a finite set $A\Subset\Sites$ such that
	$(x_n,y_n)\in \relation{T}_A(\Omega)$ for all sufficiently large $n$ and $x_n \to x$ and $y_n \to y$ with respect
	to the topology on $\Omega$.

	An equivalent way to define the topology on $\relation{T}(\Omega)$, is by declaring that for any $A \Subset \Sites$ the set
	$\relation{T}_A(\Omega)$ is an open subset and that the induced topology on $\relation{T}_A(\Omega)$ coincides
	with the topology induced by $\Omega \times \Omega$.
	It follows that for any $A\Subset\Sites$, $\relation{T}_A(\Omega)$ is also compact.
	Note that the topology on  $\relation{T}(\Omega)$ itself is not compact, and does not always coincide with the relative product topology.  The important feature of this topology is that for each finite set $A$, $\relation{T}_A(\Omega)$
is open as well as compact.
	\hfill\remarkqed
\end{remark}
\begin{remark}
	Whenever $\Omega$ has the TMP, $\relation{T}(\Omega)$  is an \emph{\'etale equivalence relation} with respect
	to the topology described above (in fact, an \emph{AF-equivalence relation}). Interested readers can
	read~\cite{Put18} and references within for more on \'etale equivalence relations,
	approximately finite equivalence relations and their significance in the theory of
	topological orbit equivalence of Cantor minimal systems.
	\hfill\remarkqed
\end{remark}

A configuration space $\Omega\subseteq\Sigma^\Sites$ is said to have
the \emph{pivot property} if for all asymptotic configurations $x,y\in\Omega$,
there is a finite sequence of admissible configurations
$x=z^{(0)},z^{(1)},\ldots,z^{(n)}=y$ such that for each $i=1,2,\ldots,n$,
the configurations $z^{(i-1)}$ and $z^{(i)}$ differ on a single site.
We call $z^{(i-1)}\to z^{(i)}$ an (admissible) \emph{pivot move} at site $k_i$,
where $k_i$ is the unique site at which $z^{(i-1)}$ and $z^{(i)}$ differ.
We say that $\Omega$ has the \emph{uniform pivot property}
if for every $A\Subset\Sites$, there exists $B\Subset\Sites$ %
such that for each $(x,y)\in\relation{T}_A(\Omega)$, there is a sequence of
admissible pivot moves transforming $x$ to $y$ in which all moves are inside $B$.
It follows from the compactness of the sets $\relation{T}_A(\Omega)$
that if a configuration space has both the pivot property
and the TMP, then it also satisfies the uniform pivot property (see Proposition~\ref{prop:TMP+pivot->uniform-pivot}).

Note that every configuration space with a safe symbol has the uniform pivot property.
Namely, if $\blank$ is a safe symbol for $\Omega$,
for every $(x,y)\in\relation{T}(\Omega)$,
we can replace, one-by-one, each of the symbols of $x$ at sites that differ from $y$
by $\blank$, and then revert, again one-by-one, to the symbols of $y$ at sites that
differ from~$x$.

\begin{example}[$q$-coloring shift]
\label{ex:colorings}
	Given an integer $q\geq 2$, the $d$-dimensional \emph{$q$-coloring} shift is defined as
	\begin{align}
		\Omega^d_{\coloring(q)} &\isdef
		\big\{
		x\in\{\symb{0},\dots,q-1\}^{\ZZ^d} :
		\text{$x_i\neq x_j$ whenever $\norm{i-j}_1=1$}
		\big\} \;.
	\end{align}
	The $q$-coloring shift $\Omega^d_{\coloring(q)}$ in statistical mechanics occurs as the support of the \emph{anti-ferromagnetic Potts model}~\cite{GeoHagMae00} in the zero-temperature limit. From the definition, it is clear that $\Omega^d_{\coloring(q)}$ does not have a safe symbol.
	But it has the pivot property when $q\geq 2d+2$ (Proposition~3.4 of~\cite{ChaMey16}).
	The shift space $\Omega^d_{\coloring(q)}$ also has the pivot property when $d=2$
	and $q\in \{2,3\}$ (Proposition~4.4 of~\cite{ChaMey16}).
	In contrast, $\Omega^d_{\coloring(q)}$ does not have the pivot property when $d = 2$ and  $q \in \{4,5\}$.
	Let us illustrate this for $q = 4$. Consider the configuration
	$x\in\{\symb{0}, \symb{1}, \symb{2}, \symb{3}\}^{\ZZ^2}$ defined by
	\begin{align}
		x_{n,m} &\isdef (n+2m) \bmod 4 \;,
	\end{align}
	and note that for every $(i,j)\in \ZZ^2$, we have $x_{i+1,j} = x_{i,j}+1\pmod{4}$,
	$x_{i,j+1} = x_{i,j-1} = x_{i,j}+2 \pmod{4}$, and $x_{i-1,j} = x_{i,j}+3\pmod{4}$.
	This shows that $x \in \Omega^2_{\coloring(4)}$.
	Furthermore, since under $x$ every site $(i,j)\in\ZZ^2$ has three distinct colors in its neighbourhood,
	no other configuration in $\Omega^2_{\coloring(4)}$ differs from $x$ at a single site.
	Nevertheless, there are configurations in $\Omega^2_{\coloring(4)}$ that are asymptotic to $x$
	but distinct from it.
	For instance, the configuration $y\in\{\symb{0}, \symb{1}, \symb{2}, \symb{3}\}^{\ZZ^2}$
	defined by
	\begin{align}
		y_{n,m} &\isdef
			\begin{cases}
				\symb{1} & \mbox{ if } (n,m) = (0,0) \;, \\
				\symb{0} & \mbox{ if } (n,m) = (1,0) \;, \\
				x_{n,m} & \text{otherwise}
		\end{cases}
	\end{align}
	is in $\Omega^2_{\coloring(4)}$ and disagrees from $x$ in exactly two positions.
	In these two configurations the symbols for each site are fixed and cannot be independently pivoted,
	one site at a time. However, we may pivot patterns, changing two sites at a time, as follows:
	\begin{align}
		\xConfig{
			\symb{1} \& \symb{2} \& \symb{3} \& \symb{0} \\
			\symb{3} \& \symb{0} \& \symb{1} \& \symb{2} \\
			\symb{1} \& \symb{2} \& \symb{3} \& \symb{0} \\
		}
		\qquad\longleftrightarrow\qquad
		\xConfig{
			\symb{1} \& \symb{2} \& \symb{3} \& \symb{0} \\
			\symb{3} \& \symb{1} \& \symb{0} \& \symb{2} \\
			\symb{1} \& \symb{2} \& \symb{3} \& \symb{0} \\
		} \;.
	\end{align}
	A similar trick works for $q = 5$. The symbols on the configuration $x$ given by
	$x_{n,m} \isdef (n+3m) \bmod 5$ are fixed but there are exchangeable patterns as above.
	\hfill\exampleqed
\end{example}

A function $f\colon\Omega\to\RR$ on a compact metric space $\Omega$
will be referred to as an \emph{observable}.
The Banach space of continuous observables with the uniform norm $\norm{\cdot}$
will be denoted by $\banach{C}(\Omega)$.
An observable $f\colon\Omega\to\RR$ on a configuration space $\Omega\subseteq\Sigma^\Sites$
is said to be \emph{local} if
there exists a finite set $A\Subset\Sites$, called the \emph{base} of $f$,
such that $f(x)$ is uniquely determined by the restriction $x_A$.
The local observables form a dense linear subspace of~$\banach{C}(\Omega)$.
Given a set $A\subseteq\Sites$, we denote by $\field{F}_A(\Omega)$ %
the $\sigma$-algebra on $\Omega$ generated by cylinder sets whose bases are included in $A$.
The Borel $\sigma$-algebra on $\Omega$ is denoted by $\field{F}(\Omega)=\field{F}_{\Sites}(\Omega)$.
The set of Borel probability measures on $\Omega$ is denoted by $\measures{P}(\Omega)$.
With the weak-* topology,
$\measures{P}(\Omega)$ is a compact metric space. %

\subsection{Cocycles and specifications}
\label{sec:prelim:cocycles_specifications}

The original problem motivating this paper and its predecessors is the problem of existence
of Gibbsian representations for almost-Markovian specifications.
The proofs of our main results invoke $\relation{T}(\Omega)$-cocycles,
which are a certain convenient parametrization of specifications on $\Omega$.
A $\relation{T}(\Omega)$-cocycle assigns to each pair of asymptotic configurations
a real number which can be thought of as the ``energy difference'' between the two configurations.
In this section, we provide definitions and basic properties of specifications and cocycles
and discuss the sense in which they are equivalent.
For the rest of the paper, we will mostly work with cocycles.

\subsubsection{Specifications}
\label{sec:prelim:specification}
In models coming from equilibrium statistical mechanics, the macroscopic states of the system
are represented by probability measures on the configuration space
(i.e., distributions of random fields).
However, the microscopic description coming from physics prescribes not the
measure itself but its conditional probabilities given the configuration
outside each finite set of sites.
A specification refers to a consistent family of such prescribed conditional probabilities.

\begin{definition}[Specification]
	Let $\Omega\subseteq\Sigma^\Sites$ be a configuration space.
	A \emph{specification} on $\Omega$ is a family $K\isdef (K_A)_{A\Subset\Sites}$ of functions
	$K_A: \Omega \times \field{F}(\Omega) \to [0,1]$ such that
	\begin{enumerate}[label={(\roman*)}]
		\item \label{sepc_cond_1} For each $x\in\Omega$, $K_A(x,\cdot)$ is a probability measure on $(\Omega,\field{F}(\Omega))$.
		\item \label{sepc_cond_2} For each $E\in\field{F}(\Omega)$, the function $x \mapsto K_A(x,E)$ is measurable with respect to $\field{F}_{\Sites\setminus A}(\Omega)$.
		\item \label{sepc_cond_3} $K_A(x,[x_B])=1$  whenever $B \Subset\Sites \setminus A$ and $x \in \Omega$.
		
		\item \label{sepc_cond_4}
		For every $x\in\Omega$ and $A\subseteq B\Subset\Sites$,
		\begin{align}
		K_B(x,[x_B]) &= K_B(x,[x_{B\setminus A}])K_A(x,[x_A]) \;.
		\end{align}
	\end{enumerate}
\end{definition}

In the older literature (for instance~\cite{Dob68c}), specifications have been called ``a consistent system of conditional probabilities''.
Following the formulation by Preston~\cite{Pre76} and Georgii~\cite{Geo88}, conditions~\ref{sepc_cond_1}--\ref{sepc_cond_3}
state that each $K_A$ is a proper probability kernel from $(\Omega,\field{F}_{\Sites \setminus A}(\Omega))$ to $(\Omega,\field{F}(\Omega))$.
Condition~\ref{sepc_cond_2} means that $K_A(x,\cdot)$ only depends on the restriction of $x$ to $\Sites \setminus A$.
Condition~\ref{sepc_cond_3} means that the measure $K_A(x,\cdot)$ is concentrated at the configurations
$y\in\Omega$ with $y_{\Sites\setminus A}=x_{\Sites\setminus A}$.
Condition~\ref{sepc_cond_4} is a consistency condition.

\begin{example}[Uniform specification]
\label{exp:uniform-specification}
	An example of a specification on an arbitrary configuration space $\Omega$
	is the \emph{uniform} specification $K^\circ\isdef(K^\circ_A)_{A\Subset\Sites}$
	defined by
	\begin{align}
		K^\circ_A(x,[v_B]\cap[u_A]) &\isdef
			\begin{dcases}
				\frac{1}{\abs[\big]{\Lang_{A|x_{\Sites\setminus A}}(\Omega)}}
					& \text{if $x\in[v_B]$ and $u_A\in\Lang_{A|x_{\Sites\setminus A}}(\Omega)$,} \\
				0	& \text{otherwise,}
			\end{dcases}
	\end{align}
	for $u\in\Lang_A(\Omega)$ and $v\in\Lang_B(\Omega)$ with $A\cap B=\varnothing$.
	\hfill\exampleqed
\end{example}

A probability measure $\mu$ is said to be \emph{consistent} with
a specification $K\isdef(K_A)_{A\Subset\Sites}$ (or \emph{specified} by $K$)
if $\mu\big([p_A]\,\big|\,\field{F}_{\Sites\setminus A}(\Omega)\big)(x)=K_A(x,[p_A])$ for every $p\in\Lang_A(\Omega)$ and
$\mu$-almost every $x$.
Every probability measure on a configuration space $\Sigma^\Sites$
is consistent with some specification on $\Sigma^\Sites$~\cite{Gol78,Pre80},
but one is often interested in specifications that satisfy symmetry/continuity conditions.

If $\Omega$ is a shift space, we say that
a specification $K\isdef(K_A)_{A\Subset\Sites}$ on $\Omega$ is \emph{shift-invariant}
if $K_{k+A}(x,E)=K_A(\sigma^k x, \sigma^k E)$ for every $A\Subset\Sites$ and $k\in\Sites$.

A specification $K\isdef(K_A)_{A\Subset\Sites}$ is \emph{local} (or \emph{Markovian}) if for every
$A\Subset\Sites$, there exists $B\Subset\Sites$ with $B\supseteq A$ such
that for each cylinder set $[w_A]$ with base $A$,
the function $K_A(\cdot,[w_A])$ is $\field{F}_{B\setminus A}(\Omega)$-measurable.
In this case, we refer to $B$ as a \emph{memory set} for $A$ with respect to $K$.
A specification is \emph{continuous}
(or  \emph{almost-Markovian})
if all its kernels $K_A$ are continuous with respect to the first variable.

Compactness of $\measures{P}(\Omega)$
ensures that every continuous specification on $\Omega$ has at least one consistent measure.
This is shown by picking an arbitrary measure $\mu$ (a \emph{boundary condition})
and taking an accumulation point of the sequence of measures $\mu K_A$ defined by $\mu K_A(W) \isdef \int K_A(\cdot,W)\, \dd\mu$ as $A\to\Sites$ along a cofinal
chain of finite subsets of $\Sites$ (i.e., taking a \emph{thermodynamic limit}).
A continuous specification may have more than one consistent measure
(e.g., the specification of the Ising model at low temperature~\cite[Chapter 3]{FriVel17}).
In general, without assuming continuity of the specification,
probability measures consistent with a specification $K\isdef(K_A)_{A\Subset\Sites}$
may or may not exist.
The set of measures consistent with a continuous specification
is a closed and convex subset of $\measures{P}(\Omega)$.
It follows (via averaging, and again using compactness) that on a shift space,
every shift-invariant continuous specification has a shift-invariant consistent measure.

A specification $K\isdef(K_A)_{A\Subset\Sites}$ on a configuration space $\Omega$
is said to be \emph{positive} (or \emph{uniformly non-null})
if $K_A(x,[x_A])>0$ for every $A\Subset\Sites$ and $x\in\Omega$.
The uniform specification $K^\circ$ on $\Omega$ is clearly positive.

We recall the following result, stated and proven in a very similar setting in~\cite[Proposition~2.5]{BarGomMarTaa18}. For completeness, we provide a proof in the appendix~(Section~\ref{apx:specification-vs-cocycle})
\begin{proposition}[Support of a positive continuous specification]
\label{prop:specification:continuous-positive:TMP}
	There exists a positive continuous specification $K\isdef(K_A)_{A\Subset\Sites}$ on
	a configuration space~$\Omega$ if and only if $\Omega$ has the TMP.
\end{proposition}

\subsubsection{Cocycles on the asymptotic equivalence relation}
\label{sec:prelim:cocycles}

Let $\relation{R}\subseteq X\times X$ be an equivalence relation on a set $X$.
A (real-valued) \emph{cocycle} on $\relation{R}$ (or \emph{$\relation{R}$-cocycle})
is a function $\Delta\colon\relation{R}\to\RR$
satisfying $\Delta(a,b)+\Delta(b,c)=\Delta(a,c)$ whenever $(a,b),(b,c)\in\relation{R}$.
Given any $\relation{R}$-cocycle $\Delta\colon\relation{R}\to\RR$ one can find
a ``potential'' function  $F\colon A\to\RR$ such that $\Delta(a,b)=F(b)-F(a)$
for every $(a,b)\in\relation{R}$.
Moreover,  $F$ as above is uniquely determined up to a constant on each  equivalence class of $\relation{R}$.
On the intuitive level, a  cocycle $\psi\colon\relation{T}(\Omega)\to\RR$ on the asymptotic equivalence relation of a configuration space $\Omega\subseteq\Sigma^\Sites$
can be thought of as a notion of ``energy'' on configurations,
with $\psi(x,y)$ being the ``energy required to modify $x$ into $y$''.
Although as stated above we can always find a ``potential'' function $F\colon\Omega \to \RR$ so that $\psi(x,y) = F(y)-F(x)$ for every $(x,y) \in \relation{T}(\Omega)$, even under suitable assumptions on the cocycle $\psi$ one can rarely find such $F$  as above with ``nice'' proprieties such as  shift invariance and continuity or even Borel measurability.
A cocycle on the asymptotic relation $\relation{T}(\Omega)$ will sometimes be referred to as a cocycle on $\Omega$.

We call a cocycle $\psi\colon\relation{T}(\Omega)\to\RR$ \emph{continuous} (or \emph{almost-Markovian}) if
for each $A\Subset\Sites$, the restriction of $\psi$ to $\relation{T}_A(\Omega)$ is continuous
with respect to the induced topology from $\Omega\times\Omega$.
This terminology is justified by the fact that continuous cocycles are precisely the cocycles
that are continuous with respect to the topology introduced in Remark~\ref{rmk:etale-topology}.
By compactness, the restriction of a continuous cocycle to each $\relation{T}_A(\Omega)$
is bounded and uniformly continuous.

We say a cocycle $\psi\colon\relation{T}(\Omega)\to\RR$ is said to be  \emph{local}
(or \emph{Markvovian})
if for every $A\Subset\Sites$, there exists $B\Subset\Sites$ with $B\supseteq A$
such that $\psi(x',y')=\psi(x,y)$ whenever $(x,y),(x',y')\in\relation{T}_A(\Omega)$,
$x'_{B\setminus A}=x_{B\setminus A}$ and $y'_{B\setminus A}=y_{B\setminus A}$.
The set $B$ will then be referred to as the \emph{memory set} associated to $A$
for the cocycle $\psi$. Clearly, any Markovian cocycle is continuous.

Given a shift space $\Omega\in\Sigma^{\ZZ^d}$, a  cocycle $\psi\colon\relation{T}(\Omega) \to \RR$
is called \emph{shift-invariant} if
$\psi(\sigma^k x,\sigma^k y)=\psi(x,y)$ for each $(x,y)\in\relation{T}(\Omega)$ and $k\in\ZZ^d$.

\subsubsection{$\relation{T}(\Omega)$-cocycles as parametrization of positive specifications}
\label{sec:prelim:specification-cocycle:equivalence}
Our goal now is to explain why $\relation{T}(\Omega)$-cocycles can be thought of  as convenient way to parametrize positive specifications on $\Omega$.  To explain this idea, note that
given a non-empty finite set $X$, there is a simple one-to-one correspondence
between the cocycles $\Delta\colon X\times X\to\RR$ on the full equivalence relation $X\times X$
and the positive probability distributions $p\colon X\to(0,1)$ on~$X$.
The correspondence is given by the equality
\begin{align}
	\frac{p(b)}{p(a)} &= \ee^{-\Delta(a,b)}
\end{align}
for $a,b\in X$.
The probability distribution $p$ satisfying this equality
is the \emph{Boltzmann distribution} associated to $\Delta$.

On a configuration space $\Omega\subseteq\Sigma^\Sites$,
there is a similar one-to-one correspondence between
measurable cocycles $\psi\colon\relation{T}(\Omega)\to\RR$ on the asymptotic relation $\relation{T}(\Omega)$
and the positive specifications $K\isdef (K_A)_{A\Subset\Sites}$ on $\Omega$.
The correspondence is given by the equality
\begin{align}
\label{eq:specification-vs-cocycle}
	\frac{K_A(y,[y_A])}{K_A(x,[x_A])} &=
		\ee^{-\psi(x,y)}.
\end{align}
for each $A\Subset\Sites$ and $(x,y)\in\relation{T}_A(\Omega)$.
The correspondence in the reverse  direction is given by

\begin{align}
\label{eq:cocycle-vs-specification}
\psi(x,y) &= -\log\left[\frac{K_A(y,[y_A])}{K_A(x,[x_A])}\right].
\end{align}

We record this correspondence in the following proposition:
\begin{proposition}[Positive specifications $\equiv$ measurable cocycles]
\label{prop:specification-vs-cocycle:equivalence}
	Let $\Omega$ be a configuration space.
	The equalities~\eqref{eq:specification-vs-cocycle} and~\eqref{eq:cocycle-vs-specification} define
	a one-to-one correspondence between measurable cocycles $\psi$ on $\relation{T}(\Omega)$
	and positive specifications $K\isdef (K_A)_{A\Subset\Sites}$ on $\Omega$.
\end{proposition}
The proof of this proposition amounts to  a direct calculation. The conditions~\ref{sepc_cond_1}--\ref{sepc_cond_4}  for $(K_A)_{A \Subset \Sites}$ are all together equivalent to the cocycle equation
\begin{align}
	\psi(z,x) &= \psi(y,x)+\psi(z,y) \qquad\text{for every $(x,y)\in\relation{T}(\Omega)$.}
\end{align}
For completeness, we include a proof of this proposition in Appendix~\ref{apx:specification-vs-cocycle}.
The parametrization of positive specifications via cocycles provides a convenient formalism in which to state our results. As an example, the uniform specification on a configuration space $\Omega$
(Example~\ref{exp:uniform-specification})
corresponds to the zero cocycle $\psi\equiv 0$ on $\relation{T}(\Omega)$.

The cocycle associated to a positive continuous (resp., Markovian)
specification is clearly continuous (resp., Markovian).
The converse is however not true:
according to~\Cref{prop:specification:continuous-positive:TMP},
a positive specification on a configuration space $\Omega$
cannot be continuous unless $\Omega$ has the TMP,
whereas the zero cocycle on every configuration space is continuous (even Markovian).
From the proof of~\Cref{prop:specification-vs-cocycle:equivalence}
(in particular, Equation~\eqref{eq:specification-of-a-cocycle:def}),
it follows that the specification associated to a continuous cocycle on~$\Omega$
is continuous if and only if for every finite pattern $p_A\in\Lang_A(\Omega)$,
the function %
$x\mapsto\indicator{\Omega}(x_{\Sites\setminus A}\lor p_A)$
is continuous (hence, local).
The latter condition is equivalent to $\Omega$ having the TMP.

\begin{proposition}[Positive continuous specification $\equiv$ continuous cocycle]
\label{prop:specification-vs-cocycle:TMP:continuity}
	Let $\Omega$ be a configuration space satisfying the TMP.
	Let $K$ be a positive specification on $\Omega$
	and $\psi$ its corresponding cocycle on $\relation{T}(\Omega)$.
	Then,
	$K$ is continuous (resp., Markovian) if and only if $\psi$ is continuous (resp., Markovian).
\end{proposition}

\subsection{Interactions}
\label{sec:prelim:interactions}

The cocycles and specifications arising in statistical mechanics are usually
generated by interaction potentials.

An \emph{interaction potential} (an \emph{interaction}, for short)
on a configuration space $\Omega \subseteq \Sigma^\Sites$
is a function $\Phi \colon \Lang(\Omega) \to \RR$ assigning a real value $\Phi(w)$
to each admissible pattern $w\in\Lang(\Omega)$.
The ``physical interpretation'' of the value $\Phi(w)$ is ``the energy contribution
of the pattern $w$''.
Given $C\Subset\Sites$, we also define a local function $\Phi_C\colon\Omega\to\RR$
by $\Phi_C(x)\isdef\Phi(x_C)$, so that an interaction can equivalently be described
by the family $(\Phi_C)_{C\Subset\Sites}$.
An interaction $\Phi$ on a shift space $\Omega\subseteq\Sigma^{\ZZ^d}$
is \emph{shift-invariant} if
$\Phi(\sigma^k w)=\Phi(w)$ for each $w\in\Lang(\Omega)$ and $k\in\ZZ^d$,
or equivalently, if $\Phi_{k+C}(x)=\Phi_C(\sigma^k x)$ for all $x\in\Omega$,
$C\Subset\ZZ^d$ and $x\in\Omega$.

Given an interaction $\Phi\colon\Lang(\Omega)\to\RR$,
we formally define for every $(x,y)\in\relation{T}(\Omega)$
\begin{align}
\label{eq:cocycle:interaction}
	\psi_\Phi(x,y) &\isdef \sum_{C\Subset\Sites}\big[\Phi(y_C)-\Phi(x_C)\big]
\end{align}
To make sense of the  infinite sum in~\eqref{eq:cocycle:interaction}, certain assumptions on the interaction $\Phi$ are required.
The simplest case in which the sum is meaningful is when
$\Phi$ has \emph{finite range}, that is, for every $A\Subset\Sites$,
$\Phi_C\equiv 0$ for all but finitely many $C\Subset\Sites$ with $A\cap C\neq\varnothing$.
In this case $\psi_\Phi\colon\relation{T}(\Omega) \to \RR$ is clearly a Markovian cocycle.
We say that $\Phi$ is \emph{uniformly convergent} if for every $A\Subset\Sites$,
the sum in~\eqref{eq:cocycle:interaction} converges uniformly over $\relation{T}_A(\Omega)$,
where the convergence of the series is interpreted in the net sense, along the
directed family of finite subsets of $\Sites$.
In other words, $\Phi$ is uniformly convergent if
for every $\varepsilon>0$, there exists $J_0\Subset\Sites$
such that
\begin{align}
	\abs[\Bigg]{%
		\psi_\Phi(x,y) - \sum_{C\subseteq J}\big[\Phi_C(y)-\Phi_C(x)\big]
	} &< \varepsilon
\end{align}
for every $J\Subset\Sites$ satisfying $J\supseteq J_0$ and each $(x,y)\in\relation{T}_A(\Omega)$.
It follows that $\psi_\Phi\colon\relation{T}(\Omega) \to \RR$ is a continuous cocycle whenever $\Phi$ is  a uniformly convergent interaction,
because its restriction to $\relation{T}_A(\Omega)$ is the uniform limit
of a net of continuous functions.

Whenever we can express a cocycle $\psi\colon\relation{T}(\Omega) \to \RR$ in  the form $\psi = \psi_\Phi$ given by~\eqref{eq:cocycle:interaction}, we call this a \emph{Gibbsian representation}
for the cocycle $\psi$.
For certain applications, it is desirable to have a ``better'' Gibbsian
representation with stronger regularity properties, beyond uniform convergence. In particular, the well known theorems of Dobrushin, Lanford and Ruelle~\cite{Dob68b,LanRue69} relate shift-invariant Gibbs measures and equilibrium measures for the class of norm-summable interactions, which we introduce later.
The purpose of the current paper is to follow up on the question of the existence
of
Gibbsian
representations for continuous (or Markovian) cocycles in terms of
``nice'' families of interactions.

\begin{remark}[Gibbsian
representations of Markovian cocycles]
	The Markovian case was first addressed independently by
	Hammersley and Clifford~\cite{HamCli68}, Averintsev~\cite{Ave72} and Spitzer~\cite{Spi71}.
	They showed that on a full configuration space,
	every cocycle (equivalently, positive specification)
	that has the Markov property with respect to a locally-finite graph on the set of sites
	is generated by a unique finite-range interaction satisfying a certain ``canonical'' property
	(see also~\cite{Gri73}).
	This interaction is ``canonical'' in that it assigns non-zero values only to patterns whose shapes
	are cliques of the graph, and which do not have an occurrence of a fixed
	``vacuum'' symbol.

	For further references and state-of-the-art results on existence or non-existence of
	Gibbsian representations for Markovian specifications on configurations spaces with constraints
	see~\cite{ChaMey16,Cha17}.
	
	While the focus of this paper is on Gibbsian
representations of continuous cocycles,
	in Section~\ref{sec:Kozlov:Markov} we show that every
	Markov cocycle on a configuration space with the TMP is generated by a (``non-canonical'')
	finite-range interaction.
	\hfill\remarkqed
\end{remark}

\subsubsection{Variation-summable interactions.}
\label{sec:prelim:interactions:VS}

Let $\Omega\subseteq\Sigma^\Sites$ be a configuration space.
The \emph{variation} of a continuous observable $f\colon\Omega\to\RR$
on a finite set $A\Subset\Sites$ is defined as
\begin{align}
	\Var_A(f) &\isdef
		\sup_{(x,y)\in\relation{T}_A(\Omega)} \abs[\big]{f(y)-f(x)} \;.
\end{align}
We use the shorthand $\Var_s( \cdot )\isdef\Var_{\{s\}}( \cdot )$ for $s\in\Sites$.
Note that $\Var_A(f)=0$ whenever $f$ is a local observable whose base
does not intersect $A$.

An interaction $\Phi \colon \Lang(\Omega) \to \RR$ is called \emph{variation-summable}
if for every $A\Subset\ZZ^d$,
\begin{align}
\label{eq:variation_summable}
	\sum_{\substack{C\Subset\Sites\\ C\cap A\neq\varnothing}}\Var_A(\Phi_C)  < \infty \;.
\end{align}
Observe that~\eqref{eq:variation_summable} implies that the
sum~\eqref{eq:cocycle:interaction} converges absolutely,
uniformly over each $\relation{T}_A(\Omega)$.
In particular, every variation-summable interaction is uniformly convergent.

\begin{proposition}[Variation-summability under uniform pivot property]
\label{prop:interaction:variation-summable:single-site}
	Let $\Omega\subseteq\Sigma^\Sites$ be a configuration space satisfying
	the uniform pivot property.
	Then, an interaction $\Phi$ on $\Omega$ is variation-summable
	if and only if
	\begin{align}
		\sum_{\substack{C\Subset\Sites\\ C\ni s}} \Var_s(\Phi_C) &< \infty
	\end{align}
	for every $s\in\Sites$.
\end{proposition}
The proof of the above proposition can be found in Appendix~\ref{apx:interactions:variation-summable}.

Let now $\Omega\subseteq\Sigma^{\ZZ^d}$ be a shift space with the uniform pivot property.
From the above proposition, it follows that
a shift-invariant interaction $\Phi$ on $\Omega$ is variation-summable if and only if
\begin{align}\label{eq:normVS_def}
	\normVS{\Phi}
		&\isdef \sum_{\substack{C\Subset\ZZ^d\\ C\ni 0}} \Var_0(\Phi_C)
\end{align}
is finite.
The function $\normVS{\cdot}$ is a seminorm, because it clearly satisfies
the subadditivity and homogeneity conditions.
It is however not a norm because, for instance, $\normVS{c+\Phi}=\normVS{\Phi}$ for every $c\in\RR$.

Define an equivalence relation $\overset{\Omega}{\sim}$ on $\Lang(\Omega)$ by declaring $w,w' \in \Lang(\Omega)$ to be equivalent if and only if they have the same shape $C\Subset\ZZ^d$ and there exist
$(x,x') \in \relation{T}(\Omega)$ such that $x_C = w$ and $x'_C =w'$.

\begin{lemma}
\label{lem:Phi_bound_VS}
	Let $\Omega$ be a shift space with the pivot property and  let $\Phi\colon\Lang(\Omega) \to \RR$ be
	a shift-invariant variation-summable interaction.
	Then, for every $A \Subset \ZZ^d$ and  $w,w' \in \Lang_A(\Omega)$ with $w \overset{\Omega}{\sim} w'$ we have
	\begin{align}
		\abs[\big]{\Phi(w)- \Phi(w')} \leq  \abs{\Lang_A(\Omega)}\, \normVS{\Phi} \;.
	\end{align}
\end{lemma}

Lemma~\ref{lem:Phi_bound_VS} can be used to obtain the following.
\begin{proposition}
\label{prop:interaction:variation-summable:seminorm}
	Let $\Omega$ be a shift space satisfying
	the uniform pivot property and let $\Phi\colon\Lang(\Omega) \to \RR$ be a shift-invariant variation-summable interaction.
	Then $\normVS{\Phi}=0$ if and only if for every $C \Subset \ZZ^d$ the function $\Phi_C$
	is constant on each asymptotic class of $\Omega$.
\end{proposition}
See Appendix~\ref{apx:interactions:variation-summable} for the proof of Lemma~\ref{lem:Phi_bound_VS} and Proposition~\ref{prop:interaction:variation-summable:seminorm}.

If we identify two interactions $\Phi^{(1)}$ and $\Phi^{(2)}$ whenever
$\normVS{\Phi^{(2)}-\Phi^{(1)}}=0$, then we get a normed linear space.
This space together with the norm $\normVS{\cdot}$ actually forms
a Banach space which we denote by $\banach{B}_\VS(\Omega)$.
In the specific case where $\Omega$ is a shift space that admits a safe symbol $\blank$, we can identify  $\banach{B}_\VS(\Omega)$
with the space of interactions $\Phi$  satisfying $\normVS{\Phi} < \infty$ and $\Phi(\blank^C) =0$ for every $C \Subset \ZZ^d$, because for any interaction $\Phi$ satisfying $\normVS{\Phi} < \infty$, there exists a unique interaction $\Phi'$ satisfying
$\normVS{\Phi-\Phi'} =0$ and $\Phi'(\blank^C) =0$ for every $C \Subset \ZZ^d$. This basic idea extends to the more general case where $\Omega$ is a shift space with the pivot property.  Namely, choose a set $L_0 \subseteq \Lang(\Omega)$ which includes precisely one representative from each equivalence class of $\overset{\Omega}{\sim}$.  Now for any interaction $\Phi$ satisfying $\normVS{\Phi} < \infty$, there exists a unique interaction $\Phi'$ satisfying $\normVS{\Phi-\Phi'} =0$ such that $\Phi'(w) =0$ for every $w \in L_0$.

For completeness, we give a proof of completeness of $\normVS{\cdot}$.
\begin{proposition}[Completeness of the $\VS$-norm]
\label{prop:BanachVS_complete}
	Let $\Omega$ be a shift space with the uniform pivot property.
	Then, the norm $\normVS{\cdot}$ on $\banach{B}_\VS(\Omega)$ is complete.
\end{proposition}
\begin{proof}
	Let $\Phi^{(1)}, \Phi^{(2)}, \ldots$ be a Cauchy sequence in $\banach{B}_\VS(\Omega)$. We need to show that the sequence converges with respect to the  norm $\normVS{\cdot}$.  Let $L_0 \subseteq\Lang(\Omega)$ be a set containing precisely one representative from each equivalence class of $\overset{\Omega}{\sim}$.
	By the remark above the proposition, it is no loss of generality to assume that $\Phi^{(n)}(w)=0$ for each $w\in L_0$ and $n \in \NN$.
	
	Take $w \in  \Lang(\Omega)$, and let  $w'$ be the unique element of $L_0$ such that $w \overset{\Omega}{\sim} w'$.  Then, $\Phi^{(n)}(w)=\Phi^{(n)}(w)-\Phi^{(n)}(w')$.
	Using Lemma~\ref{lem:Phi_bound_VS}, it follows that $(\Phi^{(n)}(w))_{n=1}^\infty$ is a Cauchy sequence of real numbers, and thus converges to a real number which we denote by $\Phi(w)$. This defines an interaction $\Phi\colon\Lang(\Omega) \to \RR$ which is the pointwise limit of $(\Phi^{(n)})_{n=1}^\infty$.
	Clearly, $\Phi$ is shift-invariant and satisfies $\Phi(w')=0$ for each $w'\in L_0$.
	
	From pointwise convergence, it directly follows that $\Var_0(\Phi_C) = \lim_{n \to \infty}\Var_0(\Phi^{(n)}_C)$ for every $C \Subset \ZZ^d$.  Since $(\Phi^{(n)})_{n=1}^\infty$ is a  Cauchy sequence in $\banach{B}_{\VS}(\Omega)$, the series $\sum_{\substack{C\Subset\ZZ^d\\ C\ni 0}} \Var_0(\Phi^{(n)}_C)$ converges uniformly in $n$, in the sense that for any $\varepsilon >0$, there exists a finite set $B \Subset \ZZ^d$ such that
	\begin{align}
		\sup_n \sum_{\substack{C\Subset\ZZ^d\\ C\ni 0,\ C \not\subseteq B }} \Var_0(\Phi^{(n)}_C)
			&< \varepsilon \;,
	\end{align}
	This shows that
	\begin{align}
		\normVS{\Phi} &\leq
			\sum_{\substack{C\Subset\ZZ^d\\ C\ni 0,\ C \subseteq B }} \Var_0(\Phi_C) + \varepsilon \;,
	\end{align}
	and in particular $\Phi \in \banach{B}_{\VS}(\Omega)$.
	It also follows that
	\begin{align}
		\limsup_{n \to \infty} \normVS{\Phi-\Phi^{(n)}} &\leq
			\lim_{n \to \infty}\sum_{\substack{C\Subset\ZZ^d\\ C\ni 0,\ C \subseteq B }}
				\Var_0(\Phi_C-\Phi^{(n)}_C) + \varepsilon = \varepsilon \;.
	\end{align}
	Since $\varepsilon >0$ was arbitrary, this shows that the sequence $(\Phi^{(n)})_{n=1}^\infty$ converges in norm to $\Phi$.
\end{proof}

Sullivan showed that every shift-invariant continuous cocycle
(equivalently, shift-invariant positive continuous specification)
on a full shift space is generated by a shift-invariant variation-summable interaction
(see Corollary of Theorem~2 in~\cite{Sul73}). To be precise, the statement in~\cite{Sul73} only mentions a slightly weaker property called  ``absolute convergence'', but an inspection of the proof reveals that it yields a variation-summable interaction.
In Theorem~\ref{thm:sull2} below, we extend Sullivan's result to more general families of
shift spaces.

\subsubsection{Norm-summable interactions.}
\label{sec:prelim:interactions:NS}

An interaction $\Phi$ on a configuration space $\Omega\subseteq\Sigma^\Sites$
is \emph{norm-summable} (also called \emph{absolutely summable})
if for every $A\Subset\Sites$,
\begin{align}
\label{eq:norm_summable}
	\sum_{\substack{C\Subset\Sites\\ C\cap A\neq\varnothing}}\norm{\Phi_C} < \infty \;,
\end{align}
where $\norm{\cdot}$ denotes the uniform norm.
Clearly, every norm-summable interaction is also variation-summable.
In particular, norm-summable interactions are uniformly convergent.

Observe that for every interaction $\Phi$,
\begin{align}
\label{eq:norm_sum_finite}
	\sum_{\substack{C\Subset\Sites\\ C\cap A\neq\varnothing}}\norm{\Phi_C}
	&= \sum_{a\in A}\sum_{\substack{C\Subset\Sites\\ C\ni a}}\frac{1}{\abs{A\cap C}}\norm{\Phi_C}
	\leq \sum_{s \in A} \sum_{\substack{C\Subset\Sites\\ C\ni s}}\norm{\Phi_C} \;.
\end{align}
Thus,
$\Phi$ is norm-summable if and only if
\begin{align}
\label{eq:norm_summable:2}
	\sum_{\substack{C\Subset\Sites\\ C\ni s}}\norm{\Phi_C} < \infty
\end{align}
for every $s\in\Sites$.

Of special interest is the case in which $\Omega$ is a shift space
and $\Phi$ is a shift-invariant interaction.
In this case, $\Phi$ is norm-summable if and only if
\begin{align}
	\normNS{\Phi} & \isdef \sum_{\substack{C\Subset\Sites\\ C\ni 0}}
\norm{\Phi_C} < \infty \;.
\end{align}
It is well known that the space of shift-invariant norm-summable
interactions on a shift space $\Omega$ with the norm $\normNS{\cdot}$
is a Banach space (see Section 4.1 of~\cite{Rue04}).
We denote this Banach space by $\banach{B}_\NS(\Omega)$ .
Once again, for completeness, we give a proof of completeness.
\begin{proposition}[Completeness of the $\NS$-norm]
\label{prop:BanachNS_complete}
	Let $\Omega$ be a shift space.
	Then, the norm $\normNS{\cdot}$ on $\banach{B}_\NS(\Omega)$ is complete.
\end{proposition}
\begin{proof}
	Suppose that $\sum_{n=1}^\infty\normNS{\Phi^{(n)}}<\infty$. It suffices to show that there is a norm-summable interaction~$\Phi$ such that
$\normNS[\big]{\sum_{n=1}^N {\Phi^{(n)}} - \Phi}$ converges to zero as $N\to\infty$.

	For this, first observe that $\sum_{n=1}^\infty\norm[\big]{\Phi^{(n)}_C}<\infty$
	for every $C\Subset\ZZ^d$ such that $C \ni 0$.
	Since $\RR^{\Lang_C(\Omega)}$ %
	is a finite dimensional Banach space with the uniform norm,
	$\Phi_C\isdef\sum_{n=1}^\infty\Phi^{(n)}_C$ exists uniformly.
	Clearly, this defines a shift-invariant interaction $\Phi$.
	Furthermore, $\Phi$ is norm-summable because
	\begin{align}
		\normNS{\Phi} &=
			\sum_{\substack{C\Subset\ZZ^d\\ C\ni 0}} \norm{\Phi_C} \le
			\sum_{\substack{C\Subset\ZZ^d\\ C\ni 0}}\sum_{n=1}^\infty \norm[\big]{\Phi^{(n)}_C} =
			\sum_{n=1}^\infty\sum_{\substack{C\Subset\ZZ^d\\ C\ni 0}} \norm[\big]{\Phi^{(n)}_C} =
			\sum_{n=1}^\infty \normNS{\Phi^{(n)}} <
			\infty \;.
	\end{align}
	It remains to show that $\normNS[\big]{\sum_{n=1}^N {\Phi^{(n)}} - \Phi}\to 0$ as $N\to\infty$.
	
	Enumerate the finite subsets of $\ZZ^d$ that contain $0$ as $C_1, C_2, \ldots$. Since
	\begin{align}
		\sum_{m=1}^\infty \sum_{n=1}^\infty\norm{\Phi_{C_m}^{(n)}}
			= \sum_{n=1}^\infty\normNS{\Phi^{(n)}}<\infty \;,
	\end{align}
	given $\varepsilon > 0$, there exists $M>0$ such that
	\begin{align}
		\sum_{m=M+1}^\infty \sum_{n=1}^\infty\norm[\big]{\Phi_{C_m}^{(n)}} < \varepsilon.
	\end{align}
	Moreover, since $\Phi$ is norm-summable, we may also assume that
	\begin{align}
		\sum_{m=M+1}^\infty\norm{\Phi_{C_m}} < \varepsilon.
	\end{align}
	It follows that for every $N\geq 1$,
	\begin{align}
		\normNS[\bigg]{\sum_{n=1}^N {\Phi^{(n)}} - \Phi} &=
			\sum_{m=1}^\infty \norm[\bigg]{\sum_{n=1}^N {\Phi_{C_m}^{(n)}} - \Phi_{C_m}} \\
		&\leq
			\sum_{m=1}^M \norm[\bigg]{\sum_{n=1}^N {\Phi_{C_m}^{(n)}} - \Phi_{C_m}}
			+ \sum_{m=M+1}^\infty\norm[\bigg]{\sum_{n=1}^N {\Phi_{C_m}^{(n)}} - \Phi_{C_m}} \\
		&\leq
			\sum_{m=1}^M\norm[\bigg]{\sum_{n=1}^N {\Phi_{C_m}^{(n)}} - \Phi_{C_m}}
			+ \sum_{m=M+1}^\infty\sum_{n=1}^N \norm[\big]{{\Phi_{C_m}^{(n)}}}
			+ \sum_{m=M+1}^\infty\norm{\Phi_{C_m}} \\
		&\leq
			\sum_{m=1}^M\norm[\bigg]{\sum_{n=1}^N {\Phi_{C_m}^{(n)}} - \Phi_{C_m}}
			+ 2\varepsilon \;.
	\end{align}
	Now choose $N_0$ so large that for each $m = 1,2,\ldots, M$ and $N \geq N_0$,
	\begin{align}
		\norm[\bigg]{\sum_{n=1}^N {\Phi_{C_m}^{(n)}} - \Phi_{C_m}} &< \varepsilon/M
	\end{align}
	Then, for $N \geq N_0$,
	\begin{align}
		\normNS[\bigg]{\sum_{n=1}^N {\Phi^{(n)}} - \Phi} < 3\varepsilon \;,
	\end{align}
	concluding the proof.
\end{proof}

Let us remark that the set of shift-invariant finite-range interactions
on a shift space $\Omega$ is dense in~$\banach{B}_\NS(\Omega)$.

Kozlov~\cite{Koz74} showed that every continuous cocycle (equivalently, positive continuous specification)
on a full configuration space is generated by a norm-summable interaction.
However, when the cocycle is shift-invariant, Kozlov's construction does not provide
a shift-invariant norm-summable interaction.
This raises the question of whether every shift-invariant continuous cocycle on a full shift
is generated by a shift-invariant norm-summable interaction.
The main result of the current paper (Theorem~\ref{thm:Kozlov_is_annoying}) answers this question
in the negative:
there exists a continuous shift-invariant cocycle on the one-dimensional binary full shift
which is not generated by any shift-invariant norm-summable interaction.
On the other hand, we extend Kozlov's non-shift-invariant result to continuous cocycles
on any configuration space satisfying the TMP (Theorem~\ref{thm:gibbs-cocycle:interaction:non-invariant}).

\subsection{A Banach space of cocycles}
\label{sec:prelim:Banach_cocycles}

In this section we show that the space of all
continuous shift-invariant cocycles on a shift space   that has the TMP
and the pivot property is in fact a Banach space
with an appropriately defined norm. This result, which is  crucial for the proof of Theorem~\ref{thm:sull2}, is essentially due to  Sullivan~\cite{Sul73}, at least for cocycles on the full-shift.

Let $\Omega\subseteq\Sigma^{\ZZ^d}$ be a shift space.
For a site $k \in \ZZ^d$, define $\zeta_k\colon \Omega \to \Omega$ by
\begin{align} \label{eq:zeta_k_def}
	(\zeta_k x)_s \isdef
		\begin{cases}
			x_s		& \text{if $s\neq k$,} \\
			a(x,k) 	& \text{if $x=k$,}
		\end{cases}
\end{align}
for all $s\in\ZZ^d$,
where $a(x,k)$ is the smallest symbol $a \in \Sigma$, according to some fixed total order on $\Sigma$,
such that $x_{\ZZ^d \setminus \{k\}} \lor a_{\{k\}}$
is admissible in $\Omega$.
Observe that if $\Omega$ has the TMP, then $a(x,k)$ is a function of
$x_{k+B}$ for some $B \Subset \ZZ^d$. In particular, the maps
$\zeta_k$ are continuous.
When the shift space $\Omega$ has a safe symbol $\blank$, it is convenient
to assume that $\blank$ is the minimum element of $\Sigma$.
This will make sure that $a(x,k)=\blank$ for every $x \in \Omega$ and $k \in \ZZ^d$,
and the definition of $\zeta_k$ becomes
\begin{align}
\label{eq:zeta_k_safe_symbol}
	(\zeta_k x)_s \isdef
		\begin{cases}
			x_s		& \text{if $s\neq k$,} \\
			\blank 	& \text{if $s=k$.}
		\end{cases}
\end{align}

Given a shift-invariant cocycle $\psi$ on $\relation{T}(\Omega)$, define
\begin{align}
\label{eq:Sull_norm_def}
	\normsull{\psi} & \isdef \sup_{x \in \Omega} \abs[\big]{\psi(x,\zeta_0 x)} \;.
\end{align}
When $\Omega$ has the pivot property,
$\normsull{\cdot}$ is a norm, which we call the \emph{Sullivan} norm. If $\Omega$ also has the TMP, Proposition~\ref{prop:Sullivan_norm} below shows that the Sullivan norm is complete and hence turns
the space of shift-invariant continuous cocycles on $\relation{T}(\Omega)$
into a Banach space, which we denote by $\banach{B}_\Sull(\Omega)$.

In order to prove the completeness of the Sullivan norm, we use the following lemma.
\begin{lemma}
\label{lem:cont_inv_coycles_embed_linearly_in_cont_function}
	Let $\Omega$ be a shift space with the TMP and the pivot property.
	Consider the map $F\colon \banach{B}_\Sull(\Omega) \to \banach{C}(\Omega)$ given by
	\begin{align}
		F(\psi)(x) & \isdef \psi(x,\zeta_0 x) %
	\end{align}
	for each $\psi\in\banach{B}_\Sull(\Omega)$ and $x\in\Omega$.
	Then $F$ is an injective linear map.
	Furthermore, the image $F\big(\banach{B}_\Sull(\Omega)\big)$ is a closed linear
	subspace of $\banach{C}(\Omega)$ with respect to the topology of the uniform norm. %
\end{lemma}
\begin{proof} %
	Continuity of $F(\psi)$ follows from continuity of $\psi$ and continuity of the map
	$\zeta_0$. It is easy to see that $F$ is linear.

	To prove that $F$ is injective, we need to use the pivot property of $\Omega$.
	Suppose that $F(\psi)=0$. We will show that $\psi$ is the zero cocycle, meaning that
	$\psi(x, y)=0$ for every $(x,y) \in \relation{T}(\Omega)$. By the pivot property, there exists
	a sequence of sites $k_1, \ldots, k_n \in \ZZ^d$ and
	a sequence of configurations $x = z^{(0)}, z^{(1)}, \ldots, z^{(n)}=y$ such that
	$(z^{(i-1)}, z^{(i)}) \in \relation{T}_{k_i}(\Omega)$ for $i=1, \ldots, n$.
	By the cocycle property and shift-invariance,
	\begin{align}
		\psi(x, y) &= \sum_{i=1}^n\psi\big(z^{(i-1)}, z^{(i)}\big)
			= \sum_{i=1}^n
				\psi\big(\sigma^{k_i}z^{(i-1)}, \sigma^{k_i}z^{(i)}\big) \;.
	\end{align}
	Since $\big(\sigma^{k_i}z^{(i-1)}, \sigma^{k_i}z^{(i)}\big)\in \relation{T}_0(\Omega)$ for every
	$i=1, \ldots, n$, it suffices to show that $\psi(x, y) = 0$ for every $(x, y) \in \relation{T}_0(\Omega)$.
	Note that for $(x, y) \in \relation{T}_0(\Omega)$, we have $\zeta_0 x=\zeta_0 y$.
	Hence, the cocycle equation gives
	\begin{align}
	\label{eq:psi_T_0_zeta}
	 	\psi(x,y) &=
	 		\psi(x,\zeta_0x) + \psi(\zeta_0y, y) = F(\psi)(x) - F(\psi)(y) = 0
	\end{align}
	whenever $(x,y)\in\relation{T}_0(\Omega)$.
	This completes the proof of injectivity of $F$.

	It remains to show that $F\big(\banach{B}_\Sull(\Omega)\big)$ is closed in $\banach{C}(\Omega)$
	with respect to the uniform norm $\norm{\cdot}$.
	Suppose $f$ is in the closure of $F\big(\banach{B}_\Sull(\Omega)\big)$. Thus, there exists a sequence
	$(\psi_t)_{t=1}^\infty$ in $\banach{B}_\Sull(\Omega)$ such that $f$ is the uniform limit of $F(\psi_t)$.
	We show that the limit
	 \begin{align}
	 	\psi & \isdef  \lim_{t \to \infty}\psi_t
	\end{align}
	exists and that it is uniform on $\relation{T}_A(\Omega)$ for every $A \Subset \ZZ^d$.
	So, let $A\Subset\ZZ^d$.
	Then, for every $(x,y)\in\relation{T}_A(\Omega)$, as above, we can find
	a sequence of sites $k_1, \ldots, k_n \in \ZZ^d$ and
	a sequence of configurations $x = z^{(0)}, z^{(1)}, \ldots, z^{(n)}=y$ such that
	$(z^{(i-1)}, z^{(i)}) \in \relation{T}_{k_i}(\Omega)$ for $i=1, \ldots, n$.
	Hence,
	 \begin{align}
	 	\psi_t(x,y) &=
	 		\sum_{i=1}^n\psi_t \big( z^{(i-1)},z^{(i)} \big) =
	 		\sum_{i=1}^n
			\Big[
				F(\psi_t)\big(\sigma^{k_i}z^{(i-1)}\big)-F(\psi_t)\big(\sigma^{k_i}z^{(i)}\big)
			\Big] \;,
	\end{align}
	where the second equality is obtained similarly to~\eqref{eq:psi_T_0_zeta}
	with $\psi_t$ replacing $\psi$.
	Because $\Omega$ has  the TMP and the pivot property, it also has the uniform pivot property, so the number $n$ can be chosen independently of the pair $(x,y)\in\relation{T}_A(\Omega)$,
	thus only depending on $A$.
	Furthermore, the TMP implies that the sites $k_1, \ldots, k_n$ and
	the configurations $z^{(0)}, z^{(1)}, \ldots, z^{(n)}$
	can be chosen to be continuous in $(x,y)\in\relation{T}_A(\Omega)$,
	meaning that there exists a finite set $B\supseteq A$ such that
	the sites $k_i$ and the symbols $z^{(i)}_{k_i}$ depend only on $x_B$ and $y_B$.
	We see that as $t\to\infty$, the right-hand side converges to
	\begin{align}
		\sum_{i=1}^n
		\Big[
			f\big(\sigma^{k_i}z^{(i-1)}\big)
			-f\big(\sigma^{k_i}z^{(i)}\big)
		\Big] \;.
	\end{align}
	Furthermore, the convergence is uniform over $(x,y) \in \relation{T}_A(\Omega)$
	because the convergence of $F(\psi_t)$ to $f$ is uniform.
	Since a pointwise
	limit of cocycles is also a cocycle, $\psi$ is a cocycle on $\relation{T}(\Omega)$.
	Since for each $A\Subset\ZZ^d$, the convergence is uniform over $\relation{T}_A(\Omega)$,
	it follows that $\psi$ is a continuous cocycle.
	Shift-invariance of $\psi$ follows from shift-invariance of $\psi_t$ and taking limits.
	Lastly, for every $x \in \Omega$,
	\begin{align}
		F(\psi)(x) &=
			\psi(x,\zeta_0x) =
			\lim_{t\to\infty} \psi_t(x,\zeta_0 x) =
			\lim_{t \to \infty} F(\psi_t)(x) =
			f(x) \;.
	\end{align}
	Thus, $F(\psi)=f$.	
\end{proof}

\begin{proposition}[Completeness of the Sullivan norm]
\label{prop:Sullivan_norm}
	Let $\Omega$ be a shift space with the TMP and the pivot property.
	Then $\norm{\cdot}_{\Sull}$
	is a complete norm on $\banach{B}_\Sull(\Omega)$.
\end{proposition}
\begin{proof}
	The function $F$ in the statement of
	Lemma~\ref{lem:cont_inv_coycles_embed_linearly_in_cont_function}
	is an embedding of the linear space $\banach{B}_\Sull(\Omega)$ onto a closed
	linear subspace of $\banach{C}(\Omega)$, thus $F(\banach{B}_\Sull(\Omega))$ together
	with the uniform norm $\norm{\cdot}$ is a Banach space. By definition,
	the norm $\normsull{\cdot}$ is simply the pullback of
	the uniform norm $\norm{\cdot}$ via $F$, and thus it is a complete norm.
\end{proof}

\begin{remark}[Completeness of the Sullivan norm: another sufficient condition]
\label{rmk:Sullivan_norm_safe}
	If $\Omega$ is a shift space with a safe symbol, then it satisfies the pivot property, in fact the uniform pivot property, but not necessarily the TMP.  Nevertheless,  we claim that the conclusion of Proposition~\ref{prop:Sullivan_norm} still holds.  To see this, first observe that the TMP was used only to show continuity of certain quantities within the proof of  Lemma~\ref{lem:cont_inv_coycles_embed_linearly_in_cont_function},
	namely, continuity of
	\begin{enumerate*}[label={(\alph*)}]
		\item \label{item:sullivan:completeness:safe:1} $\zeta_0x$ as a function of $x$, and
		\item \label{item:sullivan:completeness:safe:2} the sites $k_i$, and
		\item \label{item:sullivan:completeness:safe:3} the symbols $z^{(i)}_{k_i}$,
			as functions of $(x,y) \in \relation{T}_A(\Omega)$.
	\end{enumerate*}
	But the safe symbol assumption guarantees continuity of these quantities,
	even without the TMP assumption:
	
	For~\ref{item:sullivan:completeness:safe:1}, continuity amounts to continuity of  $(\zeta_0x)_0$, which in this case is the constant $\blank$.
	
	For~\ref{item:sullivan:completeness:safe:2} and~\ref{item:sullivan:completeness:safe:3}, the sequence of sites $k_i$ is the concatenation of
	two enumerations of the sites of $A$ at which $x$ and $y$ differ, and the sequence of symbols $z^{(i)}_{k_i}$ is the concatenation of a sequence of the constants
	$\blank$ followed by the sequence of symbols $y_{k_i}$.
	\hfill\remarkqed
\end{remark}

\subsection{Surjective maps between Banach spaces}
\label{sec:prelim:surjective_Banach}

We are interested in the question of whether every
cocycle of a given type
has a Gibbsian representation of the form~\eqref{eq:cocycle:interaction}
in terms of a ``well-behaved'' interaction.
Specifically, for a shift space $\Omega$, we would like to know if every
shift-invariant continuous cocycle on $\Omega$ can be represented by
a shift-invariant norm-summable interaction and, failing that, if it
can be represented by a shift-invariant variation-summable interaction.
These questions can be reformulated as the question of
surjectivity of certain bounded linear transformations between Banach spaces.

\begin{proposition}[Continuity of $\Phi\mapsto\psi_\Phi$]
\label{prop:BLT}Let $\Omega$ be a shift space with the TMP and the pivot property.
The map $\Phi \mapsto \psi_\Phi$ defines a bounded linear transformation from the Banach space $\banach{B}_{\NS}(\Omega)$ to the Banach space $\banach{B}_{\Sull}(\Omega)$ and also from the Banach space $\banach{B}_{\VS}(\Omega)$ to the Banach space $\banach{B}_{\Sull}(\Omega)$.
\end{proposition}
\begin{proof}
	Since
	\begin{align}
	\label{eq:bound_linear_operator}
		\normsull{\psi_\Phi} &
		= \sup_{x \in \Omega} \abs[\big]{\psi_\Phi\big(x,\zeta_0 x\big)}
		\leq \sup_{x \in \Omega} \sum_{A\Subset \ZZ^d} \abs[\big]{\Phi_{A}(x) - \Phi_{A}(\zeta_0 x)}
		= \sup_{x \in \Omega}\sum_{A \ni 0} \abs[\big]{\Phi_{A}(x) - \Phi_{A}(\zeta_0 x)} \;,
	\end{align}
	we deduce that
	\begin{align}
		\normsull{\psi_\Phi} &
		\leq \sum_{A \ni 0}  \sup_{x \in \Omega}\abs{\Phi_{A}(x)}
		+ \sum_{A \ni 0}   \sup_{x \in \Omega}\abs{\Phi_{A}(\zeta_0 x)}
		\le 2\normNS{\Phi}
	\end{align}
	and thus $\Phi \mapsto \psi_\Phi$, viewed as a linear operator from $\banach{B}_{\NS}(\Omega)$ to $\banach{B}_{\Sull}(\Omega)$
	has operator norm at most~$2$.
	Similarly, again from	\eqref{eq:bound_linear_operator}, we deduce that
	\begin{align}
		\normsull{\psi_\Phi} &
		\leq \sum_{A \ni 0} \sup_{x \in \Omega}\abs[\big]{\Phi_{A}(x) - \Phi_{A}(\zeta_0 x)}
		\leq \sum_{A \ni 0} \Var_0(\Phi_A)
		= \normVS{\Phi}
	\end{align}
	In particular $\psi_\Phi =0$ whenever $\normVS{\Phi}=0$
	and thus $\Phi \mapsto \psi_\Phi$ gives a well defined   operator from  $\banach{B}_{\VS}(\Omega)$ to $\banach{B}_{\Sull}(\Omega)$ of norm at most~$1$.
\end{proof}

Let $\ball{Z}{R}$ denote the open ball of
radius $R\geq 0$ centered at the origin of a Banach space $Z$.

\begin{proposition}[Surjectivity of bounded linear maps]
\label{prop:Banach_equivalences}
	Let $(X,\norm{\cdot}_X)$ and $(Y,\norm{\cdot}_Y)$ be Banach
	spaces and $T \colon X \to Y$ be a bounded linear transformation.
	The following are equivalent:
	\begin{enumerate}[label={\rm(\alph*)}]
		\item \label{prop:Banach_equivalences:1}
			$T$ is surjective.
		\item \label{prop:Banach_equivalences:2}
			For some $R > 0$, $T\big(\ball{X}{R}\big)$ contains $\ball{Y}{1}$.
		\item \label{prop:Banach_equivalences:3}
			For some $R > 0$, $T\big(\ball{X}{R}\big)$ is dense in $\ball{Y}{1}$.
	\end{enumerate}
\end{proposition}

This proposition is an exercise in functional analysis based on the open mapping theorem.  We give a  proof in Appendix~\ref{apx:surjectivity}.

So, in order to show that a bounded linear transformation from one Banach space
to another is not surjective, we only need to show that condition \ref{prop:Banach_equivalences:2}
above does not hold, which essentially says that the map is not open.

The real dual of a Banach space $Z$ will be denoted by $Z^*$.
In accordance with this, we write $\norm{\cdot}^*_Z$ for the dual norm on $Z^*$. Let us also write for a real-valued function $f\colon \NN \to \RR$ that $f(n)=o(n)$ if $\lim_{n \to \infty} \frac{f(n)}{n} = 0$ and $f(n) = \Omega(n)$ if $f(n) \geq Kn$ for some positive $K>0$ and all large enough $n$.
\begin{corollary}[Sufficient condition for non-surjectivity]
\label{cor:large_norm}
	Let $X$ and $Y$ be Banach spaces.
	Let $T \colon X \to Y$ be a bounded linear transformation.
	Suppose that there exist a sequence of vectors $y_n \in Y$ and
	a sequence of bounded linear functionals $f_n\colon Y \to \RR$ in $Y^*$ such that
	\begin{enumerate}[label={\rm(\roman*)}]
		\item \label{cor:large_norm:c1} $\norm[\big]{f_n\circ T}_X^* = \smallo(n)$ as $n\to\infty$,
		\item \label{cor:large_norm:c3} $\sup_{n} \norm{y_n}_Y < \infty$,
		\item \label{cor:large_norm:c2} $\abs[\big]{f_n(y_n)} = \Omega(n)$ as $n\to\infty$.
	\end{enumerate}
	Then, $T$ is not surjective.
\end{corollary}

\begin{proof}
	Without loss of generality, we can assume that $\sup_n\norm{y_n}_Y<1$.
	Let $f \in Y^*$.
	If $y = T(x)$ for some $x \in X$, then
	\begin{align}
		\abs{f(y)} &= \abs{f\circ T(x)} \leq \norm{f \circ T}_X^* \norm{x}_X
	\end{align}
	Thus, for all $y \in Y$,
	\begin{align}
	\label{eq:lower-bound}
		\frac{\abs{f(y)}}{\norm{f \circ T}_X^*} &\leq \inf_{x \in T^{-1}(y)} \norm{x}_X \;.
	\end{align}
	The left-hand side of \eqref{eq:lower-bound} tends to infinity by
	setting $f = f_n$ and $y = y_n$ and so there is no ball centered
	at the origin in $X$ whose image contains the ball~$\ball{Y}{1}$ of radius
	$1$ centered at 0 in $Y$. This fact contradicts
	condition~\ref{prop:Banach_equivalences:2} in Proposition~\ref{prop:Banach_equivalences}
	and so $T$ is not surjective.
\end{proof}

We will apply this result in section~\ref{sec:main} to show that there exists a
shift-invariant continuous cocycle on a full shift that cannot be represented
by a shift-invariant norm-summable interaction.  We will also apply Proposition~\ref{prop:Banach_equivalences} in section~\ref{sec:sullivan} to show that, for a large
collection of SFTs $\Omega$, every shift-invariant continuous cocycle on $\Omega$ can be represented by a shift-invariant variation-summable interaction.

\section{Shift-invariant norm-summable representations may not exist!}
\label{sec:main}

The purpose of this section is to prove that the map $\Phi\mapsto\psi_\Phi$ which assigns
to every norm-summable shift-invariant interaction a
continuous shift-invariant cocycle on a full shift is not surjective.
Before we proceed to the main result, we describe simpler and more
explicit examples on proper subshifts.  These results are based on the
``linear growth condition''.

Let $\Omega\subseteq\Sigma^{\Sites}$ be a configuration space.
We say that a cocycle $\psi$ on $\relation{T}(\Omega)$ satisfies the \emph{linear growth condition} if
there exists a constant $C\geq 0$ such that
\begin{align}
\label{eq:linear_growth_condition}
	\abs{\psi(x,y)} &\leq C\abs{A}
\end{align}
for every $A\Subset\Sites$ and $(x,y)\in\relation{T}_A(\Omega)$.

\begin{proposition}[Linear growth of $\psi_\Phi$]
\label{prop:linear_growth_condition}
	Let $\Omega$ be a shift space.
	Any shift-invariant continuous cocycle on $\relation{T}(\Omega)$ that is
	represented by a shift-invariant norm-summable interaction
	satisfies the linear growth condition.
\end{proposition}

\begin{proof}
	For $A\Subset\ZZ^d$ and $(x,y) \in \relation{T}_{A}(\Omega)$,
	\begin{align}
		\psi_\Phi(x,y) &= \sum_{\substack{F\Subset\ZZ^d\\ F\cap A\neq\varnothing}}
			\big[ \Phi_F(x)-\Phi_F(y) \big]
		= \sum_{k\in A}\sum_{\substack{F\Subset\ZZ^d\\ F\ni k}}\frac{1}{\abs{A\cap F}} \big[ \Phi_F(x)-\Phi_F(y) \big] \;.
	\end{align}
	By the triangle inequality,
	\begin{align}
		\abs{\psi_\Phi(x,y)} &\leq		
			\sum_{k \in A}\sum_{\substack{F\Subset\ZZ^d\\ F\ni k}}\frac{1}{\abs{A\cap F}}\big(\abs{\Phi_{F}(x)}+\abs{\Phi_F(y)}\big)
			\leq 2\sum_{k \in A}\sum_{\substack{F\Subset\ZZ^d\\ F\ni k}}\norm{\Phi_F} \;.
	\end{align}
	Since $\psi$ is shift-invariant, for any $k \in \ZZ^d$,
	\begin{align}
		\sum_{\substack{F\Subset\ZZ^d\\ F\ni k}}\norm{\Phi_F} &= \sum_{\substack{F\Subset\ZZ^d\\ F\ni 0}}\norm{\Phi_F} = \normNS{\Phi} \;.
	\end{align}
	Thus, $\abs{\psi_\Phi(x,y)}\leq 2\normNS{\Phi}\,\abs{A}$,
	which means $\psi_\Phi$ satisfies the linear growth condition.
\end{proof}

\subsection{The ``height difference'' cocycle on $3$-colorings}
\label{sec:main:3_colored_chessboard}

The following is a simple example of a shift-invariant continuous (in fact, Markov)
cocycle on a two-dimensional SFT, which is not given by a shift-invariant
norm-summable interaction. The example appears in~\cite{ChaMey16},
providing a Markovian specification which does not come from a nearest-neighbour interaction.

\begin{example}[$3$-colorings: absence of linear growth]
\label{ex:3_colored_chessboard}
	Let $\Omega = \Omega^2_{\coloring(3)}$ denote the two-dimensional $3$-coloring shift
	as in Example~\ref{ex:colorings},
	namely, the shift space consisting of all proper colorings of the standard Cayley graph
	of $\ZZ^2$,
	\begin{align}
		\Omega^2_{\coloring(3)} & \isdef
			\big\{ x \in \{\symb{0},\symb{1},\symb{2}\}^{\ZZ^2}:
				\text{$x_n \neq x_{n + e_i}$ for every $n \in \ZZ^2$ and $i=1,2$} \big\} \;,
	\end{align}
	where $e_1\isdef(1,0)$ and $e_2\isdef(0,1)$.
	As discussed in~\cite{ChaMey16},
	every $x \in \Omega^2_{\coloring(3)}$, can be ``lifted'' to a \emph{height function}
	$\hat{x} \in {\ZZ}^{\ZZ ^2}$, where
	\begin{enumerate}[label={(\roman*)}]
		\item $\hat{x}_n = x_n \pmod{3}$ for every $n \in \ZZ ^ 2$,
		\item $\abs{\hat x_n - x_{n +e_i}}=1$ for every $n \in \ZZ ^ 2$ and $i=1,2$.
	\end{enumerate}
	Such a ``lift'' is unique up to addition by an integer multiple of $3$.
	Furthermore, for any $(x,y) \in \relation{T}(\Omega)$ there is a pair of lifts
	$(\hat{x},\hat{y}) \in \ZZ^{\ZZ ^ 2}\times \ZZ^{\ZZ ^2}$
	as above such that $\hat{x}$ and $\hat{y}$ are asymptotic,
	and this pair of lifts is unique up to addition by a common integer multiple of $3$.
	For $(x,y) \in \relation{T}(\Omega)$ we can define
	\begin{align}
		\psi(x,y) & \isdef \sum_{n \in \ZZ^2}[\hat{x}_n -\hat{y}_n],
	\end{align}
	where $(\hat x, \hat y)$ is an asymptotic pair of lifts. This is a well defined shift-invariant
	continuous cocycle (in fact Markov), as shown in Section $4$ of~\cite{ChaMey16}.

	We claim that $\psi$ is not generated by any shift-invariant norm-summable interaction.
	To show this, we show that it violates the linear growth condition
	(see Proposition~\ref{prop:linear_growth_condition}).
	Indeed, for $i \in \NN$, define $\hat x^{(i)} , \hat y^{(i)}\in \ZZ^{\ZZ^2}$ by
	\begin{align}
		\hat{x}^{(i)}_n & \isdef
		\begin{cases}
			i - \norm{n}_1 				& \text{if $\norm{n}_1 \leq i$,} \\
			(i - \norm{n}_1) \bmod 2 	& \text{otherwise,}
		\end{cases}
	\end{align}
	and
	\begin{align}
		\hat y^{(i)}_n \isdef
		\begin{cases}
			-i + \norm{n}_1				& \text{if $\norm{n}_1 \leq i$,} \\
			(i - \norm{n}_1) \bmod 2	& \text{otherwise.}
		\end{cases}
	\end{align}
	Let $x^{(i)},y^{(i)} \in \Omega$ be given by $x^{(i)}_n \isdef \hat{x}^{(i)}_n \bmod 3$
	and $y^{(i)}_n = \hat{y}^{(i)}_n \bmod 3$. Then
	$(x^{(i)},y^{(i)}) \in \relation{T}_{B_i^{(2)}}(\Omega)$ where
	$B_i^{(2)} \isdef \{ n \in \ZZ^2: \norm{n}_1 \le i\}$.
	Also, a simple calculation shows that
	$\psi(x^{(i)},y^{(i)}) = \sum_{n \in \ZZ^2}[\hat{x}^{(i)}_n - \hat{y}^{(i)}_n]$
	is precisely equal to the cardinality of the set
	$B_i^{(3)} \isdef \{ n \in \ZZ^3: \norm{n}_1 \le i\}$.
	Thus
	\begin{align}
		\lim_{i \to \infty} \frac{\abs[\big]{\psi(\hat x^{(i)} , \hat y^{(i)})}}{\abs{B_i^{(2)}}} &= \infty \;,
	\end{align}
	and so \eqref{eq:linear_growth_condition} cannot hold for any constant $C\geq 0$.
	\hfill\exampleqed
\end{example}

The cocycle $\psi \colon \relation{T}(\Omega) \to \RR$ in the above example is indeed a very explicit
example of of shift-invariant continuous cocycle which is not given by a norm-summable interaction.
However, as shown in~\cite[Proposition~6.2]{ChaMey16},
every shift-invariant probability measure on $\Omega^2_{\coloring(3)}$ that is consistent
with the specification associated to $\psi$ is ``frozen'',
meaning that the asymptotic relation restricted to the support of the measure is equal to the diagonal relation.
Note that a frozen measure on a configuration space~$\Omega$
is consistent with \emph{any} specification on~$\Omega$.
Thus, the specification associated to $\psi$ should be considered as a pathological example
which does not ``genuinely'' specify any shift-invariant almost-Markovian random field.

\subsection{Almost Markovian random fields on the ``square-islands shift''}
\label{sec:main:square_islands}

The next example, again taken from~\cite{ChaMey16}, provides another shift-invariant
continuous cocycle that is not generated by any shift-invariant norm-summable interaction.
As in Example~\ref{ex:3_colored_chessboard}, the cocycle in the following example
is defined on a two-dimensional SFT.
However, unlike in Example~\ref{ex:3_colored_chessboard}, the specification associated to the cocycle
presented below admits consistent measures that have full support on the underlying SFT.

\begin{example}[Square-islands shift: absence of linear growth]
	Section $9$ of~\cite{ChaMey16} describes a certain infinite family $(\mu_p)_{p \in (0,1)^{\NN}}$
	of two-dimensional shift-invariant Markov random fields.
	This family is parametrized by  sequences $p= (p_j)_{j=1}^\infty \in (0,1)^{\NN}$.
	The support of these  Markov random fields is a certain SFT
	called ``the square-island shift'', whose admissible configurations consist
	of ``square islands'' in a ``sea'' of blank tiles. Each ``square island''
	has one of two colors (say red and blue), is square shaped (hence the name),
	and has positive integer ``radius'' $j \in \NN$. All the probability measures
	$\mu_p$ collapse to the same probability measure $\mu$ once the colors
	of the islands are ``forgotten''.
	This projection $\mu$ is consistent with the uniform specification $K^\circ$
	(see Example~\ref{exp:uniform-specification}).
	According to $\mu_p$, given the locations and sizes of the
	islands, the colors of the islands are independent,
	and each island of size $j$ is blue with probability $p_j$.
	The cocycle $M_p$ associated to the specification of $\mu_p$
	(which is a Markov cocycle, hence continuous) is described explicitly in Section~9 of~\cite{ChaMey16}.
	We claim that $M_p$ is not generated by any shift-invariant norm-summable interaction,
	again due to the fact that $M_p$ does not satisfy the linear growth condition
	(see Proposition~\ref{prop:linear_growth_condition}).
	
	Let $x^{(j)}$ and $y^{(j)}$ be two asymptotic configurations each corresponding to
	a unique island of size $j$ centered at the origin, surrounded by an infinite sea,
	one colored blue and the other red. Then,
	\begin{align}
		\abs[\big]{M_p(x^{(j)},y^{(j)})} &= \log(p_j)-\log(1-p_j).
	\end{align}
	Thus, if $p_j$ tends rapidly enough to $0$ or to $1$ (say if, $p_j\isdef e^{-j^3}$),
	then
	\begin{align}
		\lim_{j \to \infty} \frac{\abs[\big]{M_p(\hat x^{(j)} , \hat y^{(j)})}}{j^2} &= \infty \;,
	\end{align}
	and so $M_p$ is not representable by a shift-invariant norm-summable interaction.
	\hfill\exampleqed
\end{example}

\subsection{Non-surjectivity of the map $\Phi \mapsto \psi_\Phi$ restricted to $\banach{B}_{\NS}(\Omega)$ on a full shift}
\label{sec:main:non_surjectivity}

This section is devoted to the proof of Theorem~\ref{thm:Kozlov_is_annoying_INTRO}, which is our main result. In consequence of the identification between positive, almost-Markovian specifications and continuous cocycles (Proposition~\ref{prop:specification-vs-cocycle:TMP:continuity}), it suffices to prove the following equivalent statement.

\begin{theorem}[No general shift-invariant norm-summable representation]
\label{thm:Kozlov_is_annoying}
	There exists a continuous shift-invariant cocycle
	$\psi$ on the asymptotic relation of $\Omega \isdef \{\symb{0},\symb{1}\}^{\ZZ}$ for which
	there is no shift-invariant norm-summable interaction $\Phi$ such that
	$\psi = \psi_{\Phi}$.
\end{theorem}

In the examples presented in the previous subsections, the absence of
shift-invariant norm-summable interactions was due to the failure of the linear growth condition.
On a full shift, every shift-invariant continuous cocycle satisfies the linear growth condition,
and hence this approach would not work.

\begin{proposition}[Linear growth on full shifts]
\label{prop:full-shift:linear-growth-is-automatic}
Every shift-invariant continuous cocycle on a shift space with a safe symbol
(in particular, a full shift) satisfies the linear growth condition.
\end{proposition}

\begin{proof}
	Let $\Omega$ be a shift space that has a safe symbol~$\blank$.
	Let $A\Subset\ZZ^d$ be arbitrary and take $x,y\in\relation{T}_A(\Omega)$.
	By pivoting the sites in $A$ one by one from $x$ to the safe symbol
	and then pivoting them back to $y$, we find a sequence
	of pivots of length at most $2\abs{A}$ from $x$ to $y$.
	This shows that for any shift-invariant continuous cocycle $\psi$, we have
	$\abs{\psi(x,y)}\leq 2\abs{A}\normsull{\psi}$.
	
	More specifically, let $\psi$ be an arbitrary shift-invariant continuous cocycle.
	Let $k_1,k_2,\ldots,k_n$ be an enumeration of the elements of $A$.
	Set $x^{(0)}\isdef x$ and recursively define $x^{(i)}\isdef\zeta_{k_i} x^{(i-1)}$, that is,
	$x^{(i)}$ is obtained from $x^{(i-1)}$ by turning the symbol at site $k_i$ to~$\blank$.
	By the cocycle equation and the triangle inequality,
	\begin{align}
		\abs[\big]{\psi\big(x,x_{\ZZ^\setminus A}\lor\blank^A\big)} &=
			\abs[\bigg]{\sum_{i=1}^n \psi\big(x^{(i-1)},x^{(i)}\big)} \leq
			\sum_{i=1}^n \abs[\big]{\psi\big(x^{(i-1)},x^{(i)}\big)} \leq \abs{A}\normsull{\psi} \;.
	\end{align}
	Similarly,
	$\abs[\big]{\psi\big(y,y_{\ZZ^\setminus A}\lor\blank^A\big)} \leq \abs{A}\normsull{\psi}$.
	Observe that $x_{\ZZ^\setminus A}\lor\blank^A = y_{\ZZ^\setminus A}\lor\blank^A$.
	Putting these together and using again the cocycle equation and the triangle inequality,
	we get
	\begin{align}
		\abs[\big]{\psi(x,y)} &=
			\abs[\big]{
				\psi\big(y,y_{\ZZ^\setminus A}\lor\blank^A\big) -
				\psi\big(x,x_{\ZZ^\setminus A}\lor\blank^A\big)
			} \\
		&\leq
			\abs[\big]{\psi\big(x,x_{\ZZ^\setminus A}\lor\blank^A\big)} +
			\abs[\big]{\psi\big(y,y_{\ZZ^\setminus A}\lor\blank^A\big)}
			\leq 2\abs{A}\normsull{\psi} \;,
	\end{align}
	which proves the claim.
\end{proof}

\subsubsection{Proof strategy}

For the rest of this section, we set $\Sigma\isdef\{\symb{0},\symb{1}\}$ and
the configuration space will be the one-dimensional binary full shift $\Omega\isdef\Sigma^\ZZ= \{\symb{0},\symb{1}\}^{\ZZ}$.
Let $T\colon\banach{B}_\NS(\Omega) \to \banach{B}_{\Sull}(\Omega)$ denote the map defined by
\begin{align}
	T(\Phi) &\isdef \psi_\Phi \;,
\end{align}
where $\psi_\Phi$ is the cocycle given in~\eqref{eq:cocycle:interaction}.

Every asymptotic pair $(x,y) \in \relation{T}(\Omega)$ defines a bounded linear functional
$\langle x,y\rangle\in\banach{B}_{\Sull}^*$ given by the evaluation map
\begin{align}
	\langle x,y\rangle(\psi) &\isdef \psi(x,y) \qquad \text{for $\psi\in\banach{B}_{\Sull}$.}
\end{align}

By Corollary~\ref{cor:large_norm},
in order to prove Theorem~\ref{thm:Kozlov_is_annoying}, it suffices to
show the existence of a sequence of asymptotic pairs
$\{(x^{(k)},y^{(k)})\}_{k \in \NN}\in \relation{T}(\Omega)$, a sequence
$\{\psi_k\}_{k \in \NN}$ of shift-invariant continuous cocycles,
and a strictly increasing function $n\colon \NN\to \NN$ such that:
\begin{conditions}[Sufficient for non-surjectivity of $T$]\
\label{con:kozlov_is_annoying}
	\begin{enumerate}[label={(\roman*)}]
		\item \label{con:kozlov_is_annoying:dual-norm-small}
			$\normNSD[\big]{\langle x^{(k)},y^{(k)}\rangle\oo T} = \smallo(n(k))$ as $k\to\infty$,
		\item \label{con:kozlov_is_annoying:sullivan-norm-bounded}
			$\sup_k \normsull{\psi_k}<\infty$,
		\item \label{con:kozlov_is_annoying:evaluation-large}
			$\abs[\big]{\psi_k(x^{(k)},y^{(k)})} = \Omega(n(k))$ as $k\to\infty$.
	\end{enumerate}
\end{conditions}

It will be useful to establish a concrete formula for
$\normNSD{\langle x,y\rangle\oo T} $.
Given an asymptotic pair $(x,y) \in\relation{T}(\Omega)$
and a finite pattern $w\in\{\symb{0},\symb{1}\}^A$ with shape $A \Subset \ZZ$, let
\begin{align}
	\label{eq:count-diff-cocycle}
	\Delta_w(x,y) &\isdef
		\sum_{i\in\ZZ}\big[\indicator{w}((\sigma^i y)_A) - \indicator{w}((\sigma^i x)_A)\big]
\end{align}
denote the difference in the number of occurrences of $w$ in $x$ and $y$.
Note that $\Delta_w$ is a cocycle on~$\relation{T}(\Omega)$, in fact a cocycle
generated by a finite-range interaction.

\begin{proposition}[Formula for the dual $\NS$-norm]
\label{prop:normNSD_densities}
	Let $(x,y)\in\relation{T}(\Omega)$ be an asymptotic pair.
	Then, for $\eta \isdef \langle x,y\rangle\oo T\in\banach{B}_{\NS}^*$, we have
	\begin{align}
		\normNSD{\eta} = \sup_{A\Subset \ZZ} \frac{1}{\abs{A}} \sum_{w \in \Sigma^A} \abs{\Delta_w(x,y)} \;.
	\end{align}
\end{proposition}

To prove this proposition, we use the following lemma.

\begin{lemma}[Countable linear decomposition]
\label{lem:cocycles-from-interactions:decomposition}
	For every $\Phi\in\banach{B}_\NS(\Omega)$ and $(x,y)\in\relation{T}(\Omega)$, we have
	\begin{align}
		\psi_\Phi(x,y) &=
			\sum_{A\bmod\ZZ}\sum_{w\in\Sigma^A} \Phi(w)\Delta_w(x,y) \;,
	\end{align}
	where $A\bmod\ZZ$ is a shorthand to indicate summing over finite subsets of $\ZZ$ modulo shift.
\end{lemma}
\begin{proof}
	\begin{align}
		\psi_\Phi(x,y) & =
			\sum_{A\Subset\ZZ} \big[\Phi(y_A) - \Phi(x_A)\big] \\
		& =
			\sum_{A \bmod \ZZ} \sum_{i \in \ZZ}
			\big[\Phi((\sigma^i y)_A) - \Phi((\sigma^i x)_A)\big] \\
		& =
			\sum_{A \bmod \ZZ} \sum_{i \in \ZZ}
			\Bigg[
				\sum_{w \in \Sigma^A}\Phi(w)\indicator{w}((\sigma^i y)_A) -
				\sum_{w \in \Sigma^A}\Phi(w)\indicator{w}((\sigma^i x)_A)
			\Bigg] \\
		& =
			\sum_{A \bmod \ZZ} \sum_{w \in \Sigma^A}\Phi(w) \sum_{i \in \ZZ}
			\big[
				\indicator{w}((\sigma^i y)_A) - \indicator{w}((\sigma^i x)_A)
			\big] \\
		& =
			\sum_{A \bmod \ZZ} \sum_{w \in \Sigma^A}\Phi(w)\Delta_w(x,y) \;. \\
		& \qedhere
	\end{align}
\end{proof}

\begin{proof}[Proof of Proposition \ref{prop:normNSD_densities}]
	Applying Lemma~\ref{lem:cocycles-from-interactions:decomposition},
	for every interaction $\Phi\in\banach{B}_\NS(\Omega)$, we have
	\begin{align}
		\abs{\eta(\Phi)} &= \abs{\psi_\Phi(x,y)} \\
		& \leq
			\sum_{A\bmod\ZZ} \sum_{w \in \Sigma^A} \abs[\big]{\Phi(w)\Delta_w(x,y)} \\
		& \leq
			\sum_{A\bmod\ZZ} \big(\abs{A}\sup_{v \in \Sigma^A}\abs{\Phi(v)}\big)
		    \big(1/\abs{A}\big)\sum_{w \in \Sigma^A} \abs{\Delta_w(x,y)} \\
		& \leq
			\normNS{\Phi} \sup_{A\Subset \ZZ} \frac{1}{\abs{A}}\sum_{w \in \Sigma^A} \abs{\Delta_w(x,y)} \;.
	\end{align}
	Taking supremum over $\Phi$ with $\normNS{\Phi}=1$, we get
	\begin{align}
	\label{eq:normNSD_densities}
		\normNSD{\eta} &\leq \sup_{A\Subset \ZZ^d} \frac{1}{\abs{A}} \sum_{w \in \Sigma^A} \abs{\Delta_w(x,y)} \;.
	\end{align}
	
	For the reversed inequality, given a fixed $A \Subset \ZZ^d$,
	define a shift-invariant interaction %
	$\Phi$ by
	\begin{align}
		\Phi(w) &\isdef
			\begin{cases}
				\frac{1}{\abs{A}} \sgn\big(\Delta_w(x,y)\big)
					& \text{if $w \in \Sigma^{A+i}$ for some $i \in \ZZ$,}\\
				0	& \text{otherwise.}
			\end{cases}
	\end{align}
	Then Lemma~\ref{lem:cocycles-from-interactions:decomposition} implies that
	\begin{align}
		\eta(\Phi) &= \psi_\Phi(x,y) =
			\frac{1}{\abs{A}} \sum_{w\in\Sigma^A} \abs{\Delta_w(x,y)} \;.
	\end{align}
	Note that $\normNS{\Phi} = 1$.  Thus,
	\begin{align}
		\normNSD{\eta} &\geq  \frac{1}{\abs{A}}\sum_{w \in \Sigma^A} \abs{\Delta_w(x,y)} \;.
	\end{align}
	Taking supremum over $A\Subset\ZZ$, we obtain
	\begin{align}
		\normNSD{\eta} &\geq
			\sup_{A\Subset \ZZ^d} \frac{1}{\abs{A}} \sum_{w \in \Sigma^A} \abs{\Delta_w(x,y)} \;. \\
		& \qedhere
	\end{align}
\end{proof}

Let us give an informal explanation on how we shall
construct sequences $\{(x^{(k)},y^{(k)})\}_{k \in \NN}\in \relation{T}(\Omega)$ and $\{\psi_k\}_{k \in \NN}$ satisfying Conditions~\ref{con:kozlov_is_annoying}\ref{con:kozlov_is_annoying:dual-norm-small}--\ref{con:kozlov_is_annoying:evaluation-large}.
Recall that the Hamming distance of a pair of words $u,v \in \{\symb{0},\symb{1}\}^n$
is the number of positions $i$ such that $u_i\neq v_i$.  The Hamming distance of
two asymptotic configurations $x,y\in\Omega$ is defined similarly.

First, we shall show that for every $k\in\NN$ and sufficiently large $n=n(k)$,
there exist two words $u^{(k)},v^{(k)}\in\{\symb{0},\symb{1}\}^n$
such that
the configurations $x^{(k)},y^{(k)}\in\Omega$ obtained by padding
$u^{(k)}$ and $v^{(k)}$ with~$\symb{0}$s on both sides satisfy the following properties:
\begin{enumerate}[label={(\alph*)}]
	\item \label{property:u-v:no-large-overlap}
		No shift of $x^{(k)}$ or $y^{(k)}$ is close to either $x^{(k)}$ or $y^{(k)}$ in the Hamming distance.
	\item \label{property:close-frequencies}
		For every pattern $w$ whose shape has no more than $k$ elements,
		the number of occurrences of $w$ in $x^{(k)}$ and in $y^{(k)}$ are close.
\end{enumerate}
The proof is via a probabilistic argument.  Namely, we show that for $n$ large enough,
the probability that two words $\rv{u}$ and $\rv{v}$ chosen independently uniformly at random
from $\{\symb{0},\symb{1}\}^n$ satisfy both of the above properties is positive.

Next, we shall define a shift-invariant interaction $\Phi^{(k)}$
with the following properties:
\begin{enumerate}[label={(\alph*)}, resume]
	\item \label{property:interaction:range-n}
		$\Phi^{(k)}$ assigns non-zero values only to (translations of) words of length $n(k)$.
	\item \label{property:interaction:positive-if-close-to-u}
		$\Phi^{(k)}(w)>0$ if (a translation of) $w$ is close to $u^{(k)}$ in the Hamming distance.
	\item \label{property:interaction:negative-if-close-to-v}
		$\Phi^{(k)}(w)<0$ if (a translation of) $w$ is close to $v^{(k)}$ in the Hamming distance.
\end{enumerate}
We shall use $\Phi^{(k)}$ to define a shift-invariant continuous cocycle $\psi_k$.

As we shall see, Property~\ref{property:close-frequencies} will imply that
$\normNSD{(x^{(k)},y^{(k)})} = \smallo(n(k))$ as $k\to\infty$.
Using Properties~\ref{property:u-v:no-large-overlap} and
\ref{property:interaction:range-n}--\ref{property:interaction:negative-if-close-to-v},
we will show that $\normsull{\psi_k}$ is uniformly bounded and that
$\abs[\big]{\psi_k(x^{(k)},y^{(k)})}=\Omega(n(k))$ as $k\to\infty$.  %

Let us now proceed with the detailed proof.

\subsubsection{The probabilistic argument}

In this section, we establish the existence of words $u^{(k)},v^{(k)}$
with Properties~\ref{property:u-v:no-large-overlap} and~\ref{property:close-frequencies}.
We will use the following two well-known probabilistic inequalities.

\begin{proposition}[Chernoff--Hoeffding bound \cite{Hoe63}]
	Let $X_1,\dots, X_n$ be i.i.d.\ random variables with mean~$\mu$,
	and let $X\isdef\sum_{i = 1}^n X_i$. Then, for every $\delta\in(0,1)$,
	\begin{align}
		\xPr\big( X \leq (1-\delta)n\mu\big)
		&\leq \exp\left( \frac{-\delta^2n\mu}{2}\right) \;.
	\end{align}
\end{proposition}

The \emph{variation} of a function $f\colon\RR^n\to\RR$ on its $i$th variable
is defined as
\begin{align}
	\underset{\text{$x_j=y_j$ for $j\neq i$}}{\sup_{x_1,\ldots,x_n}\sup_{y_1,\ldots,y_n}}
	\abs[\big]{f(y_1,\ldots,y_n)-f(x_1,\ldots,x_n)} \;.
\end{align}

\begin{proposition}[McDiarmid's   bounded differences inequality \cite{McD89}]
	Let $X_1,\dots, X_n$ be independent random variables,
	and let $f\colon \RR^n \to \RR$ be a measurable function
	whose variation on the $i$th variable is bounded by $c_i$.
	Then, for every $\varepsilon>0$,
	\begin{align}
		\xPr\Big(
			\abs[\big]{f(X_1,\dots,X_n)-\xExp[f(X_1,\dots,X_n)]} \geq \varepsilon
		\Big)
		&\leq 2\exp\left( \frac{-2\varepsilon^2}{\sum_{i=1}^n c_i^2} \right) \;.
	\end{align}
\end{proposition}

We will denote the Hamming distance between two words $u,v\in\{\symb{0},\symb{1}\}^n$
by $\Ham(u,v)$.
Given a finite pattern $w\in\{\symb{0},\symb{1}\}^D$, a set $I\subseteq\ZZ$ and
an asymptotic pair $(x,y)\in\relation{T}(\Omega)$, define
\begin{align}
	\Delta_w^I(x,y) &\isdef
		\sum_{i\in I}
		\big[\indicator{w}((\sigma^i y)_D) - \indicator{w}((\sigma^i x)_D)\big] \;.
\end{align}
This is consistent with the notation $\Delta_w(x,y)$ introduced in~\eqref{eq:count-diff-cocycle},
with $\Delta_w^\ZZ(x,y)=\Delta_w(x,y)$.

\begin{lemma}[Existence of marker words]
	\label{lemma:tom}
	For every $\varepsilon>0$, $\delta\in (0,1)$ and $k \in \NN$,
	there exists $n(\varepsilon,\delta,k) \in \NN$ such that for all $n\geq n(\varepsilon,\delta,k)$,
	there exist two words $u,v \in \{\symb{0},\symb{1}\}^n$ such that
	if we let $x,y \in \{\symb{0},\symb{1}\}^{\ZZ}$
	be the configurations defined by
	\begin{align}
	\label{eq:tom_lemma1}
		x_j \isdef
		\begin{cases}
			u_j				& \textup{if $0 \leq j < n$,}\\
			\symb{0}		& \textup{otherwise,}
		\end{cases}
		\qquad \textnormal{and} \qquad
		y_j \isdef
		\begin{cases}
			v_j				& \textup{if $0 \leq j < n$,}\\
			\symb{0}		& \textup{otherwise,}
		\end{cases}
	\end{align}
	then the following properties hold:
	\begin{enumerate}[label={\textup{(\alph*)}}]
		\item \label{lemma:tom:HamDis}
			We have
			\begin{align}
				\Ham(u,v) &> \left(1-\delta \right)\frac{n}{2} \;, \label{eq:tom_lemma2:1}
			\end{align}
			and for every $j\in\ZZ\setminus\{0\}$,
			\begin{align}
				\Ham\big((\sigma^jx)_{[0,n-1]},u\big) &> \left(1-\delta \right)\frac{n}{2} \;, &
				\Ham\big((\sigma^jx)_{[0,n-1]},v\big) &> \left(1-\delta \right)\frac{n}{2} \;,
					\label{eq:tom_lemma2:2}   \\
				\Ham\big((\sigma^jy)_{[0,n-1]},u\big) &> \left(1-\delta \right)\frac{n}{2} \;, &
				\Ham\big((\sigma^jy)_{[0,n-1]},v\big) &> \left(1-\delta \right)\frac{n}{2} \;.
					\label{eq:tom_lemma2:3}
			\end{align}

		\item \label{lemma:tom:lowcounts}
			For every pattern $w\in\{\symb{0},\symb{1}\}^D$ whose shape satisfies
			$\abs{D}\leq k$ and $D\subseteq \kappa+[0,n-1]$ for some $\kappa\in\ZZ$,
			and every interval $I\subseteq\ZZ$, we have
			\begin{align}\label{eq:tom_lemma3}
				\abs[\big]{\Delta_w^{I}(x,y)} &< \varepsilon n \;.
			\end{align}
	\end{enumerate}
\end{lemma}

\begin{proof}
	Let $n \in \NN$ be a positive integer whose value we shall specify later.
	Let $\rv{x}$ and $\rv{y}$ be two random configurations from $\Omega$
	in which, the symbols $\rv{x}_i$ and $\rv{y}_i$ for $i\in[0,n-1]$ are
	independent with $(\nicefrac{1}{2},\nicefrac{1}{2})$-Bernoulli distribution,
	whereas $\rv{x}_i\isdef\rv{y}_i\isdef \symb{0}$ for $i\in\ZZ\setminus[0,n-1]$.
	Let $\rv{u} \isdef \rv{x}_{[0,n-1]}$ and $\rv{v}\isdef \rv{x}_{[0,n-1]}$.
	Note that if we define
	\begin{align}
		\indicator{\rv{x}_i = \rv{y}_i} \isdef
		\begin{cases}
			1 & \text{if $\rv{x}_i \neq \rv{y}_i$,}\\
			0 & \text{if $\rv{x}_i = \rv{y}_i$,}\\
		\end{cases}
	\end{align}
	then the random variables $\indicator{\rv{x}_i = \rv{y}_i}$ for $i\in[0,n-1]$
	are also independent with $(\nicefrac{1}{2},\nicefrac{1}{2})$-Bernoulli distribution.
	By definition,
	$\Ham(\rv{u},\rv{v})
	= \sum_{i \in [0,n-1]}\indicator{\rv{x}_i = \rv{y}_i}$.
	Noting that
	$\xExp[\indicator{\rv{x}_i = \rv{y}_i}] = \nicefrac{1}{2}$
	and using the Chernoff--Hoeffding bound yields
	\begin{align}
		\xPr\left( \Ham(\rv{u},\rv{v}) \leq (1-\delta)\frac{n}{2}\right)
		&\leq \exp\left( \frac{-\delta^2n}{4}\right) \;.
	\end{align}

	Similarly, for any fixed  $j \in \ZZ \setminus \{0\}$ each of the
	$n$-tuples $(\indicator{\rv{x}_{i+j} = \rv{x}_i})_{ i \in [0,n-1]}$,
	$(\indicator{\rv{x}_{i+j} = \rv{y}_i})_{ i \in [0,n-1]}$,
	$(\indicator{\rv{y}_{i+j} = \rv{x}_i})_{ i \in [0,n-1]}$
	and $(\indicator{\rv{y}_{i+j} = \rv{y}_i})_{ i \in [0,n-1]}$
	consists of independent $(\nicefrac{1}{2},\nicefrac{1}{2})$-Bernoulli random variables.	
	Thus, using the Chernoff--Hoeffding bound as before we get
	\begin{align}
		\label{eq:Chernoff1}
		\xPr\left( \Ham((\sigma^j\rv{x})_{[0,n-1]},\rv{u}) \leq (1-\delta)\frac{n}{2}\right)
		& \leq \exp\left( \frac{-\delta^2n}{4}\right),\\
		\label{eq:Chernoff2}
		\xPr\left( \Ham((\sigma^j\rv{x})_{[0,n-1]},\rv{v}) \leq (1-\delta)\frac{n}{2}\right)
		& \leq \exp\left( \frac{-\delta^2n}{4}\right),\\
		\label{eq:Chernoff3}
		\xPr\left( \Ham((\sigma^j\rv{y})_{[0,n-1]},\rv{u}) \leq (1-\delta)\frac{n}{2}\right)
		& \leq \exp\left( \frac{-\delta^2n}{4}\right),\\
		\label{eq:Chernoff4}
		\xPr\left( \Ham((\sigma^j\rv{y})_{[0,n-1]},\rv{v}) \leq (1-\delta)\frac{n}{2}\right)
		& \leq \exp\left( \frac{-\delta^2n}{4}\right).
	\end{align}
	By definition of $\rv{x}$ and $\rv{y}$,
	whenever $\abs{j}\geq n$,
	\begin{align}
		\Ham((\sigma^j\rv{x})_{[0,n-1]},\rv{u})
		& = \Ham((\sigma^j\rv{y})_{[0,n-1]},\rv{u})
		= \sum_{i \in [0,n-1]}\indicator{\rv{u}_i=\symb{1}},\\
		\Ham((\sigma^j\rv{x})_{[0,n-1]},\rv{v})
		& = \Ham((\sigma^j\rv{y})_{[0,n-1]},\rv{v})
		= \sum_{i \in [0,n-1]}\indicator{\rv{v}_i=\symb{1}}.
	\end{align}
	Therefore, in order to satisfy conditions~\eqref{eq:Chernoff1},~\eqref{eq:Chernoff2},~\eqref{eq:Chernoff3} and ~\eqref{eq:Chernoff4} for all $j \in \ZZ \setminus \{0\}$, it suffices to satisfy them for $j \in [-n,n] \setminus \{0\}$.
	Consequently, applying the union bound, we obtain that
	\begin{align}
	\label{eq:lemma:top:proof:bound:condition:a}
		\xPr(E_{\text{a}}) &\leq (8n+1)\exp\left( \frac{-\delta^2n}{4}\right) \;,
	\end{align}
	where $E_{\text{a}}$ is the event that $\rv{u}=u$ and $\rv{v}=v$ for some $u,v\in\{\symb{0},\symb{1}\}^n$
	that fail at least one of the conditions~\eqref{eq:tom_lemma2:1}--\eqref{eq:tom_lemma2:3}.
	
	Next, let $I \subseteq \ZZ$ be an interval and $w \in \{\symb{0},\symb{1}\}^{{D}}$ where $D \Subset \ZZ$ is a shape so that $\abs{D}\leq k$ and there is $\kappa \in \ZZ$ such that $D \subseteq \kappa+[0,n-1]$. Notice that if $i \notin [-\kappa-n+1,-\kappa+n-1]$ then $\indicator{w}((\sigma^i \rv{x})_D) = \indicator{w}((\sigma^i \rv{y})_D)$ because $(\sigma^i \rv{x},\sigma^i \rv{y})\in \relation{T}_{-i+[0,n-1]}(\Omega)$ and $-i+[0,n-1] \cap \kappa+[0,n-1] = \varnothing$. Hence if we let $I' = I \cap [-\kappa-n+1,-\kappa+n-1]$ we get that
	\begin{align}
	\Delta_{w}^I(\rv{x},\rv{y}) = \Delta_{w}^{I'}(\rv{x},\rv{y}).
	\end{align}
	Let us consider a partition $\mathcal{P}$ of $I'$ into $k$ parts such that whenever $i,j$ are two distinct elements of
	some $P \in \mathcal{P}$, $i+{D}$ is disjoint from $j+{D}$.
	(As $\abs{D}\leq k$, a greedy algorithm yields such a partition.) For $i\in I'$, define a random variable $\rv{W}^\rv{x}_i$ by
	\begin{align}
		\rv{W}^\rv{x}_{i}
		&\isdef \indicator{w}\big((\sigma^i\rv{x})_D\big) =
		\begin{cases}
			1 & \text{if $\rv{x}_{i+a} = w_a$ for every $a \in {D}$,}\\
			0 & \text{otherwise.}\\
		\end{cases}
	\end{align}	
	Define $\rv{W}^\rv{y}_i$ analogously
	and note that $\rv{W}^\rv{y}_i$ and
	$\rv{W}^\rv{x}_i$ are identically distributed.
	Note also that whenever $i,j$ are distinct elements of some $P \in \mathcal{P}$,
	the random variables $\rv{W}^\rv{x}_i$ and
	$\rv{W}^\rv{x}_j$ are independent, and the same holds if we replace $\rv{x}$ by $\rv{y}$.
	
	Let $i_1<i_2<\cdots<i_\ell$ be the elements of $I'$ ordered from left to right.
	For $P \in \mathcal{P}$, define $f_P\colon \{0,1\}^\ell \to \NN$
	by $f_P(x_1,x_2,\ldots,x_\ell) = \sum_{s: i_s \in P}x_s$ and note that
	changing the value of $(x_1,x_2,\ldots,x_\ell)$ at one coordinate modifies $f_P(x_1,x_2,\ldots,x_\ell)$
	by at most $1$. Let $f\colon \{0,1\}^\ell \to \NN$ be defined by
	$f(x_1,x_2,\ldots,x_\ell) \isdef \sum_{P \in \mathcal{P}} f_P(x_1,x_2,\ldots,x_\ell) = \sum_{i=1}^\ell x_i$.
	Observe that
	\begin{align}
		\Delta_w^{I}(\rv{x},\rv{y}) & =
		\sum_{i \in I'} \big[\indicator{w}((\sigma^i\rv{x})_D)-\indicator{w}((\sigma^i\rv{y})_D)\big] \\
		& = \sum_{i \in  I'} \big[\rv{W}^\rv{x}_{i}-\rv{W}^\rv{y}_{i}\big]\\
		& = f\big(\rv{W}^\rv{x}_{i_1},\dots,\rv{W}^\rv{x}_{i_{\ell}}\big)
		     -f\big(\rv{W}^\rv{y}_{i_1},\dots,\rv{W}^\rv{y}_{i_{\ell}}\big) \\
		& = \sum_{P \in \mathcal{P}}\Big[
			f_P\big(\rv{W}^\rv{x}_{i_1},\dots,\rv{W}^\rv{x}_{i_{\ell}}\big)
			-f_P\big(\rv{W}^\rv{y}_{i_1},\dots,\rv{W}^\rv{y}_{i_{\ell}}\big)
			\Big].
	\end{align}
	Note that $\abs{P} \leq \abs{I'} \leq 2n-1 \leq 2n$ for every $P \in \mathcal{P}$.
	Using McDiarmid's inequality we have
	\begin{align}
		\xPr\left(
			   \abs[\Big]{f_P\big(\rv{W}^\rv{x}_{i_1},\dots,\rv{W}^\rv{x}_{i_{\ell}}\big)
			   -\xExp\big[f_P\big(\rv{W}^\rv{x}_{i_1},\dots,\rv{W}^\rv{x}_{i_{\ell}}\big)\big]}
			   \geq \frac{\varepsilon n}{2k}
			\right)
		& \leq 2\exp\left(- \frac{2}{\abs{P}}\left(\frac{\varepsilon n}{2k}\right)^2  \right) \\
		&  \leq  2\exp\left( -\frac{\varepsilon^2 n}{4k^2} \right) \;.
	\end{align}
	Similarly,
	\begin{align}
		\xPr\left(
			\abs[\Big]{f_P\big(\rv{W}^\rv{y}_{i_1},\dots,\rv{W}^\rv{y}_{i_{\ell}}\big)
			-\xExp\big[f_P\big(\rv{W}^\rv{y}_{i_1},\dots,\rv{W}^\rv{y}_{i_{\ell}}\big)\big]}
			\geq 	\frac{\varepsilon n}{2k}
		\right)
		&\leq  2\exp\left(- \frac{\varepsilon^2 n}{4k^2} \right) \;.
	\end{align}
	Since
	$\xExp\big[f_P\big(\rv{W}^\rv{x}_{i_1},\dots,\rv{W}^\rv{x}_{i_{\ell}}\big)\big]
	= \xExp\big[f_P\big(\rv{W}^\rv{y}_{i_1},\dots,\rv{W}^\rv{y}_{i_{\ell}}\big)\big]$,
	we obtain, by applying the pigeonhole principle and the union bound, that
	\begin{align}
		\xPr\left(
			\abs[\Big]{f_P\big(\rv{W}^\rv{x}_{i_1},\dots,\rv{W}^\rv{x}_{i_{\ell}}\big)
			-f_P\big(\rv{W}^\rv{y}_{i_1},\dots,\rv{W}^\rv{y}_{i_{\ell}}\big)}
			\geq \frac{\varepsilon n}{k}
		\right)
		\leq 4 \exp\left( \frac{-\varepsilon^2 n}{4k^2} \right) \;.
	\end{align}	
	Using the fact that $f=\sum_{P\in\mathcal{P}} f_P$ and
	applying the pigeonhole principle and the union bound once more yields
	\begin{align}
		\xPr\left(
			\abs[\Big]{f\big(\rv{W}^\rv{x}_{i_1},\dots,\rv{W}^\rv{x}_{i_{\ell}}\big)
			-f\big(\rv{W}^\rv{y}_{i_1},\dots,\rv{W}^\rv{y}_{i_{\ell}}\big)}
			\geq \varepsilon n
		\right)
		& \leq  4k\exp\left( \frac{-\varepsilon^2 n}{4k^2} \right) \;.
	\end{align}
	Therefore
	\begin{align}
		\xPr\left(
			\abs[\big]{\Delta_w^{I}(\rv{x},\rv{y})} \geq \varepsilon n
		\right)
		\leq 4k\exp\left( \frac{-\varepsilon^2 n}{4k^2} \right).
	\end{align}
	
	At this point, we would like to use the latter bound to argue that the probability
	that condition~\eqref{eq:tom_lemma3} fails for some choices of $D$, $w$ and $I$
	is small.
	To this end, notice that if $D = \kappa+D'$ for some $\kappa\in \ZZ$, then, for any $I \subseteq \ZZ$, $w \in \{\symb{0},\symb{1}\}^{{D}}$ and $w' \in \{\symb{0},\symb{1}\}^{{D'}}$ so that $w(d)=w'(d-\kappa)$, we have 
	\begin{align}
			\abs[\big]{\Delta_w^{I}(x,y)} 
			&= \abs[\big]{\Delta_{w'}^{I+\kappa}(x,y)}
			\mbox{ for every } x,y \in \{\symb{0},\symb{1}\}^\ZZ \mbox{ such that } (x,\symb{0}^\ZZ),(y,\symb{0}^\ZZ) \in \relation{T}_{[0,n-1]}(\Omega).
	\end{align}
	Therefore the conditions for $D,D'$ that are the same modulo a shift are redundant and we may assume that $D \subseteq [0,n-1]$. 
	
	Fix some $D\subseteq [0,n-1]$ of cardinality $\abs{D} = m \leq k$. There are $2^{m}$ choices for $w \in \{\symb{0},\symb{1}\}^{{D}}$ and we already argued that we may assume that $I \subseteq [-n+1,n-1]$, therefore, there are $(2n-1)(n-1)$ choices for $I$. We get the following upper bound on the number $F(n,k)$ of conditions of the form
		$\abs[\big]{\Delta_w^{I}(\rv{x},\rv{y})} \geq \varepsilon n$ to impose:
		\begin{align}
			F(n,k) &
			\leq  \sum_{m =1}^{k}(2n-1)(n-1)2^{m}\binom{n}{m} \\
			& \leq 2n^22^k\sum_{m =1}^{k}\binom{n}{m} \\
			& \leq k2^{k+1}n^{k+2}.
		\end{align}

	Therefore, applying the union bound, we obtain
	\begin{align}
	\label{eq:lemma:top:proof:bound:condition:b}
		\xPr(E_{\text{b}}) &\leq F(n,k)\, 4k\exp\left( \frac{-\varepsilon^2 n}{4k^2} \right)\\
		& \leq  2^{k+3}k^2n^{k+2}\exp\left( \frac{-\varepsilon^2 n}{4k^2} \right) \;,
	\end{align}
	where $E_{\text{b}}$ is the event that $\rv{x}=x$ and $\rv{y}=y$ for some $x,y\in\{\symb{0},\symb{1}\}^\ZZ$
	that fail~\ref{lemma:tom:HamDis} for at least one choice of $I \Subset \ZZ$, $D$ with $\abs{D}\leq k$ so that there is $\kappa \in \ZZ$ such that $D \subseteq \kappa+[0,n-1]$ and
	$w\in\{\symb{0},\symb{1}\}^D$.

	Let us now choose $n = n(\varepsilon,\delta,k)$ sufficiently large so that
	\begin{align}
		(8n+1)\exp\left( \frac{-\delta^2n}{4}\right) +
		2^{k+3}k^2n^{k+2}\exp\left( \frac{-\varepsilon^2 n}{4k^2} \right)
		&< 1 \;.
	\end{align}
	Then, combining~\eqref{eq:lemma:top:proof:bound:condition:a} and~\eqref{eq:lemma:top:proof:bound:condition:b},
	we find that $\xPr(E_{\text{a}}\cup E_{\text{b}})<1$.
	In particular, with positive probability, $\rv{u}=u$ and $\rv{v}=v$ for some $u,v\in\{\symb{0},\symb{1}\}^n$
	that satisfy both conditions~\ref{lemma:tom:lowcounts} and~\ref{lemma:tom:HamDis}.
\end{proof}

\subsubsection{Proof of Theorem~\ref{thm:Kozlov_is_annoying}}

	We shall now define the objects that will satisfy Conditions~\ref{con:kozlov_is_annoying}\ref{con:kozlov_is_annoying:dual-norm-small}--\ref{con:kozlov_is_annoying:evaluation-large}.

	Fix $\delta\isdef\nicefrac{1}{2}$ and $K\isdef 16$.
	Given $k\in\NN$, set $\varepsilon_k\isdef\frac{1}{k^2 2^{k}}$,
	and pick an increasing function $n:\NN\to\NN$ such that
	$n(k)\geq n(\varepsilon_k,\delta,k)$, where $n(\varepsilon_k,\delta,k)$ is as in Lemma~\ref{lemma:tom}.
	Let $u^{(k)},v^{(k)}\in\{\symb{0},\symb{1}\}^{n(k)}$ be two words
	satisfying conditions~\eqref{eq:tom_lemma2:1}--\eqref{eq:tom_lemma2:3} and~\eqref{eq:tom_lemma3}
	as in Lemma~\ref{lemma:tom}.
	Let $x^{(k)},y^{(k)}\in\{\symb{0},\symb{1}\}^\ZZ$ be the configurations
	corresponding to $u^{(k)},v^{(k)}$ via~\eqref{eq:tom_lemma1}.

	For $k \in \NN$, let $\Phi^{(k)} \colon \{\symb{0},\symb{1}\}^{\ZZ} \to \RR$ be given by
	\begin{align}
		\Phi^{(k)}(x) &\isdef
			\begin{multlined}[t]
				\max \left\{
					0, n(k)-K \cdot \Ham\big(x_{[0,n(k)-1]} ,u^{(k)}\big)
				\right\} \\
				- \max \left\{
					0, n(k)-K \cdot \Ham\big(x_{[0,n(k)-1]} ,v^{(k)}\big)
				\right\} \;.
			\end{multlined}
	\end{align}
	Note that $\Phi^{(k)}(x)$ depends only on $x_{[0,n(k)-1]}$.
	Therefore, the function
	$\psi_k\colon \relation{T}(\Omega) \to \RR$ defined by
	\begin{align}
	\label{eq:cocycle:sequence}
		\psi_k(x,y) & \isdef \sum_{j \in \ZZ} \big[\Phi^{(k)}(\sigma^j y)-\Phi^{(k)}(\sigma^j x)\big],
	\end{align}
	is a cocycle on $\relation{T}(\Omega)$ which is continuous (in fact, local) and shift-invariant.

	\begin{claim}[Verification of Condition~\ref{con:kozlov_is_annoying}\ref{con:kozlov_is_annoying:dual-norm-small}]
		$\normNSD[\big]{(x^{(k)},y^{(k)})} = o\big(n(k)\big)$.
	\end{claim}
	
	\begin{proof}
		By Proposition~\ref{prop:normNSD_densities},
		it suffices to show that for every $A \Subset \ZZ$,
		\begin{align}
			\label{eq:BIGsuit}
			\frac{1}{\abs{A}}\sum_{w \in \Sigma^A} \abs[\big]{\Delta_w(x^{(k)},y^{(k)})}
			\leq \frac{6n(k)}{k}.
		\end{align}
		Indeed, if the above holds then,
		\begin{align}
			\lim_{k \to \infty}\frac{1}{n(k)} \normNSD[\big]{(x^{(k)},y^{(k)})}
			= \lim_{k \to \infty}\frac{1}{n(k)}\sup_{A \Subset \ZZ} \frac{1}{\abs{A}}
			   \sum_{w \in \Sigma^A} \abs[\big]{\Delta_w(x^{(k)},y^{(k)})}
			\leq \lim_{k \to \infty} \frac{6}{k}
			= 0.
		\end{align}
		For simplicity of notation,
		for the remainder of the proof of this claim,
		we shall denote $x \isdef x^{(k)}$, $y \isdef y^{(k)}$, %
		$n\isdef n(k)$ and $\varepsilon\isdef\varepsilon_k$.
		Fix $A = \{a_1 < a_2 < \dots < a_{\abs{A}}\} \Subset \ZZ$.
		Let us define,
		\begin{align}
			R & = \big\{j \in \ZZ : (j+A)\cap [0,n-1] \neq \varnothing  \big\}\\
			R_{\leq k} & = \big\{j \in \ZZ : 1\leq \abs[\big]{(j+A)\cap [0,n-1]}\leq k \big\}\\
			R_{> k} & = \big\{j \in \ZZ :  \abs[\big]{(j+A)\cap [0,n-1]}> k \big\}.
		\end{align}		
		Note that $R$ is the disjoint union of
		$R_{\leq k}$ and $R_{>k}$ and that for every
		$j \notin R$ and $w \in \{\symb{0},\symb{1}\}^A$, we have
		$\indicator{w}((\sigma^j x)_A)-\indicator{w}((\sigma^j y)_A) = 0$.
		We may thus write
		\begin{align}
			\frac{1}{\abs{A}}\sum_{w \in \Sigma^A} \abs{\Delta_w(x,y)}
			& = \frac{1}{\abs{A}}\sum_{w \in \Sigma^A}
				\abs[\Bigg]{\sum_{j \in \ZZ}\indicator{w}((\sigma^j x)_A)-\indicator{w}((\sigma^j y)_A)}\\
			& = \frac{1}{\abs{A}}\sum_{w \in \Sigma^A}
				\abs[\Bigg]{\sum_{j \in R}\indicator{w}((\sigma^j x)_A)-\indicator{w}((\sigma^j y)_A)}\\
			& \label{eq:smiley}
				\leq
				\frac{1}{\abs{A}}\sum_{w \in \Sigma^A}
				\abs[\Bigg]{
					\sum_{j \in R_{\leq k}}\indicator{w}((\sigma^j x)_A)-\indicator{w}((\sigma^j y)_A)
				} \\
			& \label{eq:frownie}
				+ \frac{1}{\abs{A}}\sum_{w \in \Sigma^A}
				\abs[\Bigg]{
					\sum_{j \in R_{> k}}\indicator{w}((\sigma^j x)_A)-\indicator{w}((\sigma^j y)_A)
				} \;.
		\end{align}	
		It suffices to bound both terms~\eqref{eq:smiley} and \eqref{eq:frownie}.
		In order to bound~\eqref{eq:frownie}, note that
		there is an injective function
		$\varphi\colon [1,\dots,k] \times R_{>k}\to [0,n-1] \times A$
		defined by $\varphi(i,j) \isdef \big(\ell(i,j),a(i,j)\big)$ where
		$\ell(i,j)$ is the $i$-th element from left to right of
		$(A+j)\cap [0,n-1]$ (there are at least $k+1$ elements by definition
		of $R_{>k}$) and $a(i,j)$ is $\ell(i,j)-j$. In particular,
		it follows by taking the cardinality of these sets that
		\begin{align}
			k\abs{R_{>k}} &\leq \abs{A}n.
		\end{align}		
		Using this inequality, we obtain the following bound for~\eqref{eq:frownie}:		
		\begin{align}
			\eqref{eq:frownie} & \leq \frac{1}{\abs{A}}\sum_{w \in \Sigma^A}
				\sum_{j \in R_{> k}}\abs[\big]{\indicator{w}((\sigma^j x)_A)-\indicator{w}((\sigma^j y)_A)}\\
			& = \frac{1}{\abs{A}} \sum_{j \in R_{> k}}
					\sum_{w \in \Sigma^A}
					\abs[\big]{
						\indicator{w}((\sigma^j x)_A)-\indicator{w}((\sigma^j y)_A)
					} \\
			& \leq \frac{1}{\abs{A}} \sum_{j \in R_{> k}}2 = \frac{\abs{R_{>k}}}{\abs{A}} \leq \frac{2n}{k}.
		\end{align}
		
		In order to bound~\eqref{eq:smiley}, we shall further
		divide it into three sums. For $w \in \{\symb{0},\symb{1}\}^A$
		let $B(w)$ be the intersection of $A$ with the
		convex hull of $w^{-1}(\symb{1})$ in $\ZZ$, that is $B(w) = [\min w^{-1}(\symb{1}), \max w^{-1}(\symb{1})] \cap A$. Noting that
		$B(w)=\varnothing$ if and only if $w = \symb{0}^A$,
		we can write,		
		\begin{align}
			\eqref{eq:smiley}
			& \label{eq:diamondsuit}
				= \frac{1}{\abs{A}}
				\abs[\Bigg]{
					\sum_{j \in R_{\leq k}}
					\big[\indicator{\symb{0}^A}((\sigma^j x)_A)-\indicator{\symb{0}^A}((\sigma^j y)_A)\big]
				} \\
			& \label{eq:spadesuit}
				+ \frac{1}{\abs{A}}
				\sum_{\substack{w \in \Sigma^A \\ 0 < \abs{B(w)}\leq k}}
				\abs[\Bigg]{
					\sum_{j \in R_{\leq k}}
					\big[\indicator{w}((\sigma^j x)_A)-\indicator{w}((\sigma^j y)_A)\big]
				}\\
			& \label{eq:heartsuit}
				+\frac{1}{\abs{A}}
				\sum_{\substack{w \in \Sigma^A \\ \abs{B(w)}> k}}
				\abs[\Bigg]{
					\sum_{j \in R_{\leq k}}
					\big[\indicator{w}((\sigma^j x)_A)-\indicator{w}((\sigma^j y)_A)\big]
				} \;.
		\end{align}
		Note that~\eqref{eq:heartsuit} is equal to $0$. Indeed, whenever
		$(B(w)+j) \cap (\ZZ\setminus [0,n-1]) \neq \varnothing$
		we have that $\indicator{w}((\sigma^j x)_A)=\indicator{w}((\sigma^j y)_A) = 0$.
		In particular, if $\abs{B(w)}>k$ we have that
		$\indicator{w}((\sigma^j x)_A)-\indicator{w}((\sigma^j y)_A)=0$ whenever
		$j \in R_{\leq k}$.
		
		Let us now give a bound for~\eqref{eq:spadesuit}. By the same argument
		as in the case of~\eqref{eq:heartsuit},	we may write,
		\begin{align}
			\eqref{eq:spadesuit} & = \frac{1}{\abs{A}}\sum_{\substack{w \in \Sigma^A \\ 0 < \abs{B(w)}\leq k}}
			\abs[\vast]{
				\sum_{\substack{j \in R_{\leq k} \\ (B(w)+j)\subseteq[0,n-1]}}
				\big[\indicator{w}((\sigma^j x)_A)-\indicator{w}((\sigma^j y)_A)\big]
			} \;.
		\end{align}
		For $j \in R$, let $A(j)=  A \cap ([0,n-1]-j)$ and note that
		it is an induced non-empty convex subset of $A$
		(the intersection of $A$ with an interval). Let us denote
		$j \simeq_A j'$ whenever $A(j)=A(j')$. It is easy to see
		that each equivalence class $[j]_{\simeq_A}$ is an interval.
		We claim that the set of all $j \in R_{\leq k}$ such that
		$(B(w)+j)\subseteq[0,n-1] \neq \varnothing$ can be covered by
		$2k-1$ disjoint intervals $[j]_{\simeq_A}$. Indeed, as
		$\varnothing \neq B(w) \subseteq[0,n-1] - j$, we have that
		$B(w) \cap A(j) \neq \varnothing$. As $j \in R_{\leq k}$
		we know that $\abs{A(j)}\leq k$. Finally, as $A(j)$ is an induced
		convex subset and $\abs{B(w)}\leq k$, there are at most
		$2k-\abs{B(w)}$ possible choices for $A(j)$ and hence $2k-1$
		disjoint equivalence classes which cover %
		$R_{\leq k}\cap \{j: B(w) + j \subseteq [0,n-1]\}$.
		Thus, we may write
		\begin{align}
			R_{\leq k} \cap \big\{j \in \ZZ : (B(w)+j)\subseteq[0,n-1] \big\}
			= [j_1]_{\simeq_A} \uplus [j_2]_{\simeq_A} \uplus
			\cdots \uplus [j_m]_{\simeq_A}
		\end{align}
		with $m \leq 2k-1$.
		Note that for each $i\in\{1,2,\ldots,m\}$, the restriction $w_{A(j)}$
		is the same for all $j\in[j_i]_{\simeq_A}$ and $\abs{A(j)}\leq k$.
		Therefore,
		\begin{align}
			\eqref{eq:spadesuit} & = \frac{1}{\abs{A}}\sum_{\substack{w \in \Sigma^A \\ 0 < \abs{B(w)}\leq k}}
			\abs[\Bigg]{
				\sum_{i =1}^m \sum_{j \in [j_i]_{\simeq_A}}
				\big[\indicator{w}((\sigma^j x)_A)-\indicator{w}((\sigma^j y)_A)\big]
			} \\
			& \leq \frac{1}{\abs{A}}\sum_{\substack{w \in \Sigma^A \\ 0 < \abs{B(w)}\leq k}}
			\sum_{i =1}^m
			\abs[\Bigg]{
				\sum_{j \in [j_i]_{\simeq_A}}
				\big[\indicator{w_{A(j_i)}}((\sigma^j x)_{A(j_i)})-\indicator{w_{A(j_i)}}((\sigma^j y)_{A(j_i)})\big]
			} \;.
		\end{align}
		Applying Item~\ref{lemma:tom:lowcounts} of Lemma~\ref{lemma:tom}
		with $D\isdef A(j_i)$ and $I \isdef [j_i]_{\simeq_A}$ we obtain that for every $i \in \{1,\dots,m\}$,
		\begin{align}
				\abs[\big]{\Delta^{I}_{w_{D}}(x,y)} = \abs[\Bigg]{
				\sum_{j \in [j_i]_{\simeq_A}}
				\indicator{w_{A(j_i)}}((\sigma^j x)_{A(j_i)})-\indicator{w_{A(j_i)}}((\sigma^j y)_{A(j_i)})
				}
			&\leq \varepsilon n.
		\end{align}
		Also note that in order to describe a word in $\{\symb{0},\symb{1}\}^A$ for which
		$\abs{B(w)}\leq k$, it suffices to select the left-most element of $B(w)$
		from $\abs{A}$ possible values and then the potentially non-zero values from among
		$2^{k-2}$ possibilities. It follows that
		\begin{align}
			\abs[\big]{\big\{ w \in \{\symb{0},\symb{1}\}^{A} : \abs{B(w)}\leq k \big\}} &\leq \abs{A} 2^{k-2} \leq \abs{A} 2^{k}.
		\end{align}		
		Putting the two equations above together, we get
		\begin{align}
			\eqref{eq:spadesuit} & \leq
				\frac{1}{\abs{A}}\sum_{\substack{w \in \Sigma^A \\ 0 < \abs{B(w)}\leq k}}
				\sum_{i =1}^m  \varepsilon n\\
			& \leq  \frac{1}{\abs{A}} 2^k\abs{A}(2k-1)\varepsilon n \\
			& \leq  2^k(2k-1)\varepsilon n = \frac{(2k-1) n}{k^2} \leq \frac{2n}{k} \\
		\end{align}
			
		It only remains to bound~\eqref{eq:diamondsuit}.
		In this case, as at most $k$ consecutive elements of $A$
		can intersect $[0,n-1]$ at the same time, we can cover
		$R_{\leq k}$ by at most $2k\abs{A}$ different intervals $[j]_{\simeq_A}$.
		Therefore, applying part~\ref{lemma:tom:lowcounts} of Lemma~\ref{lemma:tom} as before,
		we obtain the bound,
		\begin{align}
			\eqref{eq:diamondsuit} &
			\leq  \frac{1}{\abs{A}}2k\abs{A}\varepsilon n
			= 2k \varepsilon n \leq \frac{2n}{k}.
		\end{align}
		As the choice of $A$ was arbitrary, we obtain~\eqref{eq:BIGsuit}, which completes the proof of the claim.
	\end{proof}

	Next, we shall prove two technical propositions
	which will aid us in settling Conditions~\ref{con:kozlov_is_annoying}\ref{con:kozlov_is_annoying:sullivan-norm-bounded} and~\ref{con:kozlov_is_annoying}\ref{con:kozlov_is_annoying:evaluation-large}.

	\begin{proposition}
	\label{claim:ci1}
		Let $k \in \NN$ and $w \in \{\symb{0},\symb{1}\}^{n(k)}$. Then either
		\begin{align}
			\max\left\{ 0, n(k) - K \cdot \Ham\big(w,u^{(k)}\big) \right\} = 0
			\quad\text{or}\quad
			\max\left\{ 0, n(k) - K \cdot \Ham\big(w,v^{(k)}\big) \right\} = 0 \;.
		\end{align}
	\end{proposition}

	\begin{proof}
		This will follow from Item~\ref{lemma:tom:HamDis} of Lemma~\ref{lemma:tom}.
		Indeed, suppose that both quantities above are simultaneously positive.
		Then we would have,
		\begin{align}
			\frac{n(k)}{K} &>  \max\left\{\Ham\big(w, u^{(k)}\big), \Ham\big(w, v^{(k)}\big)\right\}
		\end{align}
		which implies
		\begin{align}
			\frac{2n(k)}{K} & > \Ham\big(w, u^{(k)}\big) + \Ham\big(w, v^{(k)}\big)
			\geq \Ham\big(u^{(k)}, v^{(k)}\big) > (1-\delta)\frac{n(k)}{2} =\frac{n(k)}{4} \;,
		\end{align}
		which implies that $8>K$. Since $K = 16$, this yields a contradiction.
	\end{proof}

	\begin{proposition}\label{claim:ci2}
		Let $z \in \{\symb{0},\symb{1}\}^{\ZZ}$ be a configuration and suppose that
		$\Phi^{(k)}(z) \neq 0$. Then there exists a ``safe interval''
		$I_{\mathtt{SAFE}} \isdef \left[-\frac{n(k)}{8},\frac{n(k)}{8}\right]$
		such that
		\begin{align}
			\Phi^{(k)}(\sigma^j z) = 0 \qquad\text{for every $j \in I_{\mathtt{SAFE}} \setminus\{0\}$.}
		\end{align}
	\end{proposition}

	\begin{proof}
		For brevity, let us set $n \isdef n(k)$. If $\Phi^{(k)}(z)\neq 0$, then we either have
		\begin{align}
		\label{eq:claim:ci2:case-1}
			\max\left\{ 0, n - K \cdot \Ham\big(z_{[0,n-1]}, u^{(k)}\big) \right\} &> 0
		\shortintertext{or}
			\max\left\{ 0, n - K \cdot \Ham\big(z_{[0,n-1]}, v^{(k)}\big) \right\} &> 0 \;.
		\end{align}
		The two cases are analogous,
		so without loss of generality, let us assume~\eqref{eq:claim:ci2:case-1}.
		This means
		\begin{align}
			\frac{n}{K} & > \Ham\big(z_{[0,n-1]}, u^{(k)}\big) = \Ham\big(z_{[0,n-1]}, x^{(k)}_{[0,n-1]}\big) \;.
		\end{align}
		It follows that for every $j\in\ZZ$,
		\begin{align}
			\Ham\big((\sigma^j z)_{[0,n-1]},(\sigma^j x^{(k)})_{[0,n-1]}\big) &< \frac{n}{K} + \abs{j} \;.
		\end{align}
		By the triangle inequality,
		\begin{align}
			\MoveEqLeft[3]\nonumber
			\Ham\big((\sigma^j z)_{[0,n-1]}, x^{(k)}_{[0,n-1]}\big) \\
			&\geq
				\Ham\big((\sigma^j x^{(k)})_{[0,n-1]}, x^{(k)}_{[0,n-1]}\big) -
				\Ham\big((\sigma^j z)_{[0,n-1]},(\sigma^j x^{(k)})_{[0,n-1]}\big) \;,\\
			\MoveEqLeft[3]\nonumber
			\Ham\big((\sigma^j z)_{[0,n-1]}, y^{(k)}_{[0,n-1]}\big) \\
			&\geq
				\Ham\big((\sigma^j x^{(k)})_{[0,n-1]}, y^{(k)}_{[0,n-1]}\big) -
				\Ham\big((\sigma^j z)_{[0,n-1]},(\sigma^j x^{(k)})_{[0,n-1]}\big) \;.
		\end{align}		
		By Item~\ref{lemma:tom:HamDis} of Lemma~\ref{lemma:tom}
		we know that for every non-zero $j \in \ZZ$,
		\begin{align}
			\Ham\big((\sigma^j x^{(k)})_{[0,n-1]},x^{(k)}_{[0,n-1]}\big) &\geq (1-\delta)\frac{n}{2} \;,\\
			\Ham\big((\sigma^j x^{(k)})_{[0,n-1]},y^{(k)}_{[0,n-1]}\big) &\geq (1-\delta)\frac{n}{2} \;.
		\end{align}
		Putting these bounds together and recalling that
		$u^{(k)} = x^{(k)}_{[0,n-1]}$ and $v^{(k)} = y^{(k)}_{[0,n-1]}$, we obtain
		\begin{align}
			\Ham\big((\sigma^j z)_{[0,n-1]}, u^{(k)}\big)
				&\geq (1-\delta) \frac{n}{2} - \frac{n}{K} - \abs{j} \;, \\
			\Ham\big((\sigma^j z)_{[0,n-1]},v^{(k)}\big)
				&\geq (1-\delta) \frac{n}{2} - \frac{n}{K} - \abs{j} \;.
		\end{align}
 		Note that $\max\left\{ 0, n - K \cdot \Ham\big((\sigma^jz)_{[0,n-1]}, u^{(k)}\big)\right\} = 0$
		if and only if $\Ham\big((\sigma^j z)_{[0,n-1]}, u^{(k)}\big) \geq \nicefrac{n}{K}$
		and that the same holds if we replace $u^{(k)}$ by $v^{(k)}$.
		It follows that $\Phi^{(k)}(\sigma^j z)=0$ for all non-zero $j\in\ZZ$ for which
		\begin{align}
			(1-\delta) \frac{n}{2} - \frac{n}{K} - \abs{j} &\geq \frac{n}{K} \;.
		\end{align}
		Plugging in $\delta = \nicefrac{1}{2}$ and $K = 16$
		and solving for $\abs{j}$ we get $\Phi^{(k)}(\sigma^j z)=0$ for all non-zero $j$ with
		\begin{align}
			\abs{j} &\leq \frac{n}{8} \;,
		\end{align}
		concluding the proof.
	\end{proof}

	\begin{claim}[Verification of Condition~\ref{con:kozlov_is_annoying}\ref{con:kozlov_is_annoying:sullivan-norm-bounded}]
		$\sup_k \normsull{\psi_k}<\infty$.
	\end{claim}

	\begin{proof}
		For brevity, let us write $n \isdef n(k)$.
		Let $(x,y) \in \relation{T}_{0}(\Omega)$.
		By the definition of $\psi_k$, we have
		\begin{align}
			\psi_k(x,y)
			&= \sum_{j \in \ZZ}\big[\Phi^{(k)}(\sigma^j y)-\Phi^{(k)}(\sigma^j x)\big]
			= \sum_{j \in [-n,n]}[\Phi^{(k)}(\sigma^j y)-\Phi^{(k)}(\sigma^j x)].
		\end{align}
		Since $(x,y)\in \relation{T}_{0}(\Omega)$, for every $j \in \ZZ$,
		$\Ham\big((\sigma^j x)_{[0,n-1]},(\sigma^j y)_{[0,n-1]}\big) \leq 1$.
		Consequently,
		\begin{align}
			\abs[\Big]{\Ham\big((\sigma^j x)_{[0,n-1]},u^{(k)}\big)-\Ham\big((\sigma^j y)_{[0,n-1]},u^{(k)}\big)}
				&\leq 1 \;, \\
			\abs[\Big]{\Ham\big((\sigma^j x)_{[0,n-1]},v^{(k)}\big)-\Ham\big((\sigma^j y)_{[0,n-1]},v^{(k)}\big)}
				&\leq 1 \;.
		\end{align}
		Therefore,
		\begin{align}
			\abs[\Big]{
				\max\left\{0,n- K\cdot \Ham\big((\sigma^j x)_{[0,n-1]},u^{(k)}\big) \right\}
				- \max\left\{0,n- K\cdot \Ham\big((\sigma^j y)_{[0,n-1]},u^{(k)}\big) \right\}
			}
			&\leq K \;,\\
			\abs[\Big]{
				\max\left\{0,n- K\cdot \Ham\big((\sigma^j x)_{[0,n-1]},v^{(k)}\big) \right\}
				-  \max\left\{0,n- K\cdot \Ham\big((\sigma^j y)_{[0,n-1]},v^{(k)}\big) \right\}
			}
			&\leq K \;,
		\end{align}
		which yields $\abs[\big]{\Phi^{(k)}(\sigma^j y)- \Phi^{(k)}(\sigma^j x)} \leq 2K$
		for every $j \in \ZZ$.
	
		By Proposition~\ref{claim:ci2}, if for some $j \in [-n,n]$
		we have $\Phi^{(k)}(\sigma^j x) \neq 0$, then for every
		$j' \in (j+[-\nicefrac{n}{8},\nicefrac{n}{8}])\setminus\{j\}$ we have
		that $\Phi^{(k)}(\sigma^{j'}x) = 0$. Consequently, if we let
		$I_x \isdef \{\ell \in [-n,n] : \Phi^{(k)}(\sigma^\ell(x)) \neq 0 \}$,
		then $\abs{I_x} \leq 17$. Similarly, if we let
		$I_y \isdef \{\ell \in [-n,n] : \Phi^{(k)}(\sigma^\ell(y)) \neq 0 \}$,
		then $\abs{I_y} \leq 17$. Consequently, we obtain
		\begin{align}
			\psi_k(x,y)
			& \leq \sum_{j \in I_x \cup I_y} \abs[\big]{\Phi_k(\sigma^j y)-\Phi_k(\sigma^j x)}
				\leq 2K(\abs{I_x}+\abs{I_y})
				\leq 68K
				= 544,
		\end{align}
		which proves the claim.
	\end{proof}

	\begin{claim}[Verification of Condition~\ref{con:kozlov_is_annoying}\ref{con:kozlov_is_annoying:evaluation-large}]
		$\abs[\big]{\psi_k(x^{(k)},y^{(k)})} = \Omega\big(n(k)\big)$.
	\end{claim}

	\begin{proof}
		We shall in fact show that	$\psi_k(x^{(k)},y^{(k)}) = -2n(k)$.
		As before, let us set $n \isdef n(k)$ for simplicity.
		We have
		\begin{align}
			\psi_k(x^{(k)},y^{(k)})
			& = \sum_{j \in \ZZ} \big[\Phi^{(k)}(\sigma^jy^{(k)}) - \Phi^{(k)}(\sigma^jx^{(k)})\big]\\
			& = \Phi^{(k)}(y^{(k)}) - \Phi^{(k)}(x^{(k)})
			       + \sum_{j \in \ZZ\setminus \{0\}}
			       	\big[\Phi^{(k)}(\sigma^jy^{(k)}) - \Phi^{(k)}(\sigma^jx^{(k)})\big]\\
			& = -2n(k) + \sum_{j \in \ZZ\setminus \{0\}}
				\big[\Phi^{(k)}(\sigma^jy^{(k)}) - \Phi^{(k)}(\sigma^jx^{(k)})\big] \;.
		\end{align}
		The last equality follows by Proposition~\ref{claim:ci1}
		and the definition of $x^{(k)}$ and $y^{(k)}$.
		Using Item~\ref{lemma:tom:HamDis} of Lemma~\ref{lemma:tom},
		we can deduce that for $j \in \ZZ\setminus \{0\}$,
		\begin{align}
			n-K \cdot \Ham\big((\sigma^jx^{(k)})_{[0,n-1]},u^{(k)}\big)
				&< n \Big(1 - \frac{K}{2}(1-\delta) \Big) \;,\\
			n-K \cdot \Ham\big((\sigma^jx^{(k)})_{[0,n-1]},v^{(k)}\big)
				&< n \Big(1 - \frac{K}{2}(1-\delta) \Big) \;,\\
			n-K \cdot \Ham\big((\sigma^jy^{(k)})_{[0,n-1]},u^{(k)}\big)
				&< n \Big(1 - \frac{K}{2}(1-\delta) \Big) \;,\\
			n-K \cdot \Ham\big((\sigma^jy^{(k)})_{[0,n-1]},v^{(k)}\big)
				&< n \Big(1 - \frac{K}{2}(1-\delta) \Big) \;.
		\end{align}
		Recalling $K = 16$ and $\delta = \nicefrac{1}{2}$, note that
		$1 - \frac{K}{2}(1-\delta)<0$. Therefore,
		\begin{align}
			\Phi^{(k)}(\sigma^jx^{(k)})
				& =
					\begin{multlined}[t]
						\max\left\{ 0, n-K\cdot\Ham\big((\sigma^jx^{(k)})_{[0,n-1]},u^{(k)}\big)  \right\} \\
			     		- \max\left\{ 0, n-K\cdot\Ham\big((\sigma^jx^{(k)})_{[0,n-1]},v^{(k)}\big)  \right\}
					\end{multlined} \\
				& = 0 \;,\\
			\Phi^{(k)}(\sigma^jy^{(k)})
				& = 
					\begin{multlined}[t]
						\max\left\{ 0, n-K\cdot\Ham\big((\sigma^jy^{(k)})_{[0,n-1]},u^{(k)}\big)  \right\} \\
						- \max\left\{ 0, n-K\cdot\Ham\big((\sigma^jy^{(k)})_{[0,n-1]},v^{(k)}\big)  \right\} 
					\end{multlined} \\
				& = 0 \;,
		\end{align}
		for each $j\in\ZZ\setminus\{0\}$, from which it follows that
		\begin{align}
			\sum_{j \in \ZZ\setminus \{0\}} \big[\Phi^{(k)}(\sigma^jy^{(k)}) - \Phi^{(k)}(\sigma^jx^{(k)})\big]
				&= 0.
		\end{align}
		Hence $\psi_k(x^{(k)},y^{(k)}) = -2n$.
	\end{proof}
	
	It follows that Conditions~\ref{con:kozlov_is_annoying}\ref{con:kozlov_is_annoying:dual-norm-small}--\ref{con:kozlov_is_annoying:evaluation-large} are satisfied, and therefore the map $T\colon\banach{B}_\NS(\Omega) \to \banach{B}_{\Sull}(\Omega)$ is not surjective. This completes the proof of Theorem~\ref{thm:Kozlov_is_annoying}.

\section{Non-shift-invariant norm-summable representations (Kozlov's theorem)}
\label{sec:main:Kozlov}

In this section we present two extended versions of Kozlov's theorem,
the first is for Markovian cocycles (Theorem~\ref{thm:markov-cocycle:interaction:non-invariant}) and the second for continuous cocycles (Theorem~\ref{thm:gibbs-cocycle:interaction:non-invariant}). The latter is equivalent to Theorem~\ref{thm:Kozlov2_INTRO} but stated in the formalism of continuous cocycles.
The proof of the latter is similar to that of the former but uses
approximations.  As mentioned in the introduction, Kozlov~\cite[two paragraphs before Theorem 3]{Koz77}
stated a similar result, but without proof.

\subsection{Finite-range interactions for Markov cocycles}
\label{sec:Kozlov:Markov}

The famous Hammersley--Clifford Theorem \cite{HamCli68,Spi71,Ave72} deals with Markovian specifications
compatible with a given locally-finite graph whose vertices are $\Sites$. Assuming the specification is \emph{strictly  positive},
it implies the existence of a compatible finite-range interaction, which is furthermore supported on cliques
of the corresponding graph. It is known that the positivity assumption in the Hammersley--Clifford Theorem
cannot be completely removed, although it can be somewhat relaxed (see for instance \cite{ChaMey16, Cha17} and the references therein). However, by forgetting about the graph structure,
we can prove a completely general statement about the existence of finite-range interactions for Markov cocycles
on any configuration space satisfying the TMP.

\begin{theorem}[Finite-range interactions for Markov cocycles]
	\label{thm:markov-cocycle:interaction:non-invariant}
	Let $\Omega\subseteq\Sigma^\Sites$ be a configuration space over a countable set of sites~$\Sites$
	and assume that $\Omega$ satisfies the TMP.
	Then, every Markov cocycle on $\relation{T}(\Omega)$
	is generated by a finite-range interaction.	
\end{theorem}
The proof of Theorem~\ref{thm:markov-cocycle:interaction:non-invariant} will be based on the following lemma.
\begin{lemma}[Partial extension]\label{lem:markov_cocycle_approx_potential}
	Let $\Omega \subseteq\Sigma^\Sites$ be a configuration space with the TMP.
	Let $A \Subset \Sites$ be a finite set of sites and
	$\psi_\ast \colon \relation{T}(\Omega) \to \RR$ a Markov cocycle on $\Omega$ such that
	$\psi_\ast(x,y)=0$ for every $(x,y)\in\relation{T}_A(\Omega)$.
	Then, for every $B \Subset \Sites$, there exists a finite-range interaction
	$\Phi$ such that
	\begin{enumerate}[label={\textup{(\roman*)}}]
		\item \label{item:markov_cocycle_approx_i} $\psi_{\Phi}(x,y)=\psi_\ast(x,y)$
		for every $(x,y)\in\relation{T}_B(\Omega)$, and
		\item \label{item:markov_cocycle_approx_ii} $\Phi_C=0$ whenever $C \cap A \ne \varnothing$.
	\end{enumerate}
\end{lemma}
\begin{proof}
	Let $D\supseteq A\cup B$ be a large enough finite set such that
	\begin{enumerate}[label={(\alph*)}]
		\item \label{cond:memory_markov:1}
			$D$ is a memory set for $A$ with respect to $\Omega$, witnessing the TMP of $\Omega$,
		\item \label{cond:memory_markov:2}
			$D$ is a memory set for $B$ with respect to $\psi_\ast$,
			witnessing the Markov property of $\psi_\ast$.
	\end{enumerate}
	Let us first argue that for $(x,y)\in\relation{T}_B(\Omega)$, the value of $\psi_\ast(x,y)$
	is uniquely determined by the restrictions of $x$ and $y$ to $D\setminus A$.
	Indeed, let $(x',y')\in\relation{T}_B(\Omega)$ be any other pair such that
	$x'_{D\setminus A}=x_{D\setminus A}$ and
	$y'_{D\setminus A}=y_{D\setminus A}$.  Define
	\begin{align}
		x^{\bullet} \isdef x_{A^\complement} \lor x'_{D}
		\qquad\text{and}\qquad
		& y^{\bullet} \isdef y_{A^\complement} \lor y'_{D} \;,
	\end{align}
	and note that $x^{\bullet},y^{\bullet}\in\Omega$ by property~\ref{cond:memory_markov:1}.
	Then,
	\begin{align}
		\psi_\ast(x',y') &=
			\psi_\ast(x^{\bullet},y^{\bullet}) \\
		&=
			\psi_\ast(x^{\bullet},x) + \psi_\ast(x,y) + \psi_\ast(y, y^{\bullet}) \\
		&=
			\psi_\ast(x,y) \;,
	\end{align}
	where the first equality is by property~\ref{cond:memory_markov:2}, the second is the cocycle equation,
	and the thirds is by assumption
	and the fact that $(x^{\bullet},x), (y, y^{\bullet})\in\relation{T}_A(\Omega)$.
		
	Consider now the equivalence relation $\sim$ on $\Lang_{D\setminus A}(\Omega)$
	defined by declaring $p\sim q$ whenever there exists a pair $(x,y)\in\relation{T}_B(\Omega)$
	such that $x_{D\setminus A}=p$ and $y_{D\setminus A}=q$.
	By the above discussion, $\psi_\ast$ induces a cocycle $\Delta$ on $\sim$, where
	\begin{align}
		\Delta(p,q) &\isdef \psi_\ast(x,y)
	\end{align}
	for some (and hence every) choice of $(x,y)\in\relation{T}_B(\Omega)$
	with $x_{D\setminus A}=p$ and $y_{D\setminus A}=q$.
	Since $\Delta$ is a cocycle on an equivalence relation on a finite set, it is generated by
	a potential $F:\Lang_{D\setminus A}(\Omega)\to\RR$ in the sense that
	\begin{align}
	\label{eq:cocycle:finite:potential}
		\Delta(p,q) &= F(q) - F(p)
	\end{align}
	for every $p,q\in\Lang_{D\setminus A}(\Omega)$ with $p\sim q$.
	Define an interaction $\Phi:\Lang(\Omega)\to\RR$ by
	\begin{align}
		\Phi(w) &\isdef
			\begin{cases}
				F(w)	& \text{if $w\in\Lang_{D\setminus A}(\Omega)$,} \\
				0		& \text{otherwise.}
			\end{cases}
	\end{align}
	By definition, $\Phi_C=0$ whenever $C\neq D\setminus A$, hence condition~\ref{item:markov_cocycle_approx_ii}
	is satisfied.
	Furthermore, $\psi_\Phi(x,y)=\Delta(x_{D\setminus A},y_{D\setminus A})=\psi_\ast(x,y)$
	for every $(x,y)\in\relation{T}_B(\Omega)$, thus condition~\ref{item:markov_cocycle_approx_i}
	is also satisfied.
\end{proof}

\begin{proof}[Proof of Theorem~\ref{thm:markov-cocycle:interaction:non-invariant}]
	Let $\psi$ be a Markov cocycle on $\relation{T}(\Omega)$.
	Pick an arbitrary co-final chain $ A_1\subsetneq A_2\subsetneq \cdots$
	of finite subsets of $\Sites$.
	We will inductively construct a sequence of finite-range interactions $(\Phi^{(n)})_{n=1}^\infty$
	such that
	\begin{enumerate}[label={(\alph*)}]
		\item \label{item:kozlov:markov:proof:I}
			$\psi_{\Phi^{(n)}}(x,y)=\psi(x,y)$ for every $(x,y) \in \relation{T}_{A_n}(\Omega)$, and
		\item \label{item:kozlov:markov:proof:II}
			$\Phi^{(m)}_C = \Phi^{(n)}_C$ for all $m \ge n$ and $C\Subset\Sites$
			such that $C \cap A_{n} \ne \varnothing$.
	\end{enumerate}
	To construct $\Phi^{(1)}$, apply Lemma~\ref{lem:markov_cocycle_approx_potential} to $\psi_\ast\isdef\psi$
	with $A\isdef\varnothing$ and $B\isdef A_1$.
	For $n >1$, assume that $\Phi^{(n-1)}$ has already been constructed as above.
	Apply Lemma~\ref{lem:markov_cocycle_approx_potential} to $\psi_\ast\isdef\psi - \psi_{\Phi^{(n-1)}}$
	with $A \isdef A_{n-1}$ and $B\isdef A_n$.
	Note that by the induction hypothesis $\psi_{\Phi^{(n-1)}}(x,y)=\psi(x,y)$ for every
	$(x,y) \in \relation{T}_{A_{n-1}}(\Omega)$,
	so indeed $\psi_\ast(x,y)=0$ for every  $(x,y) \in \relation{T}_{A_{n-1}}(\Omega)$.
	We thus obtain a finite-range interaction $\delta\Phi^{(n)}$ so that
	$\psi_{\delta\Phi^{(n)}}(x,y)=\psi(x,y)-\psi_{\Phi^{(n-1)}}(x,y)$
	for every $(x,y) \in \relation{T}_{A_n}(\Omega)$
	and $\delta\Phi^{(n)}_C =0$ whenever $C \cap A_{n-1}  \ne \varnothing$.
	Now let
	\begin{align}
		\Phi^{(n)} \isdef \Phi^{(n-1)} + \delta\Phi^{(n)}\;.
	\end{align}
	This completes the inductive construction of the sequence $(\Phi^{(n)})_{n=1}^\infty$.
	We now show that this sequence converges (pointswise) to a finite-range interaction $\Phi$
	that generates $\psi$.
	
	Let $C \Subset \Sites$ be arbitrary.  Since $ A_1\subsetneq A_2\subsetneq \cdots$ is co-final,
	$C\cap A_n\neq\varnothing$ for all sufficiently large~$n$.
	Hence, by property~\ref{item:kozlov:markov:proof:II} above, the sequence $(\Phi^{(n)}_C)_{n=1}^\infty$
	eventually stabilizes.  We define $\Phi_C$ as the eventual value of this sequence.
	In this fashion, we obtain an interaction $\Phi\isdef(\Phi_C)_{C\Subset\Sites}$.
	Let us verify that $\Phi$ is finite-range.
	Indeed, let $A\Subset\Sites$ be arbitrary.  Choose $n$ such that $A\subseteq A_n$.
	Then, by property~\ref{item:kozlov:markov:proof:II},
	$\Phi_C=\Phi^{(n)}_C$ for all $C\Subset\Sites$ such that $C\cap A\neq\varnothing$.
	Since $\Phi^{(n)}$ is finite-range,
	$\Phi^{(n)}_C\neq 0$ for no more than finitely many $C\Subset\Sites$ with $C\cap A\neq\varnothing$.
	It follows that $\Phi$ is finite-range.
	
	Lastly, let $(x,y)\in\relation{T}(\Omega)$.  Choose $n$ large enough
	such that $(x,y)\in\relation{T}_{A_n}(\Omega)$.
	Then, by~\ref{item:kozlov:markov:proof:I}, $\psi_{\Phi^{(n)}}(x,y)=\psi(x,y)$
	and by~\ref{item:kozlov:markov:proof:II}, $\psi_{\Phi}(x,y)=\psi_{\Phi^{(n)}}(x,y)$.
	We conclude that $\psi_\Phi=\psi$. 
\end{proof}

Recall that on a configuration space with the TMP, %
there is a simple bijective correspondence between strictly positive Markovian specifications %
and Markovian cocycles on the asymptotic relation %
(Section~\ref{sec:prelim:specification-cocycle:equivalence}).
Furthermore, only a configuration space that has the TMP admits positive Markovian specifications
(Proposition~\ref{prop:specification:continuous-positive:TMP}).
In general, every (not necessarily positive) Markovian specification has a well-defined
``support'', which carries all the information about the specification, and supports all the measures
consistent with that specification.
Namely, %
given a specification $K$ on a configuration space $\Omega$, define
\begin{align}\label{eq:supp_K_def}
	\supp(K) &\isdef
		\big\{z\in\Omega: \text{$K_A(z,[z_A])>0$ for every $A\Subset\Sites$}\big\} \;,
\end{align}
and call it the \emph{support} of~$K$.
The support of a specification $K$ has measure~$1$ with respect to every probability measure consistent with~$K$.
In particular, $\supp(K)$ is non-empty when $K$ is continuous.
When $K$ is Markovian, $\supp(K)$ is a closed subset of $\Omega$ that satisfies
the TMP, and $K$ induces a strictly positive (Markovian) specification on $\supp(K)$.
Theorem~\ref{thm:markov-cocycle:interaction:non-invariant} thus leads to the following characterization
of Markovian specifications.

\begin{corollary}[Gibbsian representation of arbitrary Markov specifications]\label{cor:Markov_spec_char}
 	A specification $K$ on a configuration space $\Omega$ is Markovian if and only if
 	$\supp(K)$ is a non-empty closed subset of $\Omega$ which has the TMP and the restriction of $K$
 	to $\supp(K)$ is given by a finite range interaction.
\end{corollary}

\begin{remark}[Restatement of Corollary~\ref{cor:Markov_spec_char}]
	The non-trivial direction of Corollary~\ref{cor:Markov_spec_char} can be rephrased as follows:
	every Markovian specification on a configuration space $\Omega$ is generated by a
	generalized interaction $\Phi=\Phi^\infty + \Phi^\finiterange$ on $\Omega$
	where $\Phi^\infty$ takes only values $0$ and $+\infty$,
	and $\Phi^\finiterange$ is a finite-range interaction.
	Here, ``generalized'' simply means that $\Phi$ is allowed to take value~$+\infty$.
	\hfill\remarkqed
\end{remark}

\subsection{Norm-summable interactions for continuous cocycles}
Our next goal is to prove the main result of this section.
\begin{theorem}[Norm-summable interactions for continuous cocycles]
	\label{thm:gibbs-cocycle:interaction:non-invariant}
	Let $\Omega\subseteq\Sigma^\Sites$ be a configuration space over a countable set of sites~$\Sites$
	and assume that $\Omega$ satisfies the TMP.
	Then, every continuous cocycle on
	$\relation{T}(\Omega)$ is generated by a norm-summable interaction.
\end{theorem}

Our proof of Theorem~\ref{thm:gibbs-cocycle:interaction:non-invariant} will be based on the following analog of Lemma~\ref{lem:markov_cocycle_approx_potential}:

\begin{lemma}[Approximate partial extension]
\label{lem:gibbs-cocycle:interaction:non-invariant:I+II}
	Let $\Omega \subseteq\Sigma^\Sites$ be a configuration space with the TMP.
	Let $A \Subset \Sites$ a finite set of sites and $\varepsilon >0$,
	and let $\psi_\ast \colon \relation{T}(\Omega) \to \RR$ be
	a continuous cocycle on $\Omega$ such  that
	$\abs{\psi(x,y)}<\varepsilon$ for every $(x,y)\in\relation{T}_A(\Omega)$.
	Then, for every $B \Subset \Sites$ %
	and $\delta>0$,
	there exists a finite-range interaction $\Phi$ such that
	\begin{enumerate}[label={\rm (\roman*)}]
		\item \label{item:continuous_cocycle_approx_i}
			$\abs[\big]{\psi_{\Phi}(x,y)-\psi_\ast(x,y)}
			<\delta$ for every $(x,y)\in\relation{T}_B(\Omega)$, and
		\item \label{item:continuous_cocycle_approx_ii}
			$\sum_{C:C\cap A\neq\varnothing}\norm{\Phi_C}<3\varepsilon$.
	\end{enumerate}
\end{lemma}
\begin{proof}
	We construct the desired interaction in two steps.
	First, we construct an interaction $\Phi^1$
	that $(3\varepsilon)$-approximates $\psi$ on $\relation{T}_B(\Omega)$
	and satisfies $\Phi^1_C=0$ whenever $C\cap A=\varnothing$.
	Then, we enhance the approximation to find an interaction of the form $\Phi=\Phi^1+\Phi^2$
	satisfying conditions~\ref{item:continuous_cocycle_approx_i} and~\ref{item:continuous_cocycle_approx_ii}.

	For the proof, we choose a canonical element $z\in[r]$ for each $r\in\Lang(\Omega)$.
	Without loss of generality, we assume that $A\subseteq B$.
	
	For the first step, let $D_1\supseteq B$ be a large enough finite set such that
	\begin{enumerate}[label={(1-\alph*)}]
		\item \label{cond:memory:1}
			$D_1$ is a memory set for $B$ with respect to $\Omega$, witnessing the TMP of $\Omega$,
		\item \label{cond:continuity:1}
			$\abs[\big]{\psi_\ast(x',y')-\psi_\ast(x,y)}<\varepsilon$
			for every $(x,y),(x',y')\in\relation{T}_B(\Omega)$ satisfying
			$x'_{D_1}=x_{D_1}$ and $y'_{D_1}=y_{D_1}$.
	\end{enumerate}
	As in the proof of Lemma~\ref{lem:markov_cocycle_approx_potential},
	consider the equivalence relation $\simI$ on $\Lang_{D_1\setminus A}(\Omega)$
	where $p\simI q$ if and only if $p = x_{D_1\setminus A}$ and $q = y_{D_1\setminus A}$
	for some $(x,y)\in\relation{T}_B(\Omega)$.
	We shall construct a cocycle $\Delta_1$ on $\simI$ such that
	\begin{align}
		\abs[\big]{\Delta_1(x_{D_1\setminus A},y_{D_1\setminus A}) - \psi_\ast(x,y)} &< 3\varepsilon
	\end{align}
	for every $(x,y)\in\relation{T}_B(\Omega)$.  Since $\Delta_1$ is a cocycle on an equivalence relation
	on a finite set, it is generated by a potential $F_1\colon\Lang_{D_1\setminus A}(\Omega)\to\RR$ in the sense
	of~\eqref{eq:cocycle:finite:potential}.  Define $\Phi^1:\Lang(\Omega)\to\RR$ by
	\begin{align}
		\Phi^1(w) &\isdef
			\begin{cases}
				F_1(w)	& \text{if $w\in\Lang_{D_1\setminus A}(\Omega)$,} \\
				0		& \text{otherwise.}
			\end{cases}
	\end{align}
	Clearly, $\Phi^1_C=0$ unless $C=D_1\setminus A$, and
	\begin{align}
		\abs[\big]{\psi_{\Phi^1}(x,y) - \psi_\ast(x,y)} &=
			\abs[\big]{\Delta_1(x_{D_1\setminus A},y_{D_1\setminus A}) - \psi_\ast(x,y)} < 3\varepsilon
	\end{align}
	for every $(x,y)\in\relation{T}_B(\Omega)$.
	
	Let us now construct $\Delta_1$.  Given $p,q\in\Lang_{D_1\setminus A}(\Omega)$, define
	\begin{align}
		\Delta_1(p,q) &\isdef \psi_\ast(z_{D_1^\complement} \lor p\lor u_A, z_{D_1^\complement}\lor q\lor v_A)
	\end{align}
	where $z$ is the canonical element of $[p_{D_1\setminus B}]=[q_{D_1\setminus B}]$,
	and $u$ and $v$ are respectively the canonical elements of $[p_{D_1\setminus A}]$ and $[q_{D_1\setminus A}]$.
	That $z_{D_1^\complement} \lor p\lor u_A$ and $r_{D_1^\complement}\lor q\lor v_A$ belong to $\Omega$
	is guaranteed by property~\ref{cond:memory:1}.
	Now, let $(x,y)\in\relation{T}(\Omega)$ be such that $x_{D_1\setminus A}=p$ and $y_{D_1\setminus A}=q$.
	From the cocycle equation and the triangle inequality, we get
	\begin{align}
		\abs[\big]{\Delta_1(p,q) - \psi_\ast(x,y)} &=
			\abs[\big]{\psi_\ast(z_{D_1^\complement} \lor p\lor u_A, z_{D_1^\complement}\lor q\lor v_A)
				- \psi_\ast(x,y)} \\
		&\leq
			\begin{multlined}[t]
				\abs[\big]{\psi_\ast(z_{D_1^\complement}\lor x_{D_1}, z_{D_1^\complement}\lor y_{D_1}) - \psi_\ast(x,y)} \\
				+ \abs[\big]{\psi_\ast(z_{D_1^\complement}\lor x_{D_1}, z_{D_1^\complement}\lor p\lor u_A)}
				+ \abs[\big]{\psi_\ast(z_{D_1^\complement}\lor y_{D_1}, z_{D_1^\complement}\lor q\lor v)}
			\end{multlined} \\
		&< 3\varepsilon \;.
	\end{align}
	The last inequality is by property~\ref{cond:continuity:1} and the hypothesis of the lemma.
	The fact that $z_{D_1^\complement}\lor x_{D_1}, z_{D_1^\complement}\lor y_{D_1}\in\Omega$
	is again by property~\ref{cond:memory:1}.

	For the second step, let $D_2\supseteq B$ be a large enough finite set such that
	\begin{enumerate}[label={(2-\alph*)}]
		\item \label{cond:memory:2}
			$D_2$ is a memory set for $B$ with respect to $\Omega$, witnessing the TMP of $\Omega$,
		\item \label{cond:continuity:2}
			for every $(x,y),(x',y')\in\relation{T}_B(\Omega)$ satisfying $x'_{D_2}=x_{D_2}$
			and $y'_{D_2}=y_{D_2}$,	we have
			\begin{align}
				\abs[\big]{(\psi_\ast-\psi_{\Phi^1})(x',y')-(\psi_\ast-\psi_{\Phi^1})(x,y)} &< \delta \;.
			\end{align}
	\end{enumerate}
	This time consider the equivalence relation $\simII$ on $\Lang_{D_2}(\Omega)$
	where $p\simII q$ if and only if $p = x_{D_2}$ and $q = y_{D_2}$ for some $(x,y)\in\relation{T}_B(\Omega)$.
	For $p,q\in\Lang_{D_2}(\Omega)$ satisfying $p\simII q$, define
	\begin{align}
		\Delta_2(p,q) &\isdef (\psi_\ast-\psi_{\Phi^1})(z_{D_2^\complement}\lor p, z_{D_2^\complement}\lor q) \;,
	\end{align}
	where $z$ is the canonical element of $[p_{D_2\setminus B}]=[q_{D_2\setminus B}]$.
	That $z_{D_2^\complement}\lor p, z_{D_2^\complement}\lor q\in \Omega$ is by property~\ref{cond:memory:2}.
	Clearly, $\Delta_2$ is a cocycle on $\simII$, and for every $(x,y)\in\relation{T}_B(\Omega)$,
	\begin{align}
		\abs[\big]{\Delta_2(x_{D_2},y_{D_2}) - (\psi_\ast-\psi_{\Phi^1})(x,y)} &< \delta
	\end{align}
	by property~\ref{cond:continuity:2}.
	Let $F_2\colon\Lang_{D_2}(\Omega)\to\RR$ be a potential generating $\Delta_2$
	in the sense of~\eqref{eq:cocycle:finite:potential}.  Note that we can choose $F_2$
	in such a way that
	\begin{align}
		\sup_{p}\abs{F_2(p)} &\leq \sup_{(p,q)}\abs{\Delta_2(p,q)} < 3\varepsilon \;.
	\end{align}
	Define $\Phi^2:\Lang(\Omega)\to\RR$ by
	\begin{align}
		\Phi^2(w) &\isdef
			\begin{cases}
				F_2(w)	& \text{if $w\in\Lang_{D_2}(\Omega)$,} \\
				0		& \text{otherwise,}
			\end{cases}
	\end{align}
	Clearly, $\norm{\Phi^2_C}<3\varepsilon$ when $C=D_2$ and $\norm{\Phi^2_C}=0$ otherwise.
	Furthermore,
	\begin{align}
		\abs[\big]{\psi_{\Phi^2}(x,y) - (\psi_\ast-\psi_{\Phi^1})(x,y)} &=
			\abs[\big]{\Delta_2(x_{D_1},y_{D_1})
				- (\psi_\ast-\psi_{\Phi^1})(x,y)} < \delta
	\end{align}
	for every $(x,y)\in\relation{T}_B(\Omega)$.
	
	The interaction $\Phi\isdef\Phi^1 + \Phi^2$ satisfies conditions~\ref{item:continuous_cocycle_approx_i}
	and~\ref{item:continuous_cocycle_approx_ii}.	
\end{proof}

The next proof is analogous to that of Theorem~\ref{thm:markov-cocycle:interaction:non-invariant}, with Lemma~\ref{lem:gibbs-cocycle:interaction:non-invariant:I+II} replacing Lemma~\ref{lem:markov_cocycle_approx_potential}.

\begin{proof}[Proof of Theorem~\ref{thm:gibbs-cocycle:interaction:non-invariant}]
	Let $\psi$ be a continuous cocycle on $\relation{T}(\Omega)$.
	Pick an arbitrary co-final chain $ A_1\subsetneq A_2\subsetneq \cdots$
	of finite subsets of $\Sites$, and a decreasing sequence $(\varepsilon_n)_{n=1}^\infty$
	of positive real numbers such that $\sum_n\varepsilon_n<\infty$.
	We will inductively construct a sequence of finite-range interactions $(\Phi^{(n)})_{n=1}^\infty$
	such that
	\begin{enumerate}[label={(\alph*)}]
		\item \label{item:kozlov:proof:I}
			$\abs[\big]{\psi_{\Phi^{(n)}}(x,y) - \psi(x,y)}<\varepsilon_n$
			for every $(x,y) \in \relation{T}_{A_n}(\Omega)$, and
		\item \label{item:kozlov:proof:II}
			$\sum_{C:C\cap A_{n-1}}\norm{\Phi^{(n)}_C - \Phi^{(n-1)}_C}\leq 3\varepsilon_n$.
	\end{enumerate}
	To construct $\Phi^{(1)}$, apply Lemma~\ref{lem:gibbs-cocycle:interaction:non-invariant:I+II}
	to $\psi_\ast\isdef\psi$
	with $A\isdef\varnothing$, $B\isdef A_1$, $\varepsilon$ arbitrary and $\delta\isdef\varepsilon_1$.
	For $n >1$, assume that $\Phi^{(n-1)}$ has already been constructed as above.
	Apply Lemma~\ref{lem:gibbs-cocycle:interaction:non-invariant:I+II} to
	$\psi_\ast\isdef\psi - \psi_{\Phi^{(n-1)}}$
	with $A \isdef A_{n-1}$, $B\isdef A_n$, $\varepsilon\isdef\varepsilon_{n-1}$ and $\delta\isdef\varepsilon_n$.
	Note that by the induction hypothesis
	$\abs[\big]{\psi_\ast(x,y)}=\abs[\big]{\psi_{\Phi^{(n-1)}}(x,y)-\psi(x,y)}<\varepsilon_{n-1}$
	for every $(x,y) \in \relation{T}_{A_{n-1}}(\Omega)$,
	We thus obtain a finite-range interaction $\delta\Phi^{(n)}$ so that
	$\abs[\big]{\psi_{\delta\Phi^{(n)}}(x,y)- (\psi(x,y)-\psi_{\Phi^{(n-1)}})(x,y)}<\varepsilon_n$
	for every $(x,y) \in \relation{T}_{A_n}(\Omega)$
	and $\sum_{C:C\cap A_{n-1}}\norm{\delta\Phi^{(n)}_C}\leq 3\varepsilon_{n-1}$.
	Now let
	\begin{align}
		\Phi^{(n)} \isdef \Phi^{(n-1)} + \delta\Phi^{(n)}\;.
	\end{align}
	This completes the inductive construction of the sequence $(\Phi^{(n)})_{n=1}^\infty$.
	We now verify that this sequence converges (pointswise) to a norm-summable interaction $\Phi$
	that generates $\psi$.
	
	Let $C \Subset \Sites$ be arbitrary.  Since $ A_1\subsetneq A_2\subsetneq \cdots$ is co-final,
	$C\cap A_n\neq\varnothing$ for all sufficiently large~$n$.
	Hence, by property~\ref{item:kozlov:proof:II} above, the sequence $(\Phi^{(n)}_C)_{n=1}^\infty$
	is Cauchy and thus converges.  We define $\Phi_C$ as the limit of this sequence.
	In this fashion, we obtain an interaction $\Phi\isdef(\Phi_C)_{C\Subset\Sites}$.
	Let us verify that $\Phi$ is norm-summable.
	Indeed, let $A\Subset\Sites$ be arbitrary.  Choose $n$ such that $A\subseteq A_n$.
	Then, by property~\ref{item:kozlov:proof:II},
	\begin{align}
		\sum_{C:C\cap A\neq\varnothing}\norm{\Phi_C} &\leq
			\sum_{C:C\cap A\neq\varnothing} \bigg(
				\norm[\big]{\Phi^{(1)}_C}	+ \sum_{n>1}\norm[\big]{\Phi^{(n)}_C - \Phi^{(n-1)}_C}
			\bigg) \\
		&\leq
			\sum_{C:C\cap A\neq\varnothing}\norm[\big]{\Phi^{(1)}_C} + 3\sum_{n>1}\varepsilon_n
		< \infty \;.
	\end{align}
	Lastly, let $(x,y)\in\relation{T}(\Omega)$.  Choose $n$ large enough
	such that $(x,y)\in\relation{T}_{A_n}(\Omega)$.
	Then, by~\ref{item:kozlov:proof:I}, $\psi_{\Phi^{(n)}}(x,y)$ converges to $\psi(x,y)$
	and by~\ref{item:kozlov:proof:II}, $\psi_{\Phi^{(n)}}(x,y)$ converges to $\psi_{\Phi}(x,y)$.
	We conclude that $\psi_\Phi=\psi$. 
\end{proof}

\begin{remark}
	Unlike the Markovian case,
	we do not get a complete characterization of (not necessarily positive) continuous specifications similar to
	\Cref{cor:Markov_spec_char}. This is is because when $K$ is merely a continuous specification on a
	configuration space $\Omega$, the set $\supp(K)$ given by \eqref{eq:supp_K_def} might not be closed.
	\hfill\remarkqed
\end{remark}

\section{Shift-invariant variation-summable representations (Sullivan's theorem)}
\label{sec:sullivan}
In this section we provide a proof of Theorem~\ref{thm:sull2_INTRO}, that is, of Sullivan's theorem on the existence of shift-invariant, variation-summable interactions which represent shift-invariant almost-Markovian specifications. As in previous sections we shall prove the equivalent statement in terms of continuous cocycles (Theorem~\ref{thm:sull2}).
We need the following definition to state the result:

\begin{definition}[Single-site fillability]
	An SFT $\Omega \subseteq\Sigma^{\ZZ^d}$ is called \emph{single-site fillable (SSF)} if there exists
	a finite set of forbidden finite patterns $\mathcal{F}$ defining $\Omega$ such that
	for every $A\Subset\ZZ^d$ and $k \in \ZZ^d \setminus A$, and every pattern $p$ with shape $A$
	that is locally admissible with respect to $\mathcal{F}$, there exists pattern $q$ with shape
	$A\cup\{k\}$ which is locally admissible with respect to $\mathcal{F}$ and such that $q_A = p$.
\end{definition}

An SFT is single-site fillable
if and only if it has a defining finite set of forbidden finite patterns with respect to which every locally-admissible finite pattern is (globally) admissible. Obviously, the full-shift is single-site fillable. Here are some less trivial examples:

\begin{example}[Hard-core shift]\label{ex:hard-core2}
	The hard-core shift $\Omega_\hardcore$ is single-site fillable.
	In fact, every SFT with a safe symbol is single-site fillable.
	\hfill\exampleqed
\end{example}

\begin{example}[$q$-coloring shift]\label{ex:coloring2}
	The shift $\Omega^d_{\coloring(q)}$ consisting of all $q$-colorings of $\ZZ^d$ is single-site fillable
	when $q\geq 2d+1$. Note that this shift does not have a safe symbol.
	\hfill\exampleqed
\end{example}

\begin{theorem}[Shift-invariant variation-summable representation]
\label{thm:sull2}
	Let $\Omega\subseteq\Sigma^{\ZZ^d}$ be an SFT which is single-site fillable and has the pivot property.
	Then, every continuous and shift-invariant cocycle on $\relation{T}(\Omega)$ is generated by
	a shift-invariant variation-summable interaction.
\end{theorem}

By Examples~\ref{ex:hard-core1} and~\ref{ex:hard-core2}, the hard-core shift $\Omega_\hardcore$
has the pivot property and is single-site fillable.
Similarly, by Examples~\ref{ex:colorings} and~\ref{ex:coloring2},
the shift $\Omega^d_{\coloring(q)}$ of $q$-colorings of $\ZZ^d$ has the pivot property and
is single-site fillable provided that $q \geq 2d+2$.
Therefore, Theorem~\ref{thm:sull2} applies to both these examples.

An equivalent way to state~\Cref{thm:sull2} is to say that whenever $\Omega$ is single-site fillable and satisfies the pivot property, then the map $\Phi \mapsto \psi_\Phi$ from $\banach{B}_{\VS}(\Omega)$ to $\banach{B}_{\Sull}(\Omega)$
is surjective.
Recall from~\Cref{prop:BLT} that this map is a bounded linear transformation. By~\Cref{prop:Banach_equivalences}, in order to show~\Cref{thm:sull2} it suffices to prove that for some finite radius~$R$, the image of the ball of radius $R$ in $\banach{B}_{\VS}(\Omega)$ is dense in the unit ball of
$\banach{B}_{\Sull}(\Omega)$. Thus, in order to prove \Cref{thm:sull2}, it will suffice to prove the following:
\begin{proposition}[Approximation]
\label{prop:approx}
	Let $\Omega$ be an SFT which is single-site fillable and satisfies the pivot property.
	Given $\varepsilon > 0$ and  $\psi \in \banach{B}_{\Sull}(\Omega)$, there is a shift-invariant
	finite-range interaction $\Phi^\varepsilon$ such that:
	\begin{enumerate}[label={\textup{(\roman*)}}]
		\item \label{prop:sullivan:approx:VS-norm}
			$\normVS{ \Phi^\varepsilon}\le 3 \normsull[]{\psi}$.
		\item \label{prop:sullivan:approx:closeness}
			$\normsull[]{\psi_{\Phi^\varepsilon} - \psi} < \varepsilon$.
	\end{enumerate}
	In particular, the image of the ball of radius~$3$ in $\banach{B}_{\VS}$
	under $\Phi \mapsto \psi_\Phi$ is dense in the unit ball of~$\banach{B}_{\Sull}(\Omega)$.
\end{proposition}

\begin{remark}[Comparison with Sullivan's proof]
	Strictly speaking, Sullivan's original proof of \Cref{thm:sull2} deals only with the case where $\Omega$
	is the full-shift.
	The basic approach of using \Cref{prop:approx} to prove \Cref{thm:sull2} is implicit in \cite{Sul73}.
	However, Sullivan's original proof of \Cref{thm:sull2} seems to use some additional and very special
	properties of the group $\ZZ^d$ in addition to amenability, such as the existence of a left-invariant
	total order and residual finiteness. Our proof below relies only on the amenability of $\ZZ^d$ and
	can be easily adapted to show that the same result holds when
	$\ZZ^d$ is replaced by an arbitrary countable amenable group.
	\hfill\remarkqed
\end{remark}

\begin{remark}[Shift-invariant variation-summable representation: alternative hypothesis]
\label{rem:sullivan:alternative}
	Recall from Remark~\ref{rmk:Sullivan_norm_safe} that for a shift space $\Omega$ with a safe symbol,
	$\banach{B}_{\Sull}(\Omega)$  is a Banach space.   Since such an $\Omega$ satisfies
	the uniform pivot property,
	by Proposition~\ref{prop:BanachVS_complete}, $\banach{B}_{\VS}(\Omega)$ is also a Banach space.
	We claim that in this case the conclusion of
	Theorem~\ref{thm:sull2} still holds, even though $\Omega$ need not be of finite type.  This
	already recovers Sullivan's original result.

	The proof of this follows along the same lines as Proposition~\ref{prop:approx} below,
	but is much simpler.  Here, for a
	continuous cocycle $\psi\colon\relation{T}(\Omega) \to \RR$
	and $n \in \NN$ one defines a finite range interaction $\Phi^{(n)}:\Lang(\Omega) \to \RR$
	which is nonzero only on translates of $F_n \isdef [-n,n]^d \cap \ZZ^d$, namely
	\begin{align}
		\Phi^n_{k + F_n}(x) &\isdef
			\frac{1}{\abs{F_n}}\psi(x_{k+F_n} \vee \diamond^{\ZZ^d \setminus (F_n+k)},\diamond^{\ZZ^d}) \;.
	\end{align}
	Then, for every $\varepsilon >0$ and $n$ sufficiently large,
	$\Phi^\varepsilon \isdef \Phi^{(n)}$ will satisfy the conclusion of
	Proposition~\ref{prop:approx}.
	This interaction appears simpler than the one used by Sullivan~\cite{Sul73}.
	\hfill\remarkqed
\end{remark}

Our proof of  Proposition~\ref{prop:approx}  requires  two technical lemmas.
For a finite set $K \Subset \ZZ^d$, we say that $D \Subset \ZZ^d$ is \emph{$K$-separated} if for any distinct $u,v \in D$ we have $(u+K)\cap(v+K) = \varnothing$ and that $D$ \emph{$K$-covers} a set $F$ if for every $f \in F$ there is $u \in D$ such that $f \in u+K$.

\begin{lemma}\label{lemma:delonepartition}
	Let $K \Subset \ZZ^d$ be a symmetric set ($K = -K$) which contains $0$. For every $F \Subset \ZZ^d$, there is a partition of $F$ of size at most $\abs{K}^2$ such that each element of the partition is $K$-separated.
\end{lemma}

\begin{proof}
	We claim that there is a subset $D \subseteq F$ which is both $K$-separated and that it $(K+K)$-covers~$F$.
	Indeed, let $D$ be a maximal $K$-separated subset of $F$ and suppose it does not $(K+K)$-cover $F$.
	Then there is $f \in F\setminus D$ such that $f \notin u+K+K$ for every $u \in D$. As $K$ is symmetric,
	we have $(f+K) \cap (u+K) = \varnothing$ for every $u \in D$, consequently $D \cup \{f\}$
	is also $K$-separated, contradicting the choice of $D$.
	
	Since $D\subseteq F$ is $K$-separated, for every $u \in \ZZ^d$ we have that $(u+D)\cap F$
	is also $K$-separated. As $D$ is $(K+K)$-covering the union of the sets in the collection
	$\mathcal{U} = \{ (u+D) \cap F : u \in K+K \}$ is~$F$. Making this cover disjoint yields
	a partition of $F$ into $K$-separated sets with at most $\abs{K+K}\leq \abs{K}^2$ elements.
\end{proof}

For the remainder of this section, we set $F_n \isdef [-n,n]^d \cap \ZZ^d$.

\begin{lemma}\label{lemma:brayatan}
	Let $\Omega \subseteq\Sigma^{\ZZ^d}$ be a single-site fillable SFT and $w \in \Omega$.
	There exist a constant $N \in \NN$,  a finite set $\Lambda \Subset \ZZ^d$ and a continuous function
	$z\colon\Omega \times\{n\in\NN: n>N\}\to \Omega$ such that for every $n >N$  the following hold:
	\begin{enumerate}[label={\textup{(\alph*)}}]
		\item \label{item:brayatan:condition:x}
			For every $x \in \Omega$, $z(x,n)_{F_n} = x_{F_n}$.
		\item \label{item:brayatan:condition:w}
			For every $x \in \Omega$,  $z(x,n)_{\ZZ^d\setminus F_{n+N}} = w_{\ZZ^d \setminus F_{n+N}}$.
		\item \label{item:brayatan:condition:interior}
			For every $j \in F_{n-N}$ and $(x,y)\in \relation{T}_{j}(\Omega)$,
			we have $\big(z(x,n),z(y,n)\big) \in \relation{T}_{j}(\Omega)$.
		\item \label{item:brayatan:condition:margin}
			For every $j \in F_{n}\setminus F_{n-N}$ and $(x,y)\in \relation{T}_{j}(\Omega)$,
			then $\big(z(x,n),z(y,n)\big) \in \relation{T}_{j+\Lambda}(\Omega)$.
	\end{enumerate}
\end{lemma}

\begin{proof}
	Let $\mathcal{F}$ be a finite set of forbidden finite patterns defining $\Omega$
	with respect to which $\Omega$ is single-site fillable.
	Let $K' \Subset \ZZ^d$ be the union of the shapes of every pattern in $\mathcal{F}$. Let $N'$ be
	an integer such that $K' \subseteq F_{N'}$, and let $K \isdef F_{N'}$ and $N \isdef 2N'$.
	Note that $K$ is a symmetric finite subset which contains~$0$ and the support of every pattern
	in~$\mathcal{F}$.
	
	We claim that any $K$-separated set $D$ has the property that any two distinct $a_1,a_2 \in D$ may not belong to the shift of a shape of some $q \in \mathcal{F}$. Indeed, suppose there is $b \in \ZZ^d$ and $s_1,s_2$ in the shape of $q$ such that $a_1 =b+s_1$ and $a_2 = b+s_2$. We get that $a_1-s_1 = a_2-s_2$.  Since $K$ is symmetric, it follows that $(a_1 +K) \cap (a_2 +K) \neq \varnothing$, contradicting the fact that $D$ is $K$-separated.
	
	 Consider $x\in\Omega$ and $n>N$, and let us construct $z(x,n)$
	 (see Figure~\ref{fig:ejemplo_brayan} for an illustration). 
	 By~\Cref{lemma:delonepartition}, there is a partition $\{A_1,A_2,\dots A_{\ell}\}$ of $F_{n+N}\setminus F_{n}$ such that $\ell \leq \abs{K}^2$ and each $A_i$ is $K$-separated. Let us define a finite sequence of locally-admissible patterns $p^0,p^1,\dots p^{\ell}$ such that:
	\begin{enumerate}[label={(\roman*)}]
		\item \label{item:brayatan:i}
			$p^0 = x_{F_{n}}\vee w_{\ZZ^d \setminus F_{n+N}}$.
		\item \label{item:brayatan:ii}
			The shape of $p^k$ is $\ZZ^d \setminus \bigcup_{i > k} A_i$.
		\item \label{item:brayatan:iii}
			For every $1 \leq k \leq \ell$,  the restriction of $p^{k}$
			to $\ZZ^d \setminus \bigcup_{i > k-1} A_i$ is $p^{k-1}$.
	\end{enumerate}
	As $N = 2N'$, for every pair of sites $a\in F_n$ and $b\in \ZZ^d \setminus F_{n+N}$ is $K$-separated and thus $p^0$ is locally admissible. We only need to describe the values of $p^i$ on $A_i$ for $i \geq 1$.
	Let us fix an arbitrary total ordering of $\Sigma$.
	For $a \in A_i$, let us define $p^{i}_a$ as the smallest symbol of $\Sigma$ such that $p^{i-1} \vee p^{i}_a$ is a locally-admissible pattern. The existence of such symbol is guaranteed by the single-site fillability of $\Omega$. Note that the value $p^{i}_a$ only depends upon the values of $p^{i-1}$ in $a+K$
	(see Figure~\ref{fig:ejemplo_brayan2}).
	
	Let us show that $p^i$ is locally admissible. By definition, for each $a \in A_i$ we have that $p^{i-1} \vee p^{i}_a$ is locally admissible.  Therefore, if some $q \in \mathcal{F}$ appears in $p^i$, then its shape must contain at least two coordinates from $A_i$. This is impossible because $A_i$ is $K$-separated.
	
	By property~\ref{item:brayatan:ii}, the shape of $p^{\ell}$ is $\ZZ^d$.
	Let us define $z = z(x,n) = p^{\ell}$. Combining properties~\ref{item:brayatan:i} and~\ref{item:brayatan:iii},
	we have $z_{F_n} = x_{F_n}$ and $z_{\ZZ^d \setminus F_{n+N}} = w_{\ZZ^d \setminus F_{n+N}}$.
	It remains to verify conditions~\ref{item:brayatan:condition:interior}
	and~\ref{item:brayatan:condition:margin}.
	
	Let $j \in F_{n-N}$.  Note that $\{j\} \cup A_i$ is $K$-separated, hence no forbidden pattern can contain $j$ and some $a \in A_i$ simultaneously in its support. This shows that the values of $z$ at the sites in $F_{n+N}\setminus{F_n}$ do not depend upon $x_{j}$.  Therefore, $\big(z(x,n),z(y,n)\big) \in \relation{T}_{j}(\Omega)$
	whenever $(x,y)\in \relation{T}_{j}(\Omega)$.
	
	Let $j \in F_n \setminus F_{n-N}$.  Set $\Lambda_0 \isdef \{j\}$, and for $i \geq 1$ let
	\begin{align}
		\Lambda_i &\isdef \Lambda_{i-1} \cup
			\{ a \in A_i :
				\text{there is $b \in \Lambda_{i-1}$ and $c \in \ZZ^d$ such that
					$\{a,b\} \subseteq c+K$}
			\} \;.
	\end{align}
	Fix $b \in \Lambda_{i-1}$.  If $\{a,b\} \subseteq c+K $, then there are $s_1,s_2 \in K$ such that $a = c+s_1$ and $b = c+s_2$ and thus $a = b-s_2+s_1$. We get $\Lambda_{i+1} \subseteq \Lambda_i +K-K \subseteq\Lambda_i + F_{N}$.  Thus, letting $\Lambda = F_{\ell N}$, we obtain $\Lambda_{\ell} \subseteq j+\Lambda$. Note that $\Lambda_{\ell}$ contains the set of sites in $F_{n+N}\setminus F_{n}$ at which the value of $z$ depends upon $x_{j}$. We find that whenever $(x,y)\in \relation{T}_{j}(\Omega)$, we have $\big(z(x,n),z(y,n)\big) \in \relation{T}_{j+\Lambda}(\Omega)$.
\end{proof}
	
\begin{figure}[h!]
	\centering
	\begin{tikzpicture}[scale = 0.4, >=stealth']
	\begin{scope}[shift = {(0,0)}]
	\draw[fill = black!10] (0,0) rectangle (9,1);
	\draw[fill = black!10] (0,0) rectangle (1,9);
	\draw[fill = black!10] (8,0) rectangle (9,9);
	\draw[fill = black!10] (0,8) rectangle (9,9);
	\draw[fill = black!10] (2,2) rectangle (7,7);
	\node at (0.5,8.5) {$w$};
	\node at (4.5,4.5) {$x$};
	\draw[black!50] (0,0) grid (9,9);
	\end{scope}
	\draw [thick, ->] (9.5,4.5) -- (13.5,4.5);
	\begin{scope}[shift = {(14,0)}]
	\draw[fill = black!10] (0,0) rectangle (9,1);
	\draw[fill = black!10] (0,0) rectangle (1,9);
	\draw[fill = black!10] (8,0) rectangle (9,9);
	\draw[fill = black!10] (0,8) rectangle (9,9);
	\draw[fill = black!10] (2,2) rectangle (7,7);
	\node at (0.5,8.5) {$w$};
	\node at (4.5,4.5) {$x$};
	\foreach \i/\j in {2/1, 4/1, 6/1, 2/7, 4/7, 6/7, 1/2, 1/4, 1/6, 7/2, 7/4, 7/6 }{
		\draw[pattern = north east lines] (\i,\j) rectangle +(1,1);
	}
	\draw[black!50] (0,0) grid (9,9);
	\end{scope}
	\draw [thick, ->] (23.5,4.5) -- (27.5,4.5);
	\begin{scope}[shift = {(28,0)}]
	\draw[fill = black!10] (0,0) rectangle (9,1);
	\draw[fill = black!10] (0,0) rectangle (1,9);
	\draw[fill = black!10] (8,0) rectangle (9,9);
	\draw[fill = black!10] (0,8) rectangle (9,9);
	\draw[fill = black!10] (2,2) rectangle (7,7);
	\node at (0.5,8.5) {$w$};
	\node at (4.5,4.5) {$x$};
	\foreach \i/\j in {2/1, 4/1, 6/1, 2/7, 4/7, 6/7, 1/2, 1/4, 1/6, 7/2, 7/4, 7/6 }{
		\draw[fill = black!10] (\i,\j) rectangle +(1,1);
	}
	\foreach \i/\j in {1/1, 3/1, 5/1, 7/1, 1/7, 3/7, 5/7, 7/7, 1/3, 1/5, 7/3, 7/5 }{
		\draw[pattern = north east lines] (\i,\j) rectangle +(1,1);
	}
	\draw[black!50] (0,0) grid (9,9);
	\end{scope}
	\end{tikzpicture}
	\caption{%
		An illustration for the construction of $z(x,n)$ with $n=2$ in \Cref{lemma:brayatan}.
		Here, $\Omega$ is assumed to be a nearest-neighbour single-site fillable SFT.
		On the left, we start with the pattern which coincides with $x$ in $F_{n}$ and with $w$ outside $F_{n+1}$.
		In the middle picture, we fill every odd site in $F_{n+1} \setminus F_{n}$ with the smallest symbol
		that does not generate a forbidden pattern.
		(This can be done because of single-site fillability).
		Finally, in the right, we fill every even site in $F_{n+1}\setminus F_n$ with the smallest symbol that
		does not generate a forbidden pattern.
	}
	\label{fig:ejemplo_brayan}
\end{figure}
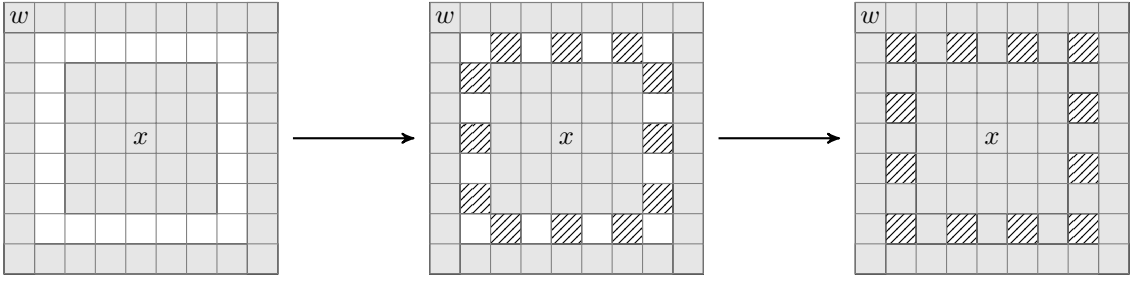

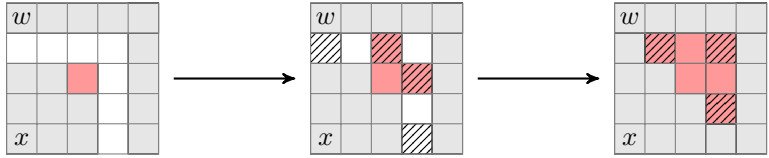
\begin{figure}[h!]
	\centering
	\begin{tikzpicture}[scale = 0.4, >=stealth']
	\begin{scope}[shift = {(0,0)}]
	\draw[fill = black!10] (8,4) rectangle (9,9);
	\draw[fill = black!10] (4,8) rectangle (9,9);
	\draw[fill = black!10] (4,4) rectangle (7,7);
	\draw[fill = red!40] (6,6) rectangle (7,7);
	\node at (4.5,8.5) {$w$};
	\node at (4.5,4.5) {$x$};
	\draw[black!50] (4,4) grid (9,9);
	\end{scope}
	\draw [thick, ->] (9.5,6.5) -- (13.5,6.5);
	\begin{scope}[shift = {(10,0)}]
	\draw[fill = black!10] (8,4) rectangle (9,9);
	\draw[fill = black!10] (4,8) rectangle (9,9);
	\draw[fill = black!10] (4,4) rectangle (7,7);
	\draw[fill = red!40] (6,6) rectangle (7,7);
	\node at (4.5,8.5) {$w$};
	\node at (4.5,4.5) {$x$};
	\foreach \i/\j in { 4/7,  7/4 }{
		\draw[pattern = north east lines] (\i,\j) rectangle +(1,1);
	}
	\foreach \i/\j in {  6/7, 7/6 }{
		\draw[fill = red!40] (\i,\j) rectangle +(1,1);
		\draw[pattern = north east lines] (\i,\j) rectangle +(1,1);
	}
	\draw[black!50] (4,4) grid (9,9);
	\end{scope}
	\draw [thick, ->] (19.5,6.5) -- (23.5,6.5);
	\begin{scope}[shift = {(20,0)}]
	\draw[fill = black!10] (8,4) rectangle (9,9);
	\draw[fill = black!10] (4,8) rectangle (9,9);
	\draw[fill = black!10] (4,4) rectangle (7,7);
	\node at (4.5,8.5) {$w$};
	\node at (4.5,4.5) {$x$};
	\draw[fill = red!40] (6,6) rectangle (7,7);
	\draw[fill = red!40] (7,6) rectangle (8,7);
	\draw[fill = red!40] (6,7) rectangle (7,8);
	\foreach \i/\j in { 4/7,  7/4 }{
		\draw[fill = black!10] (\i,\j) rectangle +(1,1);
	}
	\foreach \i/\j in {5/7, 7/7, 7/5 }{
		\draw[fill = red!40] (\i,\j) rectangle +(1,1);
		\draw[pattern = north east lines] (\i,\j) rectangle +(1,1);
	}
	\draw[black!50] (4,4) grid (9,9);
	\end{scope}
	\end{tikzpicture}
	\caption{%
		In \Cref{lemma:brayatan}, a modification of $x$ in the boundary of $F_n$ only affects the value
		of $z(x,n)$ in a bounded region.
	}
	\label{fig:ejemplo_brayan2}
\end{figure}

\begin{proof}[Proof of \Cref{prop:approx}]
	If $\psi = 0$, then the result is trivial.
	Thus, let us assume $\psi \neq 0$.
	
	Fix some $w \in \Omega$. By~\Cref{lemma:brayatan}, there is $N \in \NN$ and $\Lambda \Subset \ZZ^d$ such that
	for every $x \in \Omega$ and $n >N$ we have a configuration $z(x,n) \in \Omega$
	satisfying the conditions of the lemma.
	For $n>N$, let
	\begin{align}
		f_n(x) &= \psi(z(x,n),w) \;,
	\end{align}
	and define the interaction $\Phi^n$, supported only on translates of $F_n$, by
	\begin{align}
		\Phi^n_{k + F_n}(x) &\isdef \frac{1}{\abs{F_n}}f_n(\sigma^{k} x) \;.
	\end{align}
	We show that for $n$ sufficiently large, $\Phi^\varepsilon\isdef\Phi^n$ satisfies
	the required conditions.
	
	We start by arguing that for $k\in F_n$, the variation $\Var_{k}(f_n)$ is bounded by a constant
	independent of~$n$. 
	First, assume that $k \in F_{n-N}$.  By Lemma~\ref{lemma:brayatan},
	$\big(z(x,n),z(y,n)\big) \in \relation{T}_{k}(\Omega)$
	whenever $(x,y)\in \relation{T}_{k}(\Omega)$.
	Recall the notation $\zeta_k x$ from~\Cref{sec:prelim:Banach_cocycles},
	and note that if $(x,y)\in \relation{T}_{k}(\Omega)$, then $\zeta_k x = \zeta_k y$.
	From these two facts, we deduce
	\begin{align}
	\Var_{k}(f_n) &\label{eq:variation0} = \sup_{(x,y) \in \relation{T}_{k}(\Omega)} \abs[\big]{f_n(y) - f_n(x)}\\
	&  = \sup_{(x,y) \in \relation{T}_{k}(\Omega)} \abs[\big]{\psi(z(y,n),w) - \psi(z(x,n),w)}\\
	& = \sup_{(x,y) \in \relation{T}_{k}(\Omega)} \abs[\big]{\psi\big(z(x,n),z(y,n)\big)}\\
		& \leq \sup_{(x,y) \in \relation{T}_{k}(\Omega)} \abs{\psi(x,y)}\\
	& \label{eq:variation1} \leq \sup_{x \in \Omega} \abs[\big]{\psi(x,\zeta_k x)}+\sup_{y \in \Omega} \abs[\big]{\psi(y,\zeta_k y)}
	\leq 2 \normsull{\psi}.
	\end{align}
	Next, we claim there is a constant $K >0$ such that $\abs{\Var_{k}(f_n)} \leq K$ for all $k \in F_{n}\setminus F_{n-N}$. Indeed, if $(x,y)\in \relation{T}_k(\Omega)$ by the lemma we have $\big(z(x,n),z(y,n)\big) \in \relation{T}_{k+\Lambda}(\Omega)$. As $\psi$ is continuous  and $\relation{T}_{\Lambda}(\Omega)$ is compact, there is a $K \in \RR$ such that $\sup_{(x,y)\in \relation{T}_{\Lambda}(\Omega)}\abs{\psi(x,y)} \leq K$. Therefore,
	
	\begin{align}
	\Var_{k}(f_n) & = \sup_{(x,y) \in \relation{T}_{k}(\Omega)} \abs{f_n(y) - f_n(x)}\\
	 &\label{eq:variation2_prime} = \sup_{(x,y) \in \relation{T}_{k}(\Omega)} \abs[\big]{\psi\big(z(x,n),z(y,n)\big)}\\
	 &\label{eq:variation2} \leq \sup_{(x,y) \in \relation{T}_{k+\Lambda}(\Omega)} \abs{\psi(x,y)} = \sup_{(x,y) \in \relation{T}_{\Lambda}(\Omega)} \abs{\psi(x,y)} \leq K \;.
	\end{align}

	Now, observe that
	\begin{align}
	\normVS[]{\Phi^n} & = \sum_{\substack{k \in \ZZ^d\\ 0 \in k+F_n}}\Var_0(\Phi^n_{k + F_n}) = \sum_{k \in F_n }\Var_0(\Phi^n_{k + F_n})\\
	& = \frac{1}{\abs{F_n}}\sum_{k \in F_n }\Var_0(f_n \circ \sigma^{k}) = \frac{1}{\abs{F_n}}\sum_{k \in F_n }\Var_k(f_n) \;.
	\end{align}
	Using~\eqref{eq:variation1} and~\eqref{eq:variation2}, we find that	
	\begin{align}
		\normVS{\Phi^n} & = 
			\frac{1}{\abs{F_n}}\sum_{k \in F_{n-N} }\Var_k(f_n)
				+ \frac{1}{\abs{F_n}}\sum_{k \in F_{n}\setminus F_{n-N} }\Var_k(f_n) \\
		& \leq
			\frac{\abs{F_{n-N}}}{\abs{F_n}}2\normsull{\psi}
				+ \frac{\abs{F_n \setminus F_{n-N}}}{\abs{F_n}}K \;.
	\end{align}
	Since $\psi \neq 0$, letting $n$ be sufficiently large, we obtain
	\begin{align}
	\normVS{\Phi^n} \leq 3\normsull{\psi} \;,
	\end{align}
	and thus condition~\ref{prop:sullivan:approx:VS-norm} is satisfied.
	
	To verify condition~\ref{prop:sullivan:approx:closeness}, observe that
	\begin{align}\label{eq:psi_diff_psi_Phi_n}
		\psi_{\Phi^n}(x, \zeta_0 x) & =
			\sum_{k \in \ZZ^d} \big[\Phi^n_{k + F_n}(\zeta_0 x) - \Phi^n_{k + F_n}(x)\big] \\
		&=
			\sum_{k \in \ZZ^d} \big[\Phi^n_{F_n}(\sigma^k \zeta_0 x) - \Phi^n_{F_n}(\sigma^k x)\big] \\
		& =
			\frac{1}{\abs{F_n}}\sum_{k \in F_n} \big[f_n(\sigma^k \zeta_0 x) - f_n(\sigma^k x)\big] \\
		& \label{eq:psi_zeta}=
			\frac{1}{\abs{F_n}}\sum_{k \in F_n} \psi\big(z(\sigma^k x,n),z(\sigma^k \zeta_0 x,n) \big) \;.
	\end{align}
	By continuity of $\psi$,
	there exists $M \in \NN$ such that for every $m \geq M$
	and every $(x',y'),(x,y) \in \relation{T}_{0}(\Omega)$
	such that $x_{F_{m}} = x'_{F_{m}}$ and $y_{F_{m}} = y'_{F_{m}}$, we have
	\begin{align}
		\abs{ \psi(x',y') - \psi( x, y)} &\leq \frac{\varepsilon}{2} \;.
	\end{align}
	In particular, if we let %
	$N' = \max\{N,M\}$, then by shift-invariance of the cocycle, we have
	\begin{align}
		\MoveEqLeft\nonumber
		\abs[\big]{ \psi\big(z(\sigma^k x,n),z(\sigma^k \zeta_0 x,n)\big) - \psi(x, \zeta_0 x)} \\
		&= \abs[\big]{ \psi\big(z(\sigma^k x,n),z(\sigma^k \zeta_0 x,n)\big)
			- \psi(\sigma^k x, \sigma^k \zeta_0 x)}
			\leq \frac{\varepsilon}{2} \;.
	\end{align}
	for every $k\in F_{n-N'}$.
	For $k \in F_n \setminus F_{n-N'}$, on the other hand,
	by~\eqref{eq:variation2_prime} and~\eqref{eq:variation2}, we have
	\begin{align}
		\abs[\big]{\psi(z(\sigma^k x,n),z(\sigma^k \zeta_0 x,n) )} &\leq K \;,
	\end{align}
	and so for such $k$,
	\begin{align}
		\abs[\big]{ \psi(z(\sigma^k x,n),z(\sigma^k \zeta_0 x,n)) - \psi(x,  \zeta_0 x)}
			&\leq K + \normsull{\psi} \;.
	\end{align}
	Combining these two bounds with \eqref{eq:psi_zeta}, we obtain that for any $n \geq N'$,
	\begin{align}
		\abs[\big]{\psi_{\Phi^n}(x,\zeta_0 x) - \psi(x,\zeta_0 x)}
			&\leq \frac{\abs{F_{n-N'}}}{\abs{F_n}}\frac{\varepsilon}{2}
			+ \frac{\abs{F_n\setminus F_{n-N'}}}{\abs{F_n}}(K + \normsull{\psi}) \;.
	\end{align}
	Therefore, %
	choosing $n\geq N'$ large enough, we have
	\begin{align}
		\normsull[]{\psi_{\Phi^n} - \psi} &=
			\sup_{x \in \Omega}\abs[\big]{\psi_{\Phi^n}(x, \zeta_0 x) - \psi(x, \zeta_0 x)} < \varepsilon \;.
	\end{align}
	Hence, condition~\ref{prop:sullivan:approx:closeness} is also satisfied. 
\end{proof}

\bibliographystyle{plainurl}
\bibliography{bibliography}

\newcommand{\noopsort}[1]{}
\begin{thebibliography}{10}

\bibitem{Ave72}
M.~B. Averintsev.
\newblock Description of {M}arkovian random fields by {G}ibbsian conditional
  probabilities.
\newblock {\em Theory of Probability and Its Applications}, 17(1):20--33, 1972.
\newblock \href {http://dx.doi.org/10.1137/1117002}
  {\path{doi:10.1137/1117002}}.

\bibitem{BarGomMarTaa18}
S.~Barbieri, R.~G{\'{o}}mez, B.~Marcus, and S.~Taati.
\newblock Equivalence of relative {G}ibbs and relative equilibrium measures for
  actions of countable amenable groups.
\newblock {\em Nonlinearity}, 33(5):2409--2454, 2020.
\newblock \href {http://dx.doi.org/10.1088/1361-6544/ab6a75}
  {\path{doi:10.1088/1361-6544/ab6a75}}.

\bibitem{Cha17}
N.~Chandgotia.
\newblock Generalisation of the {H}ammersley-{C}lifford theorem on bipartite
  graphs.
\newblock {\em Transactions of the American Mathematical Society},
  369(10):7107--7137, 2017.
\newblock \href {http://dx.doi.org/10.1090/tran/6899}
  {\path{doi:10.1090/tran/6899}}.

\bibitem{ChaHanMarMeyPav14}
N.~Chandgotia, G.~Han, B.~Marcus, T.~Meyerovitch, and R.~Pavlov.
\newblock One-dimensional {M}arkov random fields, {M}arkov chains and
  topological {M}arkov fields.
\newblock {\em Proceedings of the American Mathematical Society},
  142(1):227--242, 2014.
\newblock \href {http://dx.doi.org/10.1090/S0002-9939-2013-11741-7}
  {\path{doi:10.1090/S0002-9939-2013-11741-7}}.

\bibitem{ChaMey16}
N.~Chandgotia and T.~Meyerovitch.
\newblock Markov random fields, {M}arkov cocycles and the 3-colored chessboard.
\newblock {\em Israel Journal of Mathematics}, 215(2):909--964, 2016.
\newblock \href {http://dx.doi.org/10.1007/s11856-016-1398-2}
  {\path{doi:10.1007/s11856-016-1398-2}}.

\bibitem{DacNah18}
S.~Dachian and B.~Nahapetian.
\newblock On the relationship of energy and probability in models of classical
  statistical physics.
\newblock {\em Preprint}, 2018.
\newblock \href {http://arxiv.org/abs/1810.05388} {\path{arXiv:1810.05388}}.

\bibitem{Dob68c}
R.~L. Dobrushin.
\newblock The description of a random field by means of conditional
  probabilities and conditions of its regularity.
\newblock {\em Theory of Probability and Its Applications}, 13(2):197--224,
  1968.
\newblock \href {http://dx.doi.org/10.1137/1113026}
  {\path{doi:10.1137/1113026}}.

\bibitem{Dob68b}
R.~L. Dobrushin.
\newblock Gibbsian random fields for lattice systems with pairwise
  interactions.
\newblock {\em Functional Analysis and Its Applications}, 2(4):292–--301,
  1968.
\newblock \href {http://dx.doi.org/10.1007/BF01075681}
  {\path{doi:10.1007/BF01075681}}.

\bibitem{EntFerSok93}
A.~C.~D. {\noopsort{Enter}}{van~E}nter, R.~Fern\'andez, and A.~D. Sokal.
\newblock Regularity properties and pathologies of position-space
  renormalization-group transformations: Scope and limitations of {G}ibbsian
  theory.
\newblock {\em Journal of Statistical Physics}, 72(5/6), 1993.
\newblock \href {http://dx.doi.org/10.1007/BF01048183}
  {\path{doi:10.1007/BF01048183}}.

\bibitem{Fer06}
R.~Fern\'andez.
\newblock Gibbsianness and non-{G}ibbsianness in lattice random fields.
\newblock In A.~Bovier, F.~Dunlop, F~{d}en Hollander, A.~{v}an Enter, and
  J.~Dalibard, editors, {\em Mathematical Statistical Physics}, Les Houches,
  Session LXXXIII, 2005, pages 731--799. Elsevier, 2006.
\newblock \href {http://dx.doi.org/10.1016/s0924-8099(06)80052-1}
  {\path{doi:10.1016/s0924-8099(06)80052-1}}.

\bibitem{FriVel17}
S.~Friedli and Y.~Velenik.
\newblock {\em Statistical Mechanics of Lattice Systems: a Concrete
  Mathematical Introduction}.
\newblock Cambridge U. Press, 2017.
\newblock \href {http://dx.doi.org/10.1017/9781316882603}
  {\path{doi:10.1017/9781316882603}}.

\bibitem{Geo88}
H.-O. Georgii.
\newblock {\em Gibbs Measures and Phase Transitions}.
\newblock Walter de Gruyter, 1988.
\newblock \href {http://dx.doi.org/10.1515/9783110250329}
  {\path{doi:10.1515/9783110250329}}.

\bibitem{GeoHagMae00}
H.-O. Georgii, O.~H\"aggstr\"om, and C.~Maes.
\newblock The random geometry of equilibrium phases.
\newblock In C.~Domb and J.~Lebowitz, editors, {\em Phase Transitions and
  Critical Phenomena}, volume~18, pages 1--142. Academic Press, 2000.
\newblock \href {http://dx.doi.org/10.1016/S1062-7901(01)80008-2}
  {\path{doi:10.1016/S1062-7901(01)80008-2}}.

\bibitem{Gol78}
S.~Goldstein.
\newblock A note on specifications.
\newblock {\em Zeitschrift f\"ur Wahrscheinlichkeitstheorie und Verwandte
  Gebiete}, 46(1):45--51, 1978.
\newblock \href {http://dx.doi.org/10.1007/BF00535686}
  {\path{doi:10.1007/BF00535686}}.

\bibitem{Gri73}
G.~R. Grimmett.
\newblock A theorem about random fields.
\newblock {\em Bulletin of the London Mathematical Society}, 5:81--84, 1973.
\newblock \href {http://dx.doi.org/doi:10.1112/blms/5.1.81}
  {\path{doi:doi:10.1112/blms/5.1.81}}.

\bibitem{Gro82}
L.~Gross.
\newblock Thermodynamics, statistical mechanics and random fields.
\newblock In P.~L. Hennequin, editor, {\em Ecole d'Et\'{e} de Probabilit\'{e}s
  de Saint-Flour X -- 1980}, pages 101--204. Springer, 1982.
\newblock \href {http://dx.doi.org/10.1007/bfb0095619}
  {\path{doi:10.1007/bfb0095619}}.

\bibitem{HamCli68}
J.~M. Hammersley and P.~Clifford.
\newblock Markov fields on finite graphs and lattices.
\newblock Unpublished manuscript, 1968.
\newblock URL: \url{http://www.statslab.cam.ac.uk/~grg/books/jmh.html}.

\bibitem{Hoe63}
W.~Hoeffding.
\newblock Probability inequalities for sums of bounded random variables.
\newblock {\em Journal of the American Statistical Association},
  58(301):13--30, 1963.
\newblock \href {http://dx.doi.org/10.1080/01621459.1963.10500830}
  {\path{doi:10.1080/01621459.1963.10500830}}.

\bibitem{Isr79}
R.~B. Israel.
\newblock {\em Convexity in the Theory of Lattice Gases}.
\newblock Princeton University Press, 1979.
\newblock \href {http://dx.doi.org/10.1515/9781400868421}
  {\path{doi:10.1515/9781400868421}}.

\bibitem{Koz74}
O.~K. Kozlov.
\newblock Gibbs description of a system of random variables.
\newblock {\em Problems of Information Transmission}, 10(3):258--265, 1974.
\newblock URL: \url{http://mi.mathnet.ru/ppi1046}.

\bibitem{Koz77}
O.~K. Kozlov.
\newblock Consistent systems of conditional distributions of a random field.
\newblock {\em Problems of Information Transmission}, 13(3):218--228, 1977.
\newblock URL: \url{http://mi.mathnet.ru/ppi1096}.

\bibitem{LanRue69}
O.~E. L{anford~III} and D.~Ruelle.
\newblock Observables at infinity and states with short range correlations in
  statistical mechanics.
\newblock {\em Communications in Mathematical Physics}, 13(3):194--215, 1969.
\newblock \href {http://dx.doi.org/10.1007/BF01645487}
  {\path{doi:10.1007/BF01645487}}.

\bibitem{McD89}
C.~McDiarmid.
\newblock On the method of bounded differences.
\newblock In {\em Surveys in combinatorics}, volume 141 of {\em London
  Mathematical Society Lecture Note Series}, pages 148--188. Cambridge
  University Press, 1989.
\newblock \href {http://dx.doi.org/10.1017/CBO9781107359949.008}
  {\path{doi:10.1017/CBO9781107359949.008}}.

\bibitem{PetSch97}
K.~Petersen and K.~Schmidt.
\newblock Symmetric {G}ibbs measures.
\newblock {\em Transactions of the American Mathematical Society},
  349(7):2775--2811, 1997.
\newblock \href {http://dx.doi.org/10.1090/s0002-9947-97-01934-x}
  {\path{doi:10.1090/s0002-9947-97-01934-x}}.

\bibitem{PirSin75}
S.~A. Pirogov and Ya.~G. Sinai.
\newblock Phase diagrams of classical lattice systems.
\newblock {\em Theoretical and Mathematical Physics}, 25(3):1185--1192, 1975.
\newblock \href {http://dx.doi.org/10.1007/bf01040127}
  {\path{doi:10.1007/bf01040127}}.

\bibitem{Pre76}
C.~Preston.
\newblock {\em Random Fields}, volume 534 of {\em Lecture Notes in
  Mathematics}.
\newblock Springer, 1976.
\newblock \href {http://dx.doi.org/10.1007/BFb0080563}
  {\path{doi:10.1007/BFb0080563}}.

\bibitem{Pre80}
C.~Preston.
\newblock Construction of specifications.
\newblock In L.~Streit, editor, {\em Quantum Fields --- Algebras, Processes},
  pages 269--292. Springer, 1980.
\newblock \href {http://dx.doi.org/10.1007/978-3-7091-8598-8_18}
  {\path{doi:10.1007/978-3-7091-8598-8_18}}.

\bibitem{Put18}
I.~F. Putnam.
\newblock {\em Cantor minimal systems}, volume~70 of {\em University Lecture
  Series}.
\newblock American Mathematical Society, 2018.
\newblock \href {http://dx.doi.org/10.1090/ulect/070}
  {\path{doi:10.1090/ulect/070}}.

\bibitem{Rue04}
D.~Ruelle.
\newblock {\em Thermodynamic Formalism}.
\newblock Cambridge University Press, 2nd edition, 2004.
\newblock \href {http://dx.doi.org/10.1017/CBO9780511617546}
  {\path{doi:10.1017/CBO9780511617546}}.

\bibitem{Spi71}
F.~Spitzer.
\newblock Markov random fields and {G}ibbs ensembles.
\newblock {\em American Mathematical Monthly}, 78:142--154, 1971.
\newblock \href {http://dx.doi.org/10.2307/2317621}
  {\path{doi:10.2307/2317621}}.

\bibitem{Sul73}
W.~G. Sullivan.
\newblock Potentials for almost {M}arkovian random fields.
\newblock {\em Communications in Mathematical Physics}, 33:61--74, 1973.
\newblock \href {http://dx.doi.org/10.1007/BF01645607}
  {\path{doi:10.1007/BF01645607}}.

\end{thebibliography}

\appendix

\section{Appendix}

\subsection{Some symbolic dynamics facts}

The following proposition generalizes a remark made in~\cite{ChaMey16}, at the end of Section~3.1.
\begin{proposition}[TMP + safe symbol $\Rightarrow$ SFT]
\label{prop:tmp+safe->sft}
	Every shift space with the TMP that has a safe symbol is of finite type.
\end{proposition}
\begin{proof}
	Let $\Omega\subseteq\Sigma^{\ZZ^d}$ be a shift space which has the TMP and a safe symbol~$\blank$.
	Let $B\Subset\ZZ^d$ be a memory set for the singleton $\{0\}$ witnessing the TMP of $\Omega$.
	Let $\mathcal{F}\subseteq\Sigma^B$ denote the set of patterns with shape $B$
	which are not (globally) admissible in $\Omega$.
	We claim that $\Omega$ coincides with the SFT $\Omega'$ defined by forbidding the patterns in $\mathcal{F}$.
	
	Every configuration in $\Omega$ clearly avoids the patterns in $\mathcal{F}$,
	hence $\Omega\subseteq\Omega'$.
	Conversely, let $\Omega'_0$ denote the set of configurations in $\Omega'$
	that have no more than finitely many non-safe symbols.
	We show that $\Omega'_0\subseteq\Omega$.  Since $\Omega'_0$ is dense in $\Omega'$,
	this would imply that $\Omega'\subseteq\Omega$.
	
	To show that every $x\in\Omega'_0$ is in $\Omega$, we use induction on the number of
	non-safe symbols of $x$.
	If $x$ has no non-safe symbol, it is clearly in $\Omega$.
	Suppose that every element of $\Omega'_0$ with at most $k$ non-safe symbols is in $\Omega$.
	Let $x\in\Omega'_0$ be a configuration with $k+1$ non-safe symbols.
	Pick an arbitrary $k\in\ZZ^d$ with $x_k\neq\blank$.
	On the one hand, the configuration $y\isdef x_{\ZZ^d\setminus\{k\}}\lor\blank^k$ obtained from $x$
	by replacing the symbol at site $k$ with $\blank$ has $k$ non-safe symbols and thus,
	by the induction hypothesis, is in $\Omega$.
	On the other hand, by definition, $x_{B+k}$ is admissible in $\Omega$ and thus occurs in
	a configuration $z\in\Omega$.
	Since $\Omega$ has the TMP, it follows that $x= y_{\ZZ^d\setminus\{k\}}\lor z_{B+k}$
	is also in $\Omega$.
\end{proof}

\begin{proposition}[TMP + pivot $\Rightarrow$ uniform pivot]
\label{prop:TMP+pivot->uniform-pivot}
	If a configuration space with the TMP has the pivot property, then
	it also has the uniform pivot property.
\end{proposition}
\begin{proof}
	Let $\Omega\subseteq\Sigma^\Sites$ be a configuration space which has
	the TMP and the pivot property.
	Let $A\Subset\Sites$ be fixed.
	For each $(x,y) \in\relation{T}_A(\Omega)$, fix a sequence $x = x^{(0)}\to x^{(1)}\to\cdots\to x^{(n)} = y$
	of single-site pivots at sites $v_0, v_1\ldots, v_n$, transforming $x$ to $y$.
	Let $B_{x,y}$ be a memory set for $A_{x,y} \isdef A \cup \{v_1, \ldots, v_n\}$.
	
	Observe that if $(\bar{x},\bar{y}) \in\relation{T}_{A}(\Omega)$ is any asymptotic pair
	such that $\bar{x}_{B_{x,y}} = x_{B_{x,y}}$ and $\bar{y}_{B_{x,y}} = y_{B_{x,y}}$, then one can
	construct a sequence $\bar{x} = \bar{x}^{(0)}\to \bar{x}^{(1)}\to\cdots\to \bar{x}^{(n)} = \bar{y}$
	of single-site pivot moves at the same sites $v_0, v_1\ldots, v_n$, transforming $\bar{x}$ to $\bar{y}$,
	by defining
	\begin{align}
		\bar{x}^{(i)} &\isdef \bar{x}_{\Sites\setminus A_{x,y}}\lor x^{(i)}_{B_{x,y}} \;,
	\end{align}
	for $i=0,1,\ldots,n$.
	Since $B_{x,y}$ is a memory set for $A_{x,y}$ and
	$\bar{x}_{B_{x,y}\setminus A_{x,y}}=x^{(i)}_{B_{x,y}\setminus A_{x,y}}$,
	the configurations $\bar{x}^{(i)}$ are admissible in $\Omega$.
	The set of all such pairs $(\bar{x},\bar{y})$ is an open neighbourhood of $(x,y)$ in $\relation{T}_{A}(\Omega)$
	which we denote by $N_{A,x,y}$.

	The open sets $N_{A,x,y}$ for $(x,y) \in \relation{T}_A(\Omega)$ cover $\relation{T}_{A}(\Omega)$.
	By compactness, we can choose a finite set $F\subseteq\relation{T}_A(\Omega)$ such that
	$\{N_{A,x,y}: (x,y) \in F\}$ still covers $\relation{T}_A(\Omega)$.
	Let $C\isdef\bigcup_{(x,y)\in F} B_{x,y}$.
	Then, for every $(u,v)\in\relation{T}_A(\Omega)$, there is a sequence of single-site pivots
	from $u$ to $v$ which stays within $C$.
	Since this holds for every $A\Subset\Sites$, we find that $\Omega$ has the bounded pivot property.
\end{proof}

\subsection{Specifications and cocycles}
\label{apx:specification-vs-cocycle}

\begin{proof}[Proof of Proposition~\ref{prop:specification:continuous-positive:TMP}]
	Suppose $\Omega$ has the TMP, then the uniform specification defined on Example~\ref{exp:uniform-specification} is local and hence continuous. Furthermore, by definition this specification is positive.
	
	Conversely, fix $A \Subset \ZZ^d$. As the specification is continuous and positive, by compactness of $\Omega$ it follows that $\varepsilon  \isdef \frac{1}{2}\inf_{x \in \Omega} K_A(x,[x_A])>0$. Also, by continuity of the specification,
	we can find a finite $B \supseteq A$ such that for all $p \in \Sigma^A$ and $x,y \in \Omega$ so that $x_{B \setminus A}= y_{B \setminus A}$, we have
	\begin{align}
	\abs{K_A(x, [p]) - K_A(y, [p])} &< \varepsilon.
	\end{align}
	In particular, we obtain that if $x,y \in \Omega$ so that $x_{B \setminus A}= y_{B \setminus A}$, then
	\begin{align}
	\abs{K_A(x, [x_A]) - K_A(y, [x_A])}
	&< \varepsilon \leq \frac{1}{2}K_A(x, [x_A])
	\end{align}
	and so $K_A(y, [x_A]) > 0$.  This shows that $x_A \lor y_{\ZZ^d \setminus A} \in \Omega$.
	As the choice of $B$ does not depend upon $x,y \in \Omega$ we deduce that $B$ is a memory set for $A$.
	Since $A$ was arbitrary, we conclude that $\Omega$ has the~TMP.
\end{proof}

\begin{proof}[Proof of Proposition~\ref{prop:specification-vs-cocycle:equivalence}]
	First, let $K$ be a positive specification on $\Omega$, and
	for $(x,y)\in\relation{T}(\Omega)$, define
	\begin{align}
		\psi(x,y) &\isdef -\log\left[\frac{K_A(y,[y_A])}{K_A(x,[x_A])}\right]
	\end{align}
	where $A\Subset\Sites$ is the set of sites at which $x$ and $y$ disagree.
	Note that if $B\supseteq A$ is another finite set containing $A$,
	then by the consistency of the kernels $K_A$ and $K_B$,
	\begin{align}
		\frac{K_B(y,[y_B])}{K_B(x,[x_B])} &=
			\frac{K_B(y,[y_{B\setminus A}])K_A(y,[y_A])}{K_B(x,[x_{B\setminus A}])K_A(x,[x_A])}
		=
			\frac{K_A(y,[y_A])}{K_A(x,[x_A])} \;.
	\end{align}
	Now, let $(x,y),(y,z)\in\relation{T}(\Omega)$.
	Define $B$ as the union of the disagreement positions of $(x,y)$ and $(y,z)$.
	Then,
	\begin{align}
		\psi(x,y) + \psi(y,z) &=
			-\log\left[\frac{K_B(y,[y_B])}{K_B(x,[x_B])}\right]
			-\log\left[\frac{K_B(z,[z_B])}{K_B(y,[y_B])}\right] \\
		&=
			-\log\left[\frac{K_B(z,[z_B])}{K_B(x,[x_B])}\right] \\
		&=
			\psi(x,z) \;,
	\end{align}
	which means $\psi$ is a cocycle on $\relation{T}(\Omega)$.
	Clearly $\psi$ is measurable with respect to the $\sigma$-algebra induced from $\Omega\times\Omega$.
	
	Conversely, let $\psi$ be a measurable cocycle on $\relation{T}(\Omega)$.
	For $x\in\Omega$ and $A\Subset\Sites$, define
	\begin{align}
	\label{eq:specification-of-a-cocycle:def}
		K_A\big(x,[q_B]\cap [p_A]\big) &\isdef
			\begin{cases}
				\frac{1}{Z_{A|x_{\Sites\setminus A}}}
					\ee^{-\psi(x,x_{\Sites\setminus A}\lor p_A)}
					& \text{if $x\in[q_B]$ and $x_{\Sites\setminus A}\lor p_A\in\Omega$,} \\
				0	& \text{otherwise,}
			\end{cases}
	\end{align}
	for each two patterns $p\in\Lang_A(\Omega)$ and $q\in\Lang_B(\Omega)$
	with $B\Subset\Sites\setminus A$, where
	\begin{align}
		Z_{A|x_{\Sites\setminus A}} &\isdef
			\sum_{p'_A\in\Lang_{A|x_{\Sites\setminus A}}(\Omega)}
			\ee^{-\psi(x,x_{\Sites\setminus A}\lor p'_A)} \;.
	\end{align}
	This extends to a unique probability measure $K_A(x,\cdot)$ on $\Omega$.
	The function $K_A\colon\Omega\times\field{F}(\Omega)\to[0,1]$ is a proper kernel
	from $\field{F}_{\Sites\setminus A}(\Omega)$ to $\field{F}(\Omega)$.
	Clearly, $K_A(x,[x_A])>0$ for each $x\in\Omega$ and $A\Subset\Sites$.
	It remains to show that these kernels are consistent.
	
	To this end, take $A\subseteq B\Subset\Sites$.
	Then, for every $x\in\Omega$,
	\begin{align}
		K_B(x,[x_{B\setminus A}])K_A(x,[x_A]) &=
			\Bigg(
				\sum_{r_A\in\Lang_{A|x_{\Sites\setminus A}}(\Omega)}
				\frac{1}{Z_{B|x_{\Sites\setminus B}}}
					\ee^{-\psi(x,x_{\Sites\setminus A}\lor r_A)}
			\Bigg)\cdot
			\frac{1}{Z_{A|x_{\Sites\setminus A}}}
				\ee^{-\psi(x,x)} \\
		&=
			\frac{1}{Z_{B|x_{\Sites\setminus B}}}
			\Bigg(
			{\underbrace{
				\sum_{r_A\in\Lang_{A|x_{\Sites\setminus A}}(\Omega)}
				\frac{1}{Z_{A|x_{\Sites\setminus A}}}
					\ee^{-\psi(x,x_{\Sites\setminus A}\lor r_A)}
			}_{1}}
			\Bigg)
			\ee^{-\psi(x,x)} \\
		&=
			\frac{1}{Z_{B|x_{\Sites\setminus B}}}\ee^{-\psi(x,x)} \\
		&=
			K_B(x,[x_B]) \;,
	\end{align}
	which means $K_A$ and $K_B$ are consistent.
	We conclude that $K$ is a positive specification.
\end{proof}

\subsection{Background on the $\VS$-norm}
\label{apx:interactions:variation-summable}

\begin{proof}[Proof of Proposition~\ref{prop:interaction:variation-summable:single-site}]
	Let $A\Subset\ZZ^d$.
	By the uniform pivot property, there exists a finite set $B\supseteq A$ such that
	for every $(x,y)\in\relation{T}_A(\Omega)$, there is a sequence
	$x = z^{(0)} \to z^{(1)} \to  z^{(2)} \to \cdots \to  z^{(n)} = y$
	of admissible pivot moves at sites $s_1,s_2,\ldots,s_n\in B$,
	transforming $x$ to $y$.
	Clearly, by removing the repetitions if necessary,
	this sequence can be chosen such that the number of visits to each site in $B$
	is bounded by $\ell\isdef\abs{\Lang_B(\Omega)}$. %
	Thus, for every continuous observable $f\in\banach{C}(\Omega)$,
	\begin{align}
		\abs[\big]{f(y) - f(x)} &\leq
			\sum_{i = 1}^n \abs[\big]{f(z^{(i)}) - f(z^{(i-1)})} \leq
			\sum_{i = 1}^n \Var_{s_i}(f) \leq
			\ell\sum_{s\in B}\Var_s(f) \;.
	\end{align}
	Since this is true for every $(x,y)\in\relation{T}_A(\Omega)$,
	we find that $\Var_A(f)\leq\ell\sum_{s\in B}\Var_s(f)$.
	It follows that
	\begin{align}
		\sum_{\substack{C\Subset\Sites\\ C\cap A\neq\varnothing}}\Var_A(\Phi_C) &\leq
			\sum_{\substack{C\Subset\Sites\\ C\cap A\neq\varnothing}}
			\ell\sum_{s\in B}\Var_s(\Phi_C) \\
		&=
			\ell\sum_{s\in B}\sum_{\substack{C\Subset\Sites\\ C\cap A\neq\varnothing}}\Var_s(\Phi_C) \leq
			\ell\sum_{s\in B}\smash{\overbrace{\sum_{\substack{C\Subset\Sites\\ C\ni s}}\Var_s(\Phi_C)}^{<\infty}} \;,
	\end{align}
	which is finite.  Hence, $\Phi$ is variation-summable.
\end{proof}

\begin{proof}[Proof of Lemma~\ref{lem:Phi_bound_VS}]
	Let $A\Subset\ZZ^d$ be arbitrary.
	Define a graph $G_A(\Omega)$ as follows.  The vertices of $G_A(\Omega)$ are the patterns in $\Lang_A(\Omega)$.
	Two patterns $w,w' \in \Lang_A(\Omega)$ are connected by an edge in $G_A(\Omega)$
	if and only if there exists a sequence of configurations $x^{(0)}, x^{(1)},\ldots, x^{(N)}\in\Omega$ with $x^{(0)}_A = w$, $x^{(N)}_A = w'$ such that each $x^{(i-1)}\to x^{(i)}$ is a pivot move and precisely one of these moves is in $A$.
	By the pivot property of~$\Omega$, the equivalence classes of $\overset{\Omega}{\sim}$ in $\Lang_A(\Omega)$ are precisely the connected components of the graph $G_A(\Omega)$.
	Observe that if $w,w' \in \Lang_A(\Omega)$ are adjacent in $G_A(\Omega)$, then
	\begin{align}
		\abs[\big]{\Phi(w)-\Phi(w')} &\leq \Var_{k}(\Phi_A) = \Var_0(\Phi_{A-k})
	\end{align}
	for some $k \in A$.
	It follows by induction that for any $w,w' \in \Lang_A(\Omega)$,
	\begin{align}
		\abs[\big]{\Phi(w)-\Phi(w')} &\leq \max_{k \in A}\Var_0(\Phi_{A-k})\,d_{G_A(\Omega)}(w,w') \;,
	\end{align}
	where $d_{G_A(\Omega)}$ denotes the graph distance of $w$ and $w'$ in $G_A(\Omega)$.
	If $w,w' \in \Lang_A(\Omega)$ are in the same equivalence class, then clearly $d_{G_A(\Omega)}(w,w')<\abs{\Lang_A(\Omega)}$.	On the other hand,
	$\max_{k \in A}\Var_0(\Phi_{A-k})\leq\sum_{\substack{C\Subset\ZZ^d\\ C\ni 0}}\Var_0(\Phi_C)=\normVS{\Phi}$.
	The claim follows.
\end{proof}

\begin{proof}[Proof of Proposition~\ref{prop:interaction:variation-summable:seminorm}]
	Clearly, if for every $C \Subset \ZZ^d$ the function $\Phi_C$
	is constant on each asymptotic class of $\Omega$, then $\Var_0(\Phi_C) =0$ whenever $0 \in C$,
	and thus $\normVS{\Phi}=0$.
	Conversely, if  $\normVS{\Phi}=0$, then by Lemma \ref{lem:Phi_bound_VS},
	for every $C \Subset \ZZ^d$ the function $\Phi_C$
	is constant on each asymptotic class of $\Omega$.
\end{proof}

\subsection{Surjectivity of linear maps on Banach spaces}
\label{apx:surjectivity}

\begin{proof}[Proof of Proposition~\ref{prop:Banach_equivalences}]\ 
	\begin{description}[font=\rmfamily\mdseries] %
	\item[\ref{prop:Banach_equivalences:1}~$\implies$~\ref{prop:Banach_equivalences:2}]
		By the open mapping theorem,
		the image of any ball centered at the origin in $X$
		contains a ball centered at the origin. Now scale up.
		Then the image of some ball centered at the origin
		in $X$ contains the unit ball.

	\item[\ref{prop:Banach_equivalences:2}~$\implies$~\ref{prop:Banach_equivalences:1}]
		The image of the map is the
		union of images of balls centered at the origin.
		By linearity, these images are all scalar multiples of
		one another. So, if the image of some ball centered
		at the origin contains the unit ball, then each ball
		centered at  the origin is contained in the image of
		some ball and so the map is surjective.
	
	\item[\ref{prop:Banach_equivalences:2}~$\implies$~\ref{prop:Banach_equivalences:3}]
		Trivial.

	\item[\ref{prop:Banach_equivalences:3}~$\implies$~\ref{prop:Banach_equivalences:2}]
		Let $y$ be in the unit ball in $Y$.
		We show that $y$ has a pre-image in the ball of radius $2R$ in $X$.
		Namely, the pre-image will be of the form $x\isdef\sum_{i=1}^\infty u_i$,
		where $u_i\in\ball{X}{\nicefrac{R}{2^{i-1}}}$, and the image of the partial sum
		$x_n\isdef \sum_{i=1}^n u_i$ will approximate $y$ with accuracy $\nicefrac{1}{2^{n}}$.
		
		It follows from~\ref{prop:Banach_equivalences:3} that
		for all $\delta > 0$,
		$T\big(\ball{X}{\delta R}\big)$ is dense in  $\ball{Y}{\delta}$.
		Set $v_1\isdef y$.
		Choose $u_1 \in \ball{X}{R}$ such that $\norm{v_1 - T(u_1)}<\nicefrac{1}{2}$.
		Inductively, suppose that $u_1,u_2,\ldots,u_n\in X$ are such that
		$u_i\in\ball{X}{\nicefrac{R}{2^{i-1}}}$ and $\norm{y - T(\sum_{i=1}^n u_i)}<\nicefrac{1}{2^n}$.
		Set $v_{n+1}\isdef y - T(\sum_{i=1}^n u_i)$ and
		choose $u_{n+1}\in\ball{X}{\nicefrac{R}{2^n}}$ such that $\norm{v_{n+1} - T(u_{n+1})}< \nicefrac{1}{2^n}$.
		It follows that
		\begin{align}
			\norm[\Big]{y - T\Big(\sum_{i=1}^{n+1} u_i\Big)} &=
				\norm[\Big]{y - T\Big(\sum_{i=1}^n u_i\Big) - T(u_{n+1})} =
				\norm[\Big]{v_{n+1} - T(u_{n+1})} < \frac{1}{2^n} \;.
		\end{align}
		By construction, the sequence $x_n=\sum_{i=1}^n u_i$ is Cauchy
		and thus has a limit $x$ in $X$.
		Furthermore, $\norm{x_n}\leq\sum_{i=1}^n\norm{u_i}<\sum_{i=1}^n \nicefrac{R}{2^{i-1}}$.
		Thus, $\norm{x}<\sum_{i=1}^\infty \nicefrac{R}{2^{i-1}}=2R$.
		Lastly, since $T(x_n)\to y$ as $n\to\infty$ and $T$ is continuous,
		we have $T(x)=y$.	
		\qedhere		
	\end{description}
\end{proof}

\end{document}